\documentclass[letter,11pt,final]{article}
\usepackage[numbers,sort]{natbib}

\newif\iffull
\fullfalse
\newif\ifshort
\shorttrue
\usepackage{typearea}
\typearea{14}

\usepackage{authblk}
\usepackage{comment}
\usepackage{caption}
\usepackage{graphicx}
\usepackage{subcaption}
\usepackage[normalem]{ulem}
\usepackage{algcompatible}
\usepackage{algorithm}
\usepackage{multirow}
\usepackage{amssymb}
\usepackage{pifont}% http://ctan.org/pkg/pifont
\newcommand{\cmark}{\ding{51}}%
\newcommand{\xmark}{\ding{55}}%
\usepackage{amsmath}
\usepackage{amsthm}
\usepackage{enumerate}
\usepackage[shortlabels]{enumitem}
\usepackage{hyperref}
\usepackage[all]{hypcap}
\usepackage[referable]{threeparttablex}
\usepackage{footnote}
\usepackage{pgfplots}
\usepackage{wrapfig}
\usepackage{bigstrut}
\usepackage{multicol}
\usepackage{xspace}
\usepackage{nicefrac}
\usepackage[symbol]{footmisc}

\usepackage{thmtools,thm-restate}

\usepackage{cleveref}

\theoremstyle{plain}
\newtheorem{thm}{Theorem}[section]
\newtheorem{cor}[thm]{Corollary}

\newtheorem{fact}[thm]{Fact}
\newtheorem{lem}[thm]{Lemma}
\newtheorem{Def}[thm]{Definition}
\newtheorem{obs}[thm]{Observation}

\newcommand{\eps}{\varepsilon}%
\newcommand{\p}{\textsc{P}}%

\newcommand{\FF}{\textsc{FirstFit}\xspace}%
\newcommand{\BP}{\textsc{bin packing}\xspace}
\newcommand{\acr}{a.c.r\xspace}

\newcommand{\calA}{\mathcal{A}}
\newcommand{\calB}{\mathcal{B}}

\newcommand{\alphas}{\alpha^\star}
\newcommand{\cI}{\mathcal{I}}
\newcommand{\cB}{\mathcal{B}}
\newcommand{\cN}{\mathcal{N}}
\newcommand{\calI}{\mathcal{I}}
\newcommand{\taus}{\tau^\star}
\newcommand{\sizec}{{\tt size}}
\newcommand{\ts}{\textstyle}

\newcommand{\bound}{\approx 1.387}   %1.3871
\newcommand{\binpackingub}{1.589}    %1.58889
\newcommand{\binpackinglb}{1.540}    %1.54037
\newcommand{\binpackingnewub}{1.578} %1.57829

\usepackage{xcolor}
\definecolor{darkgreen}{rgb}{0,0.5,0}

\usepackage[colorinlistoftodos,textsize=tiny,textwidth=2cm]{todonotes}
\newcommand{\agnote}[1]{\todo[color=blue!25!white]{AG: #1}\xspace}

\newcommand{\dwnote}[1]{\todo[color=red!25!white]{DW: #1}\xspace}

\newcommand{\squishlist}{
 \begin{list}{$\bullet$}
  { \setlength{\itemsep}{0pt}
     \setlength{\parsep}{3pt}
     \setlength{\topsep}{3pt}
     \setlength{\partopsep}{0pt}
     \setlength{\leftmargin}{1.5em}
     \setlength{\labelwidth}{1em}
     \setlength{\labelsep}{0.5em} } }

\newcommand{\squishend}{
  \end{list}  }

\newcommand\fullshort[2]{%
\iffull%
{#1}
\fi%
\ifshort%
{#2}
\fi}

\newcommand{\mysubsubsection}[1]{\medskip\noindent\textbf{\textsf{#1:}}}
\begin{document}

\pagenumbering{gobble} 

\title{Fully-Dynamic Bin Packing with Limited Repacking\footnote{This work was done in part while the authors were visiting the Simons Institute for the Theory of Computing.} \footnote{An extended abstract of this work, merged with \cite{feldkord2017tight}, will appear in ICALP 2018.}}
\date{}

\author[1]{Anupam Gupta}
\author[1]{Guru Guruganesh}
\author[2]{Amit Kumar}
\author[1]{David Wajc}
\affil[1]{Carnegie Mellon University} 
\affil[2]{IIT Delhi} 

%\footnotetext[1]{This work was done in part while the authors were visiting the Simons Institute for the Theory of Computing.}

\maketitle

\begin{quote}
	\emph{``To improve is to change; to be perfect is to change often.''}
	\vspace{-0.5cm}
	\begin{flushright}	--  Winston Churchill.
	\end{flushright}
\end{quote}

\begin{abstract}
We study the classic \BP problem in
a fully-dynamic setting, where new items can arrive and old items may depart.
We want algorithms with low asymptotic competitive
ratio \emph{while repacking items sparingly} between updates. Formally, each item
$i$ has a \emph{movement cost} $c_i\geq 0$, and we want to use $\alpha
\cdot OPT$ bins and incur a movement cost 
$\gamma\cdot c_i$, either in the worst case, or in an amortized sense, for
$\alpha, \gamma$ as small as possible. We call $\gamma$ the \emph{recourse}
of the algorithm.
This is motivated by  cloud storage applications, 
where fully-dynamic \BP models the problem of data
backup to minimize the number of disks used, as well as communication
incurred in moving file backups between disks.
Since the set of files changes over time, we could  
recompute a solution periodically
from scratch, but this would give a high number of disk rewrites, incurring
a high energy cost and possible wear and tear of the disks.  In this work, we
present optimal 
tradeoffs between number of bins used and number of items repacked, as well as
natural extensions of the latter measure.
\end{abstract}

\newpage
\pagenumbering{arabic} 

\renewcommand*{\thefootnote}{\arabic{footnote}}
\section{Introduction}

Consider the problem of data backup on the cloud, where multiple users' files are
stored on disks (for simplicity, of equal size).  This is
%a classic example of
modeled by the \BP problem, where the items are 
files and bins are disks, and we want to pack the items in a
minimum number of bins. However, for the backup application, files are created and
deleted over time.  The storage provider
% , interested in minimizing operation costs,
wants to use a small number of disks, and also keep the
communication costs incurred by file transfers to a minimum.
% As in
% dynamic algorithms, the goal is to change the packing by as little as
% possible when new items arrive while still using a near-optimal number
% of bins.
% Clearly, packing items in a near-minimum number of bins can be done if
% repacking is allowed (indeed, repacking is necessary to obtain \emph{any}
% guarantees); however, such an approach may be prohibitively expensive due to
% the operational and communication costs. 
Specifically, we want bounded ``recourse'',
i.e., items should be moved sparingly while attempting to minimize the number
of bins used in the packing.  These objectives are at odds with each other, and
the natural question is to give optimal tradeoffs between them.

Formally, an instance $\mathcal{I}$ of the \BP problem consists of a list of $n$
items of sizes $s_1,s_2,\dots,s_n \in [0,1]$. A bin
packing algorithm seeks to pack the items of $\mathcal{I}$ into a small
number of unit-sized bins; let $OPT(\mathcal{I})$ be the minimum number
of unit-sized bins needed to pack all of $\mathcal{I}$'s items. This NP-hard
problem has been studied since the 1950s, with hundreds of
papers addressing its many variations; see e.g.,
\cite[Chapter~2]{hochbaum1996approximation} and %for an early survey, and
\cite{coffman2013bin} for surveys. Much of this work 
(e.g.~\cite{johnson1974worst,yao1980new,lee1985simple,ramanan1989line,
richey1991improved,woeginger1993improved,van1995lower,seiden2002online,
heydrich2016beating,balogh2012new})
starting with~\cite{ullman1971performance} studies the \emph{online} setting, where items
arrive sequentially and are packed into bins immediately and irrevocably.
% Many of the simple algorithms like \FF work naturally in this setting, and
% other ingenious algorithms have been devised too; see, e.g.,~\cite{blah} for
% references. 
% The online problem is strictly harder than
% the offline version, due to lack of information about the future:
While the offline problem is approximable to within an \emph{additive}
term of $O(\log OPT)$ \cite{hoberg2017logarithmic}, in the online setting there
is a $\binpackinglb$-\emph{multiplicative gap} between the algorithm and $OPT$
in the worst case, even as $OPT \to \infty$~\cite{balogh2012new}.
%The above papers have narrowed the range of the \acr for online \BP,
%with the current best lower and upper bounds being
%$1.5403$~\cite{balogh2012new} and $1.5815$~\cite{heydrich2016beating},
%respectively.
%Given the wide applicability of the online problem, researchers have
%considered the problem where a ``small'' number of repackings is
%allowed. I.e., how well can we perform if we allow items to be moved
%between bins when new items arrive and old ones depart? Clearly, some
%repacking is necessary; to make this question non-trivial, we demand
%bounded ``recourse'', i.e., items should be moved sparingly.
Given the wide applicability of the online problem, researchers have
studied the problem where a small amount of repacking is allowed upon
item additions and deletions. Our work focuses on the bounded recourse
setting, and we give \emph{optimal tradeoffs} between the number of bins
used and amount of repacking required in the fully dynamic setting for
several settings.

A \emph{fully-dynamic \BP algorithm} $\mathcal{A}$, given a sequence of
item insertions and deletions, maintains at every time $t$ a feasible
solution to the \BP instance $\mathcal{I}_t$ given by items inserted and
not yet deleted until time $t$. Every item $i$ has a size $s_i$ and a
\emph{movement cost} $c_i$, which $\mathcal{A}$ pays every time
$\mathcal{A}$ moves item $i$ between bins. In our application, the
movement cost $c_i$ may be proportional to the size of the file (which
is the communication cost if the files are all large), or may be a
constant (if the files are small, and the overhead is entirely in setting up the
connections), or may depend on the file in more complicated ways.
%The
%measures of efficiency of a fully-dynamic \BP algorithm are
%as follows.

\vspace{-0.1cm} 
\begin{Def} A fully-dynamic algorithm $\mathcal{A}$ has (i) an
  \emph{asymptotic competitive ratio} $\alpha$, (ii)
  \emph{additive term} $\beta$ and (iii) \emph{recourse} $\gamma$, if at each
  time $t$ it packs the instance $\mathcal{I}_t$ in at most
  $\alpha\cdot OPT(\mathcal{I}_t)+\beta$ bins, 
%  with $\beta$ independent of
%  $OPT(\mathcal{I}_t)$, 
  while paying 
  at most $\gamma\cdot \sum_{i=1}^t c_i$ movement cost until time $t$. If at each time $t$ algorithm
  $\mathcal{A}$ incurs at most $\gamma\cdot c_t$ movement cost, for
  $c_t$ the cost of the item updated at time $t$, we say
  $\mathcal{A}$ has \emph{worst case} recourse $\gamma$,
  otherwise we say it has \emph{amortized} recourse $\gamma$. 
% \agnote{The
%     worst-case notion does not make much sense in the general costs
%     case, I feel, but maybe that's OK.}\dwnote{In response to above comment: Why? When you add/remove an item which has inherent ``bother value'' of $x$, you don't want to be bothered much more than $x$ :-) }
\end{Def}
%\vspace{-0.2cm} 

%\agnote{Why half here?}\dwnote{each item can be added + removed, so its $c_i$ can be counted at most twice}
Any algorithm must pay $\frac{1}{2}\sum_{i=1}^t c_i$ movement cost
just to insert the items, so an algorithm with
$\gamma$ recourse spends at most $2\gamma$ times more movement cost than the
absolute minimum.
% In a sense, algorithms with
% bounded recourse can be seen as bicriteria approximation algorithms. 
The goal is to design algorithms which simultaneously have low asymptotic
competitive ratio (\acr), additive term, and recourse.
 There is a natural
 tension between the \acr and recourse, so we want to find the optimal
 \acr-to-recourse trade-offs.

\medskip\noindent \textbf{Related Work.} Gambosi et
al.~\cite{gambosi1990new,gambosi2000algorithms} were the first to study dynamic
\BP, and gave a $4/3$-\acr algorithm for the insertion-only setting
which moves each item $O(1)$ times. This gives amortized $O(1)$
recourse for general movement costs. For unit movement costs, Ivkovi\'c
and Lloyd~\cite{ivkovic1996fundamental} and Balogh et
al.~\cite{balogh2008lower,balogh2014line} gave lower bounds of $4/3$ and
$1.3871$ on the \acr of algorithms with $O(1)$ recourse; Balogh et
al.~\cite{balogh2014line} gave an algorithm with \acr
$3/2+\eps$ and $O(\eps^{-1})$ movements per update in the insertion-only setting.
% Moreover, they gave a dynamic algorithm with \acr of $4/3$ using an
% amortized constant number of \emph{shifting moves} (move of a single
% large item or a number of small items grouped together) per insertion.
% Ivkovi\'c and Lloyd \cite{ivkovic1998fully,ivkovic1997partially}, gave a
% $5/4$-competitive fully-dynamic algorithm (i.e., allowing for both
% insertions \emph{and deletions}) algorithm using $O(\log n)$
% \emph{shifting moves} (move of a single large item or a number of small
% items grouped together) per update, and a $(1+\eps)$-asymptotically
% competitive algorithm using $O(\log n)$ such moves.\dwnote{incomparable
%   with our results.. Perhaps remove?}
% or $O(\log^2n)$ moves with time polynomial in $\eps^{-1}$.  In our
% terminology these results correspond to low recourse subject to unit
% movement cost (for the former result) and size movement cost (for
% the latter results).
Following \citet{sanders2009online}, who studied
makespan minimization with recourse, 
Epstein et al.~\cite{epstein09robust}
re-introduced the dynamic \BP problem with the 
movement costs $c_i$ equaling the size $s_i$ (the
\emph{size cost} setting, where worst-case recourse is also called  \emph{migration factor}).
% recourse under size movement costs. Sanders et al.~observed the
% importance of this result in the context of sensitivity analysis. While
% a single item (or job, for scheduling) may clearly change an optimal
% solution by a super-constant amount (unless P=NP), only a small number
% of changes need to be made in order to preserve near optimality. 
% Epstein
% et al.~\cite{epstein09robust} 
Epstein et al.~gave a solution with \acr $(1+\eps)$ and bounded
(exponential in $\epsilon^{-1}$) recourse for the insertion-only setting. Jansen and Klein
\cite{jansen13robust} improved the recourse to $poly(\eps^{-1})$ for
insertion-only algorithms, and Berndt et al.~\cite{berndt15fully} gave
the same recourse for fully-dynamic algorithms. They also showed that
any $(1+\eps)$ \acr algorithm must have \emph{worst-case recourse}
$\Omega(\eps^{-1})$. While these give nearly-tight results for ``size costs'' $c_i = s_i$, the unit cost ($c_i = 1$) and
general cost cases were not so well understood prior to this
work. 
\vspace{-0.1cm}

\subsection{Our Results}\label{sec:results}

We give (almost) tight characterizations for the recourse-to-asymptotic
competitive ratio trade-off for fully-dynamic \BP under (a) unit movement
costs, (b) general movement costs and (c) size movement costs. Our results are
summarized in the following theorems. (See \S\ref{sec:tabula} for a tabular listing of 
our results contrasted with the best previous results.) 
In the context of sensitivity analysis (see \cite{sanders2009online}), our bounds provide tight 
characterization of the \emph{average} change between 
\emph{approximately}-optimal solutions of slightly modified instances.

%Quote from Peter Sanders, Naveen Sivadasan, Martin Skutella:
%"Sensitivity analysis. Given an optimum solution to an instance of an optimization problem and a slightly
%modified instance, can the given solution be turned into an optimum solution for the modified instance without
%changing the solution too much? This is the compelling question in sensitivity analysis."

\mysubsubsection{Unit Costs}
Consider the most natural movement cost: \emph{unit costs}, where
$c_i = 1$ for all items $i$. Here we give tight upper and lower
bounds. Let $\alpha=1- \frac{1}{W_{-1}(-2/e^3)+1} \approx 1.3871$ (here
$W_{-1}$ is the lower real branch of
the %\href{https://en.wikipedia.org/wiki/Lambert\_W\_function}{Lambert $W$-function})
Lambert $W$-function \cite{corless1996lambertw}). %Balogh et al.~\citet{balogh2008lower} 
\citet{balogh2008lower} showed $\alpha$ is a lower bound on the \acr
% of all fully-dynamic \BP algorithms
with constant recourse. We present
an alternative and simpler proof of this lower bound, also giving tighter bounds: doing better than $\alpha$
requires either \emph{polynomial} additive term or recourse. Moreover,
we give a matching algorithm proving $\alpha$ is tight for this
problem.\footnote{Independently of this work, \citet{feldkord2017tight} provided a different optimal algorithm for unit movement costs. This then inspired us to provide our worst-case recourse algorithm under movement costs using our framework.}

\begin{thm}[Unit Costs Tight Bounds]
  \label{thm:main-unit}
  For any $\eps>0$, there exists a fully-dynamic \BP algorithm with \acr $(\alpha+\eps)$, additive term
  $O(\eps^{-2})$ and amortized recourse $O(\eps^{-2})$, or additive term $poly(\eps^{-1})$ and worst case recourse $\tilde{O}(\eps^{-4})$ under unit
  movement costs.
  Conversely, any algorithm with \acr
  $(\alpha-\eps)$ has additive term and 
  amortized recourse whose product is $\Omega(\eps^{4}\cdot n)$ under unit movement costs.
\end{thm}

%These complementary results give us a sharp threshold 
%on the \acr of any algorithm for the
%unit costs model. This is the technical heart of the paper: we show both
%upper and lower bounds using linear programming techniques. We give a
%linear program that completely captures the performance of the
%algorithm. Indeed, we first use this LP as a \emph{gap-revealing} LP, to
%prove that a certain family of instances give an \acr of at least
%$\approx \alpha$. Then we use it as a \emph{factor-revealing} LP to show
%that our algorithm achieves an \acr at most $\approx \alpha$. 

% by exhibiting a suitable dual, we show
% its value is at least $\alpha - O(\eps)$, hence getting a lower
% bound. But now we take the same LP, and use it as a ``factor-revealing''
% linear program---we first show that its optimal value is indeed $\alpha
% + O(\eps)$. We then show how to maintain a solution with competitive
% ratio bounded by this value, with small recourse. The main ideas are
% \alert{blah blah}.

% As we formally argue later, the lower bound should be read as a sharp
% transition, as any dependence on $n$ in the additive term or recourse
% cost can be made arbitrarily bad by adding any arbitrarily large $N$
% items of size $1/N$ without modifying the instance in any meaningful
% way.

\mysubsubsection{General Costs} 
Next, we consider the most general problem, with % fully-dynamic bin packing in utmost generality, with
arbitrary movement costs.
%The results of \cite{epstein09robust, jansen13robust, berndt15fully}
%show that in the size cost model a little recourse can circumvent the
%negative effects of online arrivals, and allow for arbitrarily good \acr. 
\Cref{thm:main-unit} showed that % a little recourse circumvents
% the negative effects of online arrivals 
in the unit cost model, we can
get a better \acr than for online \BP without repacking, whose optimal
\acr is at least $\binpackinglb$ (\cite{balogh2012new}).  % Can repacking
% give such an improvement even for general costs?
Alas, % our next result
% shows that even allowing for repacking, 
the fully-dynamic \BP problem with the general costs is
no easier than the arrival-only online problem (with no
repacking).% \dwnote{Removed comparison to $(1+\epsilon)$-\acr, because of
  % previous section.}
%In all these results, we
%use the term \emph{online \BP} to refer to the classical insertion-only
%model without any repacking, and \emph{fully-dynamic \BP (with
%	recourse)} to refer to the model with both insertions and deletions
%where we can repack items.\dwnote{explain order}

\begin{thm}[Fully Dynamic as Hard as Online]
	\label{thm:main-lb-gen}
	Any fully-dynamic \BP algorithm with bounded recourse under general movement costs
	has \acr at least as high as that of any
	online \BP algorithm. Given current bounds (\cite{balogh2012new}), this is at least
	$\binpackinglb$.
\end{thm}

%Now we consider the flip side:
Given this result, is it conceivable that the fully-dynamic model is
\emph{harder} than the arrival-only online model, even allowing for
recourse? We show this is likely not the case, as we can almost match
the \acr of the current-best algorithm for online \BP.

\begin{thm}[Fully Dynamic Nearly as Easy as Online]\label{thm:main-ub-gen}
	Any algorithm in the Super Harmonic family of algorithms can be
	implemented in the fully-dynamic setting with constant recourse under
	general movement costs.
%	, yielding the same \acr. 
	This implies an
	algorithm with \acr $\binpackingub$ using~\cite{seiden2002online}.
\end{thm}

The current best online \BP algorithm~\cite{balogh2017new} is not from the Super Harmonic family
but is closely related to them. It has an \acr of
$\binpackingnewub$, so our results for fully-dynamic \BP are within a hair's width of the best bounds known for 
online \BP. It remains an open question as to whether our techniques can be extended to achieve the improved a.c.r bounds while maintaining constant recourse. 
While we are not able to give a black-box
reduction from fully-dynamic algorithms to online algorithms, we
conjecture that such a black-box reduction exists and that these problems' \acr are equal.

%As an aside, if all items have size $\Omega(1)$, the recourse bound from
%Theorem~\ref{thm:main-ub-gen} becomes worst-case without harming the
%\acr. Alternatively, we can obtain the same \acr with worst-case
%recourse if we allow an $O(\log n)$ additive term. 

\mysubsubsection{Size Costs}
Finally, we give an extension of the already strong results known
for the size cost model (where $c_i=s_i$ for every item $i$). 
We show that the lower bound known in the
worst-case recourse model extends to the amortized model as well, for
which it is easily shown to be tight.

\begin{thm}[Size Costs Tight Bounds]
	\label{thm:main-size}
	For any $\eps>0$, there exists a $(1+\eps)$-\acr
	algorithm with $O(\eps^{-2})$ additive term and $O(\eps^{-1})$ 
	amortized recourse under size costs.   
	Conversely, for
	infinitely many $\eps>0$, any $(1+\eps)$-\acr algorithm with $o(n)$ additive term requires $\Omega(\eps^{-1})$ amortized recourse under size costs. 
\end{thm}

The hardness result above was previously only known for
\emph{worst-case} recourse~\cite{berndt15fully}; this previous
lower bound consists of a hard instance which effectively disallowed any
recourse, while lower bounds against amortized recourse can of course
not do so. 

\subsection{Techniques and Approaches}
\label{sec:techniques}

\mysubsubsection{Unit Costs}
For unit costs, our lower bound is based on a natural  
instance consisting of small items, to which large items of various
sizes (greater than $\frac12$) are repeatedly
added/removed~\cite{balogh2008lower}. The
near-optimal solutions oscillate between %instances where near-optimal solutions
either having most bins containing 
both large and small items, %and instances where near-optimal solutions
% have
or most bins containing only
small items. Since small items can be moved rarely, this
effectively makes them static, giving us hard instances.
%As small items are arbitrarily small, an algorithm with bounded recourse effectively cannot afford to move a 
%non-negligible number of ``bins' worth'' of such items, and the algorithm must therefore pack the small items in a way to 
%``prepare for all eventualities'' determined by large items to be added. 
%Trading off these different constraints yield a family of lower bounds, where 
We can optimize for the lower bound arising thus
%instances
via a \emph{gap-revealing} LP, similar to~\cite{balogh2008lower}.
However, rather than bound this LP precisely, we exhibit a near-optimal
dual solution showing a lower
bound of $\alpha - \eps$.  

For the upper bound, the same LP
now changes from a gap-revealing LP to a
\emph{factor-revealing} LP. The LP solution shows how the
small items should be packed in order to ``prepare'' for arrival of
large items, to ensure $(\alpha+\epsilon)$-competitiveness.
An important building block of our algorithms is the ability to deal with
(sub)instances made up solely of small items. In particular, we give
fully-dynamic \BP algorithms with \acr $(1+\eps)$ and only $O(\eps^{-2})$
\emph{worst case} recourse if all items have size $O(\eps)$. 
The ideas we develop are versatile, and are used, e.g., to handle the
small items for general costs -- we discuss them in more detail below.  We then extend these ideas
to allow for (approximately) packing items according to a ``target
curve'', which in our case is the solution to the above LP.
At this point the LP makes a re-appearance, allowing us to analyze a 
simple packing of large items ``on top'' of the curve (which we can either 
maintain dynamically using the algorithm of \citet{berndt15fully} or recompute 
periodically in linear time); in particular, using this LP we are able to 
prove this packing of large items on top of small items has an optimal \acr of $\alpha+O(\eps)$.

%The fact that we can break our instance into small and large items and solve them separately stems from the LP solution,  which in a sense yields a lower bound with respect to the volume bound, and hence we can compare our solution with this stronger lower bound. 

\mysubsubsection{General Costs}  To transform any online lower
bound to a lower bound for general costs even with recourse, we give the
same sequence of items but with movement costs that drop sharply. This
prohibits worst-case recourse, but since we allow amortized recourse,
the large costs $c_i$ for early items can still be used to pay for
movement of the later items (with smaller $c_i$ values). Hence, if the
online algorithm moves an element $e$ to some bin $j$, we delete all the
elements after $e$ and proceed with the lower bound based on assigning
$e \mapsto j$. Now a careful analysis shows that for some sub-sequence
of items the dynamic algorithm effectively performs no changes, making
it a recourse-less online algorithm -- yielding the claimed bound.

%This shows that the fully-dynamic problem is at least as hard as the online problem. 
%It is conceivable that this problem is strictly harder due to the presence
%of deletions, despite the presence of recourse. 

For the upper bound, we show how to emulate any algorithm in the
\emph{Super Harmonic} class of algorithms \cite{seiden2002online} even
in the fully-dynamic case.  
%A super harmonic (SH) algorithm partitions
%large items into two colors (red and blue), and packs similarly-sized
%items in groups, with at most two groups per bin, one of each
%color. Finally, an SH algorithm packs small items using \FF into bins
%which are $1-\epsilon$ full.
%and 
%maintains a set of invariants on large items by partitioning them into two classes (red and blue). 
We observe that the analysis for SH essentially relies on maintaining a
stable matching in an appropriate compatibility graph.
% , and crucially use
% this for our simulation. % We show that these invariants can be
% maintained, in a fully-dynamic setting with $O(1)$-recourse.
Now our
algorithm for large items uses a subroutine similar to the Gale-Shapley
algorithm.
Our simulation also requires a solution for packing similarly-sized
items.  The idea here is to sort items by their cost (after rounding
costs to powers of two), such that any insertion or deletion of items of
this size can be ``fixed'' by moving at most one item of each smaller
cost, yielding $O(1)$ worst-case recourse.

The final ingredient of our algorithm is packing of small items into
$(1-\epsilon)$-full bins.  Unlike the online algorithm which can just
use \FF, we have to extend the ideas of \citet{berndt15fully} for the
size-cost case.  Namely we group bins into buckets of
$\Theta(1/\epsilon)$ bins, such that all but one bin in a bucket are
$1-O(\epsilon)$ full, guaranteeing these bins are $1-O(\epsilon)$ full
on average. While for the size-cost case, bucketing bins and sorting the
small items by size readily yields an $O(\epsilon^{-1})$ recourse bound,
this is more intricate for general case where  size and cost are not
commensurate.
Indeed, we also maintain the small items in
sorted order according to their size-to-movement-cost ratio (i.e., their
\emph{Smith ratio}), and only move items to/from a bin once it has
$\Omega(\epsilon)$'s worth of volume removed/inserted (keeping track of
erased, or ``ghost'' items). This gives an amortized $O(\epsilon^{-2})$
recourse cost for the small items.
% , completing our fully-dynamic
% simulation of super harmonic algorithms for the general costs case.

\mysubsubsection{Size Costs}  % For the size-costs model, an
% optimal algorithm is straightforward: we lazily compute a near-optimal
% solution,
% and lazily deal with updates until the overall size changes by some
% $1\pm \epsilon$ factor. 
The technical challenge for size costs is in
proving the lower bound with \emph{amortized} recourse that matches the
easy upper bound.
% the asymptotic competitive ratio to amortized recourse
% trade-off is inherent.
Our proof uses item sizes which are roughly the reciprocals of the
Sylvester sequence \cite{sylvester1880point}. While this sequence has
been used often in the context of \BP algorithms, our sharper lower
bound explicitly relies on its divisibility properties (which have not
been used in similar contexts, to the best of our knowledge).
%properties we use have not been used in the past in similar contexts.
We show that for these sizes, any algorithm
$\mathcal{A}$ with \acr $(1+\epsilon)$ must, on an instance containing $N$
items of each size, have $\Omega(N)$ many bins
with one item of each size. In contrast, on an instance containing $N$
items of each size \emph{except the smallest size} $\eps$, 
algorithm $\mathcal{A}$ must have some (smaller) $O(N)$ many bins with
more than one distinct size in them. Repeatedly adding
and removing the $N$ items of size $\eps$, the algorithm must suffer a 
high amortized recourse.

\medskip
We give a more in-depth description of our algorithms and lower bounds for unit, general and size costs, highlighting some salient points of their proofs in \S\ref{sec:unit},\S\ref{sec:general} and \S\ref{sec:migration}, respectively. Full proofs are deferred to \S\ref{sec:unit-appendix},\S\ref{sec:general-appendix} and \S\ref{sec:migration-appendix}, respectively.

\section{Unit Movement Costs}\label{sec:unit}

We consider the natural unit movement costs model. First, we show
that an \acr better than $\alpha = 1-
\frac{1}{W_{-1}(-2/e^3)+1} \bound$ implies either
polynomial additive term or recourse. 
Next, we give tight $(\alpha+\epsilon)$-a.c.r algorithms, with both 
additive term and recourse polynomial in $\epsilon^{-1}$. Our key idea %ingredient in our algorithm
is to use an LP that acts both as a \emph{gap-revealing LP} for a lower bound, and also as a \emph{factor-revealing LP} (as well as an inspiration) for our algorithm. 
%In our proofs we use the fact that $1-1/\alpha$ is a solution to $3+ \ln(1/2) = \ln(x) + 1/x$.

\subsection{Impossibility results}

\iffull
We begin by observing that any algorithm with absolute
competitive ratio better than $3/2$ cannot have recourse cost $o(n)$.
Consider two instances, both having $2n-4$ items of size $1/n$ and one instance also having $4$
items of size $1/2+1/2n$. Alternating between these two instance by inserting/deleting the latter 4 items,
we get the following.
\begin{obs}\label{thm:3-halfs-lb}
	For all $\eps>0$, any $\left(\frac{3}{2}-\eps \right)$-competitive online
	bin packing algorithm must have $\Omega(n)$ recourse, where $n$ is the
	maximum number of items encountered by the algorithm.
\end{obs}
\fi

\ifshort
Alternating between two instances, where both have $2n-4$ items of size
$1/n$ and one instance also has $4$ items of size $1/2+1/2n$, we get
that any $\left(\frac{3}{2}-\eps \right)$-competitive online bin packing
algorithm must have $\Omega(n)$ recourse. 
\fi
However, this only rules out
algorithms with a \emph{zero} additive term.
\citet{balogh2008lower} showed that any $O(1)$-recourse algorithm
(under unit movement costs) with $o(n)$ additive term must have \acr at
least $\alpha \bound$.  We strengthen both this impossibility
result by showing the need for a large additive term or recourse to achieve any
\acr below $\alpha$.  Specifically,
%the following theorem implies that
%arbitrarily small
arbitrarily small polynomial additive terms imply near-linear recourse.
%that is arbitrarily %Moreover,
Our proof is shorter and
simpler than in~\cite{balogh2008lower}.
% that of Balogh et al. 
As an added benefit, the LP we use to bound the competitive ratio
will inspire our algorithm in 
the next section.

\ifshort
\begin{restatable}[Unit Costs: Lower Bound]{thm}{unitFinalUb}
  \label{thm:final-lb}
  For any $\eps>0$ and $\frac{1}{2}> \delta>0$, for any algorithm
  $\mathcal{A}$ with \acr $\left(\alpha-\eps\right)$ with additive term
  $o(\eps\cdot n^{\delta})$, there exists an dynamic bin
  packing input with $n$ items on which $\mathcal{A}$ uses recourse at least
  $\Omega(\eps^2 \cdot n^{1-\delta})$ under unit movement cost.
\end{restatable}
\fi
\iffull
\unitFinalUb*
\fi

\noindent\textbf{The Instances:} The set of instances is the a natural one,
and was also considered by \cite{balogh2008lower}. Let $\nicefrac{1}{2}>
\delta>0$, $c=\nicefrac1\delta-1$, and let $B\geq \nicefrac{1}{\eps}$ be
a large integer. Our hard instances consist of \emph{small} items of
size $1/B^c$, and \emph{large} items of size $\ell$ for all sizes $\ell\in
\mathcal{S}_{\eps} \triangleq \{\ell=1/2+i\cdot \eps \mid i\in
\mathbb{N}_{>0},\, \ell \leq 1/\alpha\}$. Specifically, input $\mathcal{I}_s$ consists of $\lfloor B^{c+1}
\rfloor$ small items, and for each $\ell\in \mathcal{S}_{\eps}$, the input
$\mathcal{I}_{\ell}$ consists of $\lfloor B^{c+1} \rfloor$ small items
followed by $\lfloor \frac{B}{1-\ell} \rfloor$ size-$\ell$ items.
The optimal bin packings for $\mathcal{I}_s$ and $\mathcal{I}_{\ell}$ require
precisely $OPT(\mathcal{I}_s)=B$ and $OPT(\mathcal{I}_{\ell})=\lceil
\frac{B}{1-\ell}\rceil$ bins respectively. Consider any fully-dynamic bin
packing algorithm $\mathcal{A}$ with limited recourse and bounded
additive term. When faced with input $\mathcal{I}_s$, algorithm $\mathcal{A}$ needs to
distribute the small items in the available bins almost uniformly. And
if this input is extended to $\mathcal{I}_{\ell}$ for some $\ell \in
\mathcal{S}_{\eps}$, algorithm $\mathcal{A}$ needs to move many small items to
accommodate these large items (or else use many new bins). Since $\calA$
does not know the value of $t$ beforehand, it cannot ``prepare''
simultaneously for all large sizes $\ell\in \mathcal{S}_{\eps}$.

\begin{wrapfigure}{r}{0.425\textwidth}
		\vspace{-.3cm}	

	\centering
	\capstart
	
	\includegraphics[width=2.in]{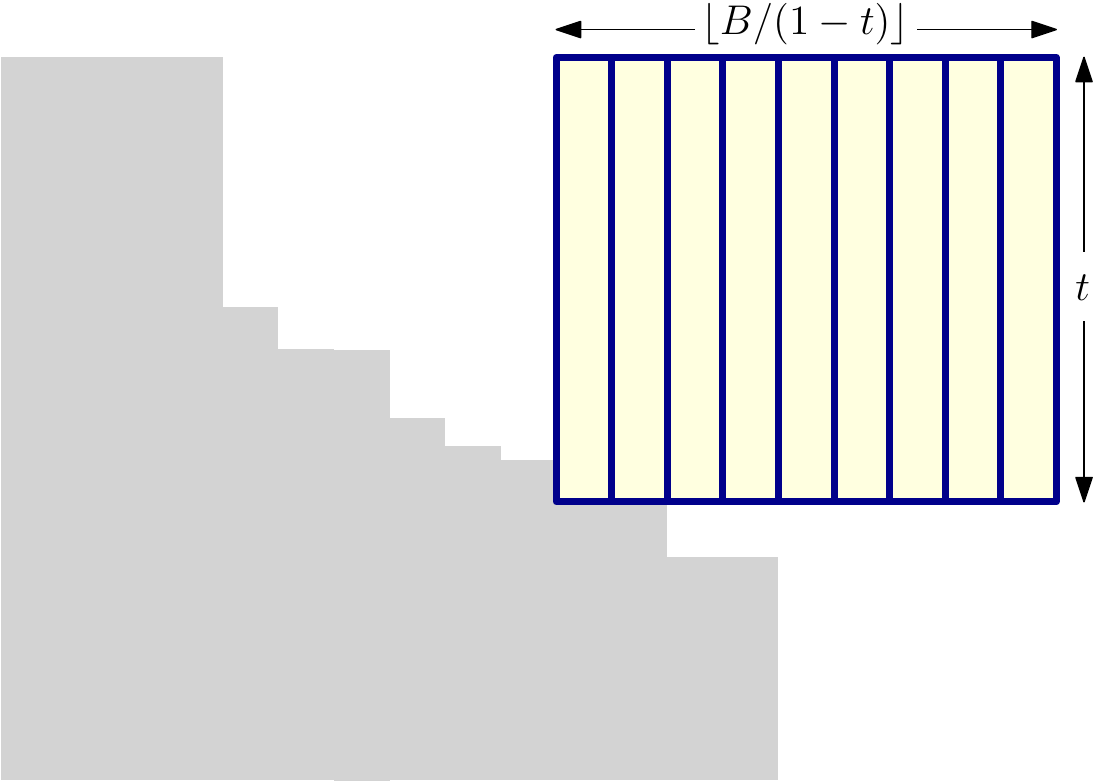}
	\caption{\small A packing of instance $\mathcal{I}_{\ell}$. The $\ell$-sized items (in yellow) are packed ``on top'' of the small items (in grey).}
	\label{fig:lowerbound}
	\vspace{-.35cm}
	%  \end{center}
	%\end{figure}
\end{wrapfigure}

As a warm-up we show that the linear program~\eqref{lpeps} below gives a lower bound on the \emph{absolute}
competitive ratio $\alpha_\eps$ of any algorithm $\mathcal{A}$ with no
recourse. Indeed, instantiate the variables as follows. On input
$\mathcal{I}_s$, let $N_x$ be the number of bins in which 
$\mathcal{A}$ keeps free space in the range $ [x,x+\eps)$ for each $x
\in \{0\}\cup \mathcal{S}_{\eps}$. Hence the total volume packed is at
most $N_0+\sum_{x\in \mathcal{S}_{\eps}} (1-x)\cdot N_x$.  This must
be at least $Vol(\mathcal{I}_s) %= \lfloor B^{c+1} \rfloor\cdot 1/B^c
\geq B-1/B^c$, implying constraint~\eqref{Vol-eps}.  Moreover, as
$OPT(\mathcal{I}_s)=B$, the $\alpha_\eps$-competitiveness of
$\mathcal{A}$ implies constraint~\eqref{small-eps}. Now if instance
$\mathcal{I}_s$ is extended to $\mathcal{I}_{\ell}$, since 
$\mathcal{A}$ moves no items, these $\ell$-sized items are placed either in
bins counted by $N_x$ for $x\geq \ell$ or in new bins. Since
$\mathcal{A}$ is $\alpha_\eps$-competitive and $OPT(\mathcal{I}_{\ell})\leq
\lceil \frac{B}{1-\ell}\rceil$ we get constraint~\eqref{CR-t}. Hence the
optimal value of~\eqref{lpeps} is a valid lower bound on the competitive
ratio~$\alpha_\eps$.

\vspace{-0.6cm}

%\begin{figure}[t]
%\begin{mdframed}
%  \vspace{-0.6cm}		
  \begin{alignat}{5}
%  \label{fig:lp1}
    \text{minimize }  & \alpha_{\eps} & \tag{LP$_{\eps}$}  \label{lpeps}\\
    \text{s.t. } & \ts N_0 + \sum_{x \in \mathcal{S}_{\eps}} (1-x)\cdot N_x &&\geq B-1/B^c & \tag{Vol$_{\eps}$} \label{Vol-eps} \\
    & \ts N_0 + \sum_{ x \in \mathcal{S}_{\eps}} N_x && \leq \alpha_{\eps} \cdot B  \tag{small$_{\eps}$} \label{small-eps}\\
    & \ts N_0 + \sum_{x \in \mathcal{S}_{\eps}, x \leq \ell-\eps} N_x
    + \Big \lfloor \frac{B}{1-\ell} \Big \rfloor && \ts \leq \alpha_{\eps} \cdot \Big \lceil \frac{B}{1-\ell} \Big \rceil  & \quad \forall \ell\in \mathcal{S}_{\eps} \tag{CR$_\eps$} \label{CR-t}  \\
    & N_x\geq 0 & \notag
  \end{alignat}
%  \caption{\small Linear Program LP$_\eps$}	
%\end{mdframed}	
%\end{figure}

The claimed lower bound on the \acr of recourse-less algorithms follows from Lemma
\ref{lem:lp-opt-lb-ub-combo}.
% which we prove in \S\ref{sec:unit-appendix}
%by providing explicit feasible dual and primal solutions. 

\begin{restatable}{lem}{LPopt}\label{lem:lp-opt-lb-ub-combo}
	The optimal value $\alphas_\eps$ of \eqref{lpeps} satisfies
	$\alphas_\eps \in [\alpha - O(\eps), \alpha + O(\eps)]$.
\end{restatable}	

\iffull
\subsubsection{\ref{lpeps} as a Gap-Revealing LP}
\label{sec:lpeps-factor-rev}
\fi
\ifshort
To extend the above argument to the fully-dynamic case, we observe that any solution to \eqref{lpeps} defined by packing of $\mathcal{I}_s$ as above must satisfy some constraint (\ref{CR-t}) for some $\ell\in \mathcal{S}_{\eps}$ with equality, implying a competitive ratio of at least $\alpha_\epsilon$. Now, to beat this bound, a fully-dynamic algorithm must move at least $\epsilon$ volume of small items from bins which originally had less than $x-\epsilon$ free space. As small items have size $1/B^c=1/n^\delta$, this implies that $\Omega(\epsilon n^\delta)$ small items must be moved for every bin affected by such  movement. 
This argument yields Lemma \ref{lem:lp-valid-lb}, which together with Lemma \ref{lem:lp-opt-lb-ub-combo} implies Theorem \ref{thm:final-lb}. 
%\dwnote{Makes sense?}

\begin{restatable}{lem}{unitLpValidLb}
  \label{lem:lp-valid-lb}
  For all $\eps>0$ and $\frac{1}{2}> \delta>0$, if $\alphas_\eps$ is the
  optimal value of \eqref{lpeps}, then any fully-dynamic bin packing
  algorithm $\mathcal{A}$ with \acr 
  $\alphas_\eps-\eps$ and additive term $o(\eps^2\cdot n^\delta)$
  has recourse $\Omega(\eps^2\cdot n^{1-\delta})$ under unit movement
  costs.
\end{restatable}
\fi

%
%This proves
%Theorem~\ref{thm:final-lb} via Lemma~\ref{lem:lp-valid-lb}.
%To prove the lower and upper bound, we provide explicit feasible dual and primal solutions to this LP.
%The idea behind both solutions is to drop floors and
%ceilings (incurring small loss), replacing sums by integrals and solve the continuous
%versions, giving rise to Inequality \ref{eqn:alpha}.
% These solutions can then be discretized to give the primal and dual
%Discretization gives us the final solutions.

\subsection{Matching Algorithmic Results}     

As mentioned earlier, \ref{lpeps} also guides our algorithm. For the rest of
this section, items smaller than $\eps$ are called \emph{small}, and the
rest are 
\emph{large}. Items of size $s_i>1/2$ are \emph{huge}.

To begin, imagine a total of $B$ volume of small items
come first, followed by large items. Inspired by the LP analysis above,
we pack the small items such that an $N_\ell/B$ fraction of bins have $\ell$
free space for all $\ell \in \{0\} \cup \mathcal{S}_{\eps}$, where the $N_\ell$ values
are near-optimal for~\eqref{lpeps}. We call the space profile used by
the small items the ``curve''; see Figure~\ref{fig:lowerbound}. In the
LP analysis above, we showed that this packing can be extended to a
packing with \acr $\alpha + O(\eps)$ if $B/(1-\ell)$ items of size $\ell$ are added. But what if 
large items of \emph{different} sizes
are inserted and deleted? In 
\fullshort{\S\ref{sec:largeitems-full}}{\S\ref{sec:largeitems}}we 
show that this approach retains its $(\alpha+O(\epsilon))$-\acr in 
this case too, and outline a linear-time algorithm to obtain such a packing.
%configuration LP (of size depending on $\eps$ only) based on the small items' curve,
%which can then be used to pack the large items with the same \acr.
%Repacking large items in a lazy fashion once the number of large items changes by a multiplicative $1\pm O(\epsilon)$ factor results in bounded amortized recourse, with $\alphas_\epsilon+O(\epsilon)$ \acr.

The next challenge is that the small items may also be inserted and deleted. 
%Suppose only the small items change. 
In \S\fullshort{\ref{sec:fitcurve-full}}{\ref{sec:fitcurve}} 
we show how to dynamically pack the small items with bounded recourse, so that the 
number of bins with any amount of free space $f\in \mathcal{S}_\eps$ induce a near-optimal solution to \ref{lpeps}.
%The idea is to pack the small items into ``clumps'' of bins such that the volume of small items in them fit the curve perfectly, up to an $\epsilon$ fraction on average. To do so with little recourse we extend a simple $(1+\epsilon)$-approximate algorithm for instances consisting of solely small items in \S\ref{sec:unit-small-warm}.\dwnote{Where is the appendix?}

Finally, in \S\fullshort{\ref{sec:mainalgo-full}}{\ref{sec:mainalgo}} we combine the two
ideas together to obtain our fully-dynamic algorithm.

\subsubsection{\ref{lpeps} as a Factor-Revealing LP}
\fullshort{\label{sec:largeitems-full}}{\label{sec:largeitems}}

% Now we address our bounded-recourse fully-dynamic bin packing
% algorithm. As mentioned above the~\ref{lpeps} is instrumental in
% designing the on-line algorithm as well. In the previous section, we
% had defined small items to be those with size at most $1/B^c$, where $B
% \geq 1/\eps$. However, now we will call an item small if its size is at
% most $\eps$. This affects~\ref{lpeps} by changing the RHS of
% constraint~\eqref{Vol-eps} to $B-\eps$. It is easy to check that the
% conversion to~\ref{lpepsnew} still incurs only $O(\pm \eps)$ additive
% change in the objective value.

In this section we show how we use the linear program \ref{lpeps} to analyze and guide the following algorithm: pack the small items according to a near-optimal solution to \ref{lpeps}, and pack the large items near-optimally ``on top'' of the small items.

To analyze this approach, we first show it yields a good packing if all large items are huge and have some common size $\ell>1/2$.
Consider an instance $\cI_s$ consisting solely of small items with total volume $B$, packed using $N_x$ bins with gaps in the range $[x, x +
\eps)$ for all $x \in \{0\} \cup \mathcal{S}_{\eps}$, where $\{\alpha_\eps,
N_0, N_x : x \in \mathcal{S}_{\eps}\}$ form a feasible solution
for~\eqref{lpeps}. We say such a packing is \emph{$\alpha_\eps$-feasible}. By the LP's definition, any $\alpha_\eps$-feasible packing of small items can be extended to an
$\alpha_\eps$-competitive packing for any extension $\mathcal{I}_{\ell}$ of $\cI_s$ with
$\ell \in \mathcal{S}_{\eps}$, by packing the size $\ell$ items in the bins
counted by $N_x$ for $x \geq \ell$ before using new bins. In fact, this solution satisfies 
a similar property for any extension $\mathcal{I}_{\ell}^k$ obtained from
$\cI_s$ by adding \emph{any} number $k$ of items all of size $\ell$. (Note that  $\mathcal{I}_{\ell}$ is the special case of
$\mathcal{I}_{\ell}^k$ with $k = \lfloor B/(1-\ell) \rfloor$.) 

%\ifshort
%The proof of this lemma is deferred to \S\ref{sec
%\fi

\ifshort
\begin{restatable}[Huge Items of Same Size]{lem}{UnitStrongCurve}
  \label{lem:strong-curve} 
  Any $\alpha_\eps$-feasible packing of small items of $\cI_s$ induces an
  $\alpha_\eps$-competitive packing for all extensions $\mathcal{I}_{\ell}^k$ of
  $\cI_s$ with $\ell > 1/2$ and $k \in \mathbb{N}$.
\end{restatable}
\fi
\iffull
\UnitStrongCurve*
\fi

\ifshort 

Now, to pack the large items of the instance $\cI$, we first create a similar new instance $\cI'$ whose large items are all huge items. To do so, we first need the following observation.

\begin{restatable}{obs}{UnitNiceOpt}
	\label{obs:nice-opt}
	For any input $\mathcal{I}$ made up of solely large items and function $f(\cdot)$, a 
	packing of $\cI$ using at most $(1+\eps)\cdot OPT(\mathcal{I})+f(\eps^{-1})$ bins has all but at most $2\eps\cdot OPT(\mathcal{I})+2f(\eps^{-1})+3$ of its bins containing either no large items or being more than half filled by large items.
\end{restatable}

\begin{wrapfigure}{r}{0.425\textwidth}
%	\vspace{-.5cm}		
	\centering
	\capstart	
	\includegraphics[width=1.8in]{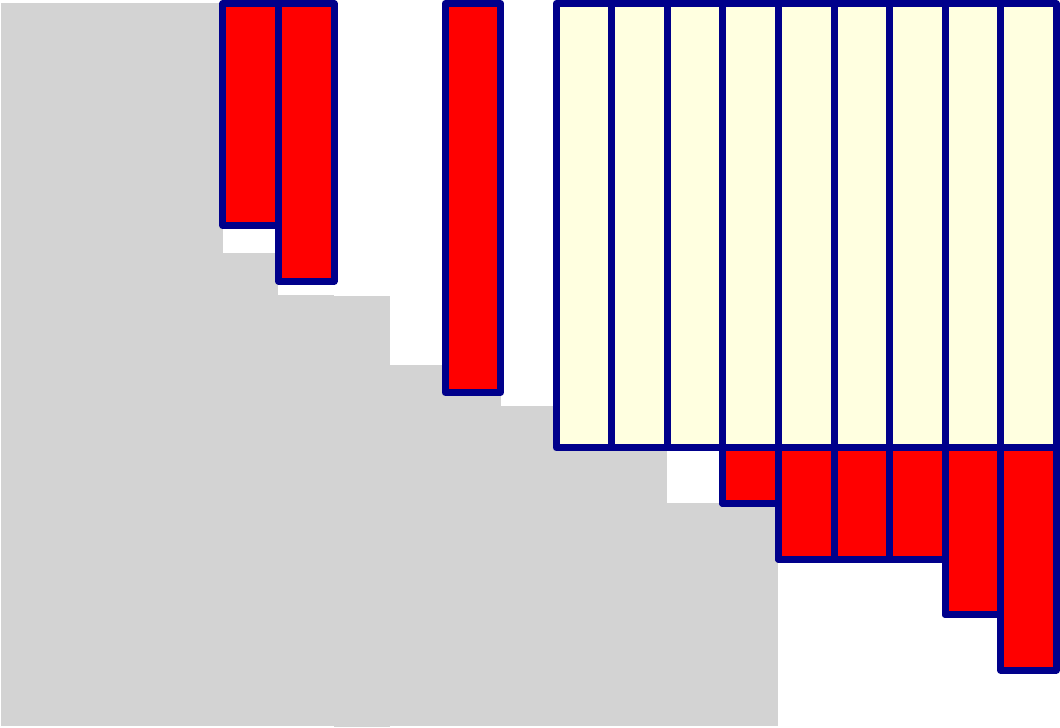}
	\caption{\small A packing of instance $\mathcal{I}'$. Large items are packed ``on top'' of the small items (in grey). Parts of large items not in the instance $\mathcal{I}_{\ell}^k$ are indicated in red.}
	\label{fig:upperbound}
	\vspace{-.25cm}
	%  \end{center}
	%\end{figure}
\end{wrapfigure}

%\color{blue}

Consider a packing $\mathcal{P}$ of the large items of $\cI$ using at most
$(1+\eps)\cdot OPT(\cI)+f(\eps^{-1})$ bins. By \Cref{obs:nice-opt}, at most
$2\eps\cdot OPT(\cI)+O(f(\eps^{-1}))$ bins in $\mathcal{P}$ are at most half
full.  We use these bins when packing $\cI$. For each of the remaining bins of
$\mathcal{P}$, we ``glue'' all large items occupying the same bin into a single
item, yielding a new instance $\cI'$ with only huge items, with any packing of
$\cI'$ ``on top'' of the small items of $\cI$ trivially inducing a similar
packing of $\cI$ with the same number of bins.  We pack the huge items of
$\cI'$ on top of the curve greedily, repeatedly packing the smallest such item
in the fullest bin which can accommodate it. Now, if we imagine we remove and
decrease the size of some of these huge items, this results in a new (easier)
instance of the form $\mathcal{I}_{\ell}^k$ for some $k$ and $\ell>1/2$, packed
in no more bins than $\cI'$ (see \Cref{fig:upperbound}). By
\Cref{lem:strong-curve}, the number of bins used by our packing of $\cI'$ (and of $\cI$) is at most
\begin{align*}
\alpha_\eps\cdot OPT(\mathcal{I}_{\ell}^k) & \leq \alpha_\eps\cdot OPT(\cI') \leq
(\alpha_\eps+O(\epsilon))\cdot OPT(\cI)+O(f(\eps^{-1})).
\end{align*}

To obtain worst-case bounds, we extend this idea as follows: we near-optimally
pack large items of size exceeding $1/4$. Given this packing of small and large
items,  we pack items of size in the range $(\eps,1/4]$ (which we
refer to as \emph{medium} items) using first-fit so that we only open a new bin if all bins are
at least $3/4$ full.  If we open a new bin, by a volume bound this packing has
\acr $4/3<\alpha+O(\epsilon)$. If we don't open a new bin, this packing is
$\alpha+O(\epsilon)$ competitive against an easier instance (obtained by
removing the medium items), and so we are  $\alpha+O(\epsilon)$ competitive.
These ideas underlie the following two theorems, a more complete proof of which
appears in \S\ref{sec:unit-cost-algo-appendix}.

\begin{restatable}{thm}{UnitOnlyLargeUBamortized}
	\label{thm:only-large-ub-amortized}
	An $\alpha_\eps$-feasible packing of the small items of an instance $\cI$ can be extended into a packing of all of $\cI$ using at most 
	$(\alpha_\epsilon+O(\eps))\cdot OPT(\cI)+O(\eps^{-2})$ bins in linear time for any fixed $\epsilon$.
\end{restatable}
\begin{restatable}{thm}{UnitOnlyLargeUBwc}
	\label{thm:only-large-ub-wc}
	For a dynamic instance $\mathcal{I}_t$, given a fully-dynamic $(1+\epsilon)$-\acr algorithm with additive term $f(\epsilon^{-1})$ for items of size greater than $1/4$ in $\cI_t$, and a fully-dynamic
	$\alpha_\eps$-feasible 
	packing of its small items, one can maintain a packing with \acr $(\alpha_\epsilon+O(\eps))$ with additive term $O(f(\epsilon^{-1}))$. This can be done using worst-case recourse
	\begin{enumerate}
		\item $O(\eps^{-1})$ per item move in the near-optimal fully-dynamic packing of the items of size $>1/4$,
		\item $O(\epsilon^{-1})$ per insertion or deletion of medium items, and
		\item $O(1)$ per item move in the $\alpha_\epsilon$-feasible packing of the small items.
	\end{enumerate} 
\end{restatable}
%\color{black}

%Now we remove the assumption that all large items are of the same size.
%Again, assume the $B$ volume of small items is placed according to the
%profile given by the LP solution. Let $\cI$ denote the instance
%consisting of the small and the large items. We use an extension of the
%Gomory-Gilmore configuration LP that takes a curve representing the
%packing the small items, and find the packing of configurations of large
%items to minimize the bins used.  
%%Naturally, we use the linear rounding
%%idea of de la Vega and Lueker~\cite{de1981bin} to reduce the
%%number of distinct item sizes, and hence bound the number of
%%constraints: the resulting ILP (with variable $y_c$ for each
%%configuration $c$) is given
%%in~\ref{ILP1}. 
%\fullshort{}{For readers unfamiliar with this configuration
%ILP, details are in Appendix~\ref{sec:unit-full}.} % \ggnote{Leave a not for the appendix}.
%Let $\taus$ be the optimal value of this
%ILP.
%%\dwnote{Explain variables and ILP? Alternatively, mention correctness etc'. in appendix}
%
\fi
\iffull
\fi
%\vspace{-0.3cm}
%\begin{figure}[H]
%  \begin{mdframed}	
  %\hrule\medskip
%    \begin{alignat}{2}
%      \text{minimize }  & \tau & & \tag{ILP$_\eps$} \label{ILP1}\\
%      \text{s.t. }  \textstyle N_{< x} + \sum_{c : \sizec(c) \geq x}
%      y_c &\leq \tau & \qquad \qquad  &\forall x \in  X \cup \{\infty\}  \tag{CR-ilp}\label{CR-ilp} \\
%      \textstyle \sum_c c_i y_c  &\geq n_i&  &\forall i = 1, \ldots, u \tag{it-i} \label{itemi} \\
%       y_c &\in \mathbb{Z}_{\geq 0} && \notag
%    \end{alignat}
%    \vspace{-0.5cm}
%    \caption{{\small Integer Linear Program ILP$_\eps$}}
%    \label{fig:gg-ilp}
%  %\medskip\hrule
%  \end{mdframed}	
%\end{figure}

% Note: there are $\big\lvert \cup_{x \in \tilde{X}} C(x) \big\rvert =
% \frac{1}{\eps}^{O(\frac{1}{\eps^2})}$ configurations, which is
% independent of $n$.

%\ifshort
%\begin{restatable}{lem}{UnitIlp}
%  \label{lem:ILP}
%  The algorithm to pack large items uses at most $\taus$ bins. Further,
%  any packing of $\mathcal{I}$ which respects the profile constraint for
%  the small items needs at least $\frac{\taus}{1+\eps}$ bins.
%\end{restatable}
%\fi

%The final crucial next step is to relate $\taus$---which fixes the
%profile of the small jobs and then packs the large ones
%near-optimally---to $OPT(\mathcal{I})$, which is not required to pack
%the small items in any restricted manner.  We defer the proof to \S\ref{}.
%We first prove that the
%optimal solution for $\mathcal{I}$ can be assumed to have some nice
%structure.

\ifshort
%\begin{restatable}{obs}{UnitNiceOpt}
%  \label{obs:nice-opt}
%  Fix $\eps\leq \frac{1}{3}$. For any input $\mathcal{I}$, there exists
%  a solution using at most $(1+3\eps)\cdot OPT(\mathcal{I})$ bins, such
%  that each bin containing large items (except maybe one such bin) has
%  more than $\frac{1}{2}$ unit of volume occupied by large items.\dwnote{Push to appnedix with Proof of 2.7.}
%\end{restatable}
\fi

\iffull
\UnitNiceOpt*

\begin{proof}
  We modify an optimal packing $\mathcal{P}$ of $\mathcal{I}$ requiring
  $B=OPT(\mathcal{I})$ bins so that it satisfies the criteria. During
  this process, some small items will be removed from their assigned
  bins and put aside. We denote this set by $E$, and we pack these items
  at the end. Initially $E$ is empty. Whenever there are two bins $i$
  and $j$ such that the space occupied by large items is $v_i, v_j\in
  (0,1/2]$ respectively, we transfer all the large items in these two
  bins to bin $i$. The small items are then packed arbitrarily into bin
  $j$ and into the remaining space in bin $i$. Since small items have
  size at most $\eps$, at most $2 \eps$ volume of small items will
  remain unassigned. These are put in the set $E$. Each operation
  creates one new bin with no large items, so there are at most $B$ such
  operations. Hence the total volume of (small) jobs in $E$ is $\leq 2
  \eps B$, and packing them using \FF, say, will require at most
  $2\eps\cdot B/(1-\eps)+1 \leq 3\eps\cdot B$ new bins.%  (because each
  % such bin, except perhaps for the last bin, will be at least $1-\eps$
  % full, and $\eps \leq 1/3$).
\end{proof}

We can now prove the main theorem for packing large items.
\fi

\ifshort
%\begin{restatable}{thm}{UnitOnlyLargeLb}
%  \label{thm:only-large-ub}
%  The number of bins $\taus$ is at most
%  $(\alphas_\eps + O(\eps)) \cdot OPT(\mathcal{I}) +
%  O(1)$.
%\end{restatable}
\fi
\iffull
\UnitOnlyLargeLb*
\fi

It remains to address the issue of maintaining an $\alpha_\eps$-feasible packing of small items dynamically using limited recourse.

\subsubsection{Dealing With Small Items: ``Fitting a Curve''}
\fullshort{\label{sec:fitcurve-full}}{\label{sec:fitcurve}}

We now consider the problem of packing $\eps$-small items according to
an approximately-optimal solution of \eqref{lpeps}. We abstract the
problem thus.

\ifshort	
\begin{restatable}[Bin curve-fitting]{Def}{BinCurveFit}
  Given a list of bin sizes $0\leq b_0\leq b_1\leq\ldots, b_K\leq 1$ and
  relative frequencies $f_0, f_1,f_2,\dots,f_K$, such that $f_x \geq 0$
  and $\sum_{x=0}^K f_x = 1$, an algorithm for the \emph{bin
  curve-fitting problem} must pack a set of $m$ of items
  with sizes $s_1,\ldots,s_m\leq 1$ into a minimal number of bins $N$
  such that for every $x \in [0,K]$ the number of bins of size $b_x$
  that are used by this packing lie in $\{\lfloor N \cdot f_x \rfloor,
  \lceil N\cdot f_x \rceil\}$.
\end{restatable}
\fi

\iffull
\BinCurveFit*
\fi

If we have $K=0$ with $b_0 = 1$ and $f_0 = 1$, we get standard \BP. We
want to solve the problem only for (most of the) small items, in the fully-dynamic
setting. We consider the special case with relative frequencies $f_x$
being multiples of $1/T$, for $T \in \mathbb{Z}$; e.g., $T =
O(\eps^{-1})$. Our approach is inspired by the algorithm of \cite{jansen13robust}, and 
maintains bins in
increasing sorted order of item sizes. The number of bins is
always an integer product of $T$. Consecutive bins are aggregated into
\emph{clumps} of exactly $T$ bins each, and clumps aggregated into
$\Theta(\eps^{-1})$ \emph{buckets} each.  Formally, each clump has $T$
bins, with $f_x \cdot T\in \mathbb{N}$ bins of size $b_x$ for
$x=0, \ldots, K$. The bins in a clump are ordered according to their
capacity $b_x$, so each clump looks like its target curve. Each bucket except the last 
consists of some $s\in[1/\eps, 3/\eps]$ consecutive clumps (the last bucket may have fewer than $1/\epsilon$ clumps).  See
Figure~\ref{fig:target-curve}.  For each bucket, all bins except those
in the last clump are full to within an additive $\eps$.
%(And the last bucket may have fewer than $1/\eps$ clumps.) 
%\agnote{Redraw figure.}

\begin{figure}[h]
\begin{center}
  \includegraphics[scale=0.6]{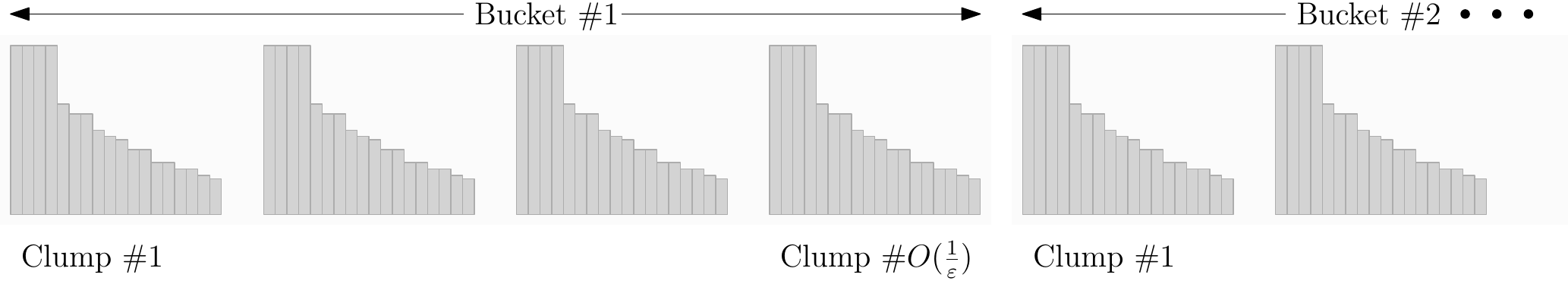}
  \caption{\small Buckets have $O(\eps^{-1})$ clumps, clumps have $T$ bins.}
  \label{fig:target-curve}
\end{center}
\end{figure}

%\vspace{-0.4cm}

Inserting an item adds it to the correct bin according to its size. If
the bin size becomes larger than the target size for the bin, the
largest item overflows into the next bin, and so on. Clearly this
maintains the invariant that we are within an additive $\eps$ of the
target size.  We perform $O(T/\epsilon)$ moves in the same bucket; if we
overflow from the last bin in the last clump of the bucket, we add a new
clump of $T$ new bins to this bucket, etc. If a bucket contains too many
clumps, it splits into two buckets, at no movement cost. An analogous
(reverse) process happens for deletes. Loosely, the process maintains
that on average the bins are full to within $O(\eps)$ of the target
fullness -- one loss of $\eps$ because each bin may have $\eps$ space,
and another because an $O(\eps)$ fraction of bins have no guarantees
whatsoever.

We now relate this process to the value of
\ref{lpeps}. We first
show that setting $T = O(\eps^{-1})$ and restricting to frequencies which are
multiples of $\eps$ does not hurt us. Indeed, for us, $b_0 = 1$, and
$b_x = (1-x)$ for $x \in \mathcal{S}_{\eps}$. Since \eqref{lpeps} depends on the
total volume $B$ of small items, and $f_x$ may change if $B$ changes, it
is convenient to work with \fullshort{the normalized LP \eqref{lpepsnew}.}{a normalized
  version, referred to as~\eqref{lpepsnew} (see Figure~\ref{fig:newlp}
  in \S\ref{sec:unit-appendix}).} Now $n_x$ can be interpreted as
just being proportional to number of bins of size $b_x$, and we can
define $f_x =n_x/\sum_x n_x$. However, we also need each $f_x$
to be an integer multiple of $1/T$ for some integer $T = O(\eps^{-1})$. We achieve this by slightly modifying the LP solution, obtaining the following.

\ifshort
\begin{restatable}[Multiples of $\eps$]{lem}{UnitTildeN}
  \label{lem:tildeN}
  For any optimal solution $\{n_x\}$ to~\eqref{lpepsnew} with 
  objective value $\alpha_\eps$, we can construct in linear time a  
  solution $\{\tilde{n}_x\}\subseteq \eps\cdot \mathbb{N}$ with objective value 
  $\alpha_\eps + O(\eps)$.
%   such that all $\tilde{n}_x$ are integral multiples of $\eps$.
\end{restatable}
\fi

Using a near-optimal solution to \eqref{lpeps} as in the above lemma, 
we obtain a bin curve-fitting instance with $T=O(\eps^{-1})$. Noting 
that if we ignore the last bucket's $O(T/\eps^{-1})=O(\eps^{-2})$ 
bins in our solution, we obtain an $(\alphas_\eps+O(\eps))$-feasible 
packing (with additive term $O(\eps^{-2}))$ of the remaining items by our algorithm above, using $O(\eps^{-2})$ worst-case recourse. We obtain the following.

\ifshort
\begin{restatable}[Small Items Follow the LP]{lem}{UnitSmallCurveFitting}
  \label{lem:small-curve-fitting-ub}
  Let $\eps \leq 1/6$. 
  Using $O(\eps^{-2})$ worst-case recourse we can maintain packing of 
  small items such that the content of all but $O(\eps^{-2})$ 
  designated bins in this packing form 
  an $(\alphas_\eps+O(\eps))$-feasible packing.
\end{restatable}
\fi

\subsubsection{Our Algorithm}
\fullshort{\label{sec:mainalgo-full}}{\label{sec:mainalgo}}

\ifshort
From \Cref{lem:tildeN} and \Cref{lem:small-curve-fitting-ub}, we can maintain an $(\alphas_\eps+O(\eps))$-feasible packing of the small items (that is, a packing of inducing a solution to \eqref{lpeps} with objective value $(\alphas_\eps+O(\eps))$), all while using $O(\eps^{-2})$ worst-case recourse. 
From \Cref{thm:only-large-ub-amortized} and \Cref{thm:only-large-ub-wc}, 
using a $(1+\eps)$-\acr packing of the large items one can extend such a packing of the small items into an $(\alphas_\eps+O(\eps))$-approximate packing for $\cI_t$, where $\alphas_\eps \leq \alpha +O(\eps)$, by \Cref{lem:lp-opt-lb-ub-combo}. It remains to address the recourse incurred by extending this packing to also pack the large items.

\mysubsubsection{Amortized Recourse} Here we periodically recompute in linear time the extension guaranteed by  \Cref{thm:only-large-ub-amortized}. 
Dividing the time into epochs and lazily 
addressing updates to large items (doing 
nothing on deletion and opening new bins 
on insertion) for $\eps\cdot N$ steps, 
where $N$ is the number of bins we use at 
the epoch's start, guarantees a 
$(\alpha+O(\eps))$-\acr throughout the 
epoch. As for this approach's recourse, 
we note that the number of large items at 
the end of an epoch of length 
$\epsilon\cdot N$ is at most 
$O(N/\epsilon)$, and so repacking them 
incurs $O(N/\eps)/(\epsilon\cdot N)=O(\epsilon^{-2})$ amortized recourse.

%\color{blue}
\mysubsubsection{Worst-Case Recourse} To obtain worst-case recourse we rely on the fully-dynamic 
$(1+\eps)$-\acr algorithm of Berndt al.~\cite[Theorem 9]{berndt15fully} to maintain 
a $(1+\eps)$-approximate packing of the sub-instance made of items of size exceeding $1/4$ 
using $\tilde{O}(\eps^{-3})$ size cost recourse, and so at most $\tilde{O}(\eps^{-3})$ item moves 
(as these items all have size $\Theta(1)$). By \Cref{thm:only-large-ub-wc}, our 
$(\alpha+O(\eps))$-feasible solution for the small items of $\cI_t$ can be extended 
dynamically to an $(\alpha+(O(\eps)))$-\acr packing of all 
of $\cI_t$, using $\tilde{O}(\eps^{-4})$ 
worst-case recourse. 
%\color{black}

From the above discussions  we obtain our main result -- tight algorithms for unit costs.

\begin{restatable}[Unit Costs: Upper Bound]{thm}{UnitUnitUb}
  \label{thm:unit-ub}
  There is a polytime fully-dynamic \BP algorithm which
  achieves $\alpha+O(\eps)$ \acr, additive
  term $O(\eps^{-2})$ and $O(\eps^{-2})$ amortized recourse, or additive term $poly(\eps^{-1})$ and $\tilde{O}(\eps^{-4})$ worst 
  case recourse, under unit movement costs. 
\end{restatable}
\fi

%%% Local Variables: 
%%% mode: latex
%%% TeX-master: "main"
%%% End: 

\section{General Movement Costs}
\label{sec:general}

We now consider %address the problem of fully-dynamic \BP with
%bounded recourse under
the case of general movement costs, and show a close connection with
the (arrival-only) online problem. We first show that the
fully-dynamic problem under general movement costs cannot achieve a
better \acr than the online problem.  Next, we %show how to
match the \acr of any \textsc{super-harmonic} algorithm
for the online problem in the fully-dynamic setting.
%We omit proofs here for brevity; the missing proofs can be found in \S\ref{sec:general-appendix}.
% performs no worse than the current
% known upper bounds for online \BP.

\subsection{Matching the Lower Bounds for Online Algorithms}
\label{sec:general-lower-bounds}

% We first show a black-box reduction implying that the fully-dynamic
% bounded-recourse problem cannot have an \acr smaller than the
%arrival-only online setting without recourse.
% \agnote{Do we need deterministic? Double-check.} For
% simplicity we consider deterministic lower bounds here, the proofs can
% be extended to randomized algorithms as well.
Formally, an \emph{adversary process} $\calB$ for the online \BP
problem is an adaptive process that, depending on the state
of the system (i.e., the current set of configurations used to pack the
current set of items) either adds a new item to the system, or stops the
request sequence. We say that an adversary process shows a lower bound
of $c$ for online algorithms for \BP if for any
online algorithm $\calA$, this process starting from the
empty system always eventually reaches a state where the \acr is at
least $c$.

\begin{restatable}{thm}{LBGeneralCost}\label{thm:lb-gen-recourse}
  Let $\beta \geq 2$. Any adversary process $\calB$ showing a lower
  bound of $c$ for the \acr of any online \BP algorithm can
  be converted into a fully-dynamic \BP instance with general
  movement costs such that any fully-dynamic \BP algorithm
  with amortized recourse at most $\beta$ must have \acr at least $c$.
\end{restatable}

Such a claim is simple for \emph{worst-case} recourse. Indeed,
given a recourse bound $\beta$, set the movement cost of the $i$-th item to be
$(\beta+\eps)^{n-i}$ for $\eps>0$. 
%As each item has cost at least $r+\eps$ 
%factor smaller than all previous items,
%For a sequence of insertions.
When the $i$-th item arrives we cannot repack \emph{any} previous item because
their movement costs are larger by a factor of $> \beta$.
%even if
%items are only inserted, 
So this is a regular online algorithm. The argument fails, however, if we allow
amortization.
% =======
% and hence this becomes a regular online algorithm. This argument does 
% not hold if we allow for amortization, and hence we give 
% more careful argument in the appendix. \ggnote{I am not sure how much of the actual argument, we wnat to preserve. I will add a better transition.}
%However, we now show a more involved construction where repacking
%cannot help in general.
% Consequently,
% %any online \BP instance or family of such inputs which imply a
% a lower bound of $c$ on the \acr of any online \BP algorithm
% implies a similar lower bound of $c$ on the \acr of any
% \emph{fully-dynamic} \BP algorithm with bounded recourse under
% general movement costs. 

To construct our lower bound instance, we start with an adversary process and create a dynamic instance as follows. Each subsequent arriving item
has exponentially decreasing (by a factor $\beta$) movement cost. When the
next item arrives, our algorithm could move certain existing items. These items
would have much higher movement cost than the arriving item, and so, this
process cannot happen too often. Whenever this happens, we {\em reset} 
the state of the process to an earlier time and remove all jobs arriving
in the intervening period. This will ensure that the algorithm
always behaves like an online algorithm which has not moved any of the
existing items. Since it cannot move jobs too often, the set of existing items
which have not been moved by the algorithm grow over time. 
This idea allows us to show: 

\begin{cor}
  No fully-dynamic \BP algorithm with bounded recourse under
  general movement costs and $o(n)$ additive term is better than
  $\binpackinglb$-asymptotically competitive.
\end{cor}

\subsection{(Nearly) Matching the Upper Bounds for Online Algorithms} 

We outline some of the key ingredients in obtaining 
an algorithm with competitive ratio nearly matching our upper bound of the previous section. The first is an algorithm to pack similarly-sized items.

\begin{restatable}[Near-Uniform Sizes]{lem}{SimilarlySized}
	\label{lem:c-items-per-bin}
	There exists a fully-dynamic \BP algorithm with constant worst case
	recourse which given items of sizes
	$s_i\in [1/k,1/(k-1))$ for some integer $k\geq 1$, packs them into bins of
	which all but one contain $k-1$ items and are hence at least $1-1/k$
	full. (If all items have size $1/k$, the algorithm packs $k$
	items in all bins but one.)
\end{restatable}

Lemma \ref{lem:c-items-per-bin} readily yields a $2$-\acr
algorithm with constant recourse (see \S\ref{sec:general-simple-two-apx}). We now discuss how to obtain  
$1.69$-\acr based on this lemma and the \textsc{Harmonic}
algorithm~\cite{lee1985simple}.

\mysubsubsection{The Harmonic
  Algorithm}\label{sec:harmonic}
The idea of packing items of nearly equal size together is commonly
used in online \BP algorithms. For example, the \textsc{Harmonic}
algorithm \cite{lee1985simple} packs large items (of size $\geq \epsilon$) as in Lemma \ref{lem:c-items-per-bin}, while packing small items (of
size $\leq \eps$) into dedicated bins which are at least $1-\eps$
full on average, using e.g., \FF. This algorithm uses
$(1.69+O(\eps))\cdot OPT + O(\eps^{-1})$ bins~\cite{lee1985simple}.
Unfortunately, due to item removals, \FF won't suffice to pack small items into nearly-full bins.

To pack small items in a dynamic setting we extend our ideas for the unit cost case, 
partitioning the bins into {\em buckets} 
of $\Theta(1/\epsilon)$ many bins,
such that all but one bin in a bucket are $1-O(\epsilon)$ full, and
hence the bins are $1-O(\epsilon)$ full on average. Since the size and
cost are not commensurate, we maintain the small items in sorted order
according to their \emph{Smith ratio} ($c_i/s_i$). 
However, insertion of a small
item can create a large cascade of movements throughout the bucket. 
We only move items to/from a bin once it has $\Omega(\epsilon)$'s worth of
volume removed/inserted (keeping track of erased, or ``ghost'' items).
A potential function argument allows us to show amortized $O(\epsilon^{-2})$
recourse cost for this approach, implying the following lemma, and by \cite{lee1985simple}, a $(1.69+O(\eps))$-\acr algorithm with $O(\epsilon^{-2})$ recourse.

\begin{restatable}{lem}{SmallGeneralRecourse}\label{lem:small-ub-gen}
	For all $\eps\leq \frac{1}{6}$ there exists an asymptotically
	$\left(1+O(\eps)\right)$-competitive bin packing algorithm with
	$O(\eps^{-2})$ amortized recourse if all items
	have size at most $\eps$.
	%That is, up to time step $t$, the
	%  algorithm's overall recourse cost is $O\big(\frac{1}{\eps^2}\cdot
	%  \sum_{j=1}^t c_j\big)$.
\end{restatable}

%The full
%proof in \S\ref{sec:gen-small-items}, combined with the proof for large
%items gives:
%\begin{lem}[Simulation of Harmonic with Bounded Recourse]
%  For all $\eps>0$, there exists a fully-dynamic
%  $1.6901+O(\eps)$ algorithm with $O(\eps^{-1})$ additive term
%  and $O(\eps^{-2})$ recourse.
%\end{lem}

\mysubsubsection{Seiden's Super-Harmonic Algorithms}
%
%Seiden~\cite{seiden2002online} introduced the \emph{Super Harmonic}
%family of algorithms to capture the many extensions of \emph{Harmonic}
%(e.g.,
%\cite{lee1985simple,woeginger1993improved,ramanan1989line,richey1991improved,seiden2002online}).
%We give a short overview of this family of algorithms and show how to
%generalize our approach of \S\ref{sec:harmonic} to capture all Super
%Harmonic algorithms.
%\ggnote{I commented out the first paragraph but it might fit better in
%  the intro.} 
We now discuss our remaining ideas to match the bounds of
any Super-Harmonic algorithm \cite{seiden2002online} in the fully-dynamic setting.
A \emph{Super-harmonic} (SH) algorithm partitions the unit
interval $[0,1]$ into $K+1$ intervals
$[0,\eps],(t_0=\eps,t_1](t_1,t_2],\dots,(t_{K-1},t_K=1]$.  Small items
(of size $\leq \eps$) are packed into dedicated bins which
are $1-\eps$ full. A large item has type $i$ if its size is in the range
$(t_{i-1},t_i]$. The algorithm also colors items blue or red. Each bin
contains items of at most two distinct item types $i$ and $j$. If a bin
contains only one item type, all its items are colored the same. If a
bin contains two item types $i \neq j$, all type $i$ items are colored
blue and type $j$ ones are colored red (or vice versa). The SH algorithm
is defined by four sequences
$(\alpha_i)_{i=1}^K, (\beta_i)_{i=1}^K, (\gamma_i)_{i=1}^K$, and a
bipartite \emph{compatibility graph} $\mathcal{G} = (V,E)$.  A bin with
blue (resp., red) type $i$ items contains at most $\beta_i$ (resp.,
$\gamma_i$) items of type $i$, and is \emph{open} if it contains less
than this many type $i$ items.
The compatibility graph $\mathcal{G} = (V,E)$ has vertex set
$V=\{b_i \mid i\in [K]\}\cup \{r_j \mid j\in [K]\}$, with an edge
$(b_i,r_j)\in E$ indicating blue items of type $i$ and red items of type
$j$ are \emph{compatible} and may share a bin. In addition, an SH
algorithm must satisfy the following invariants.

\begin{enumerate}[(P1)]
\setlength{\itemsep}{0pt}
     \setlength{\parsep}{3pt}
     \setlength{\topsep}{3pt}
     \setlength{\partopsep}{0pt}
\item\label{prop:num-open-bins} The number of open bins is $O(1)$.
	\item\label{prop:frac-red} If $n_i$ is the number of type-$i$ items, the number of red
	type-$i$ items is $\lfloor \alpha_i\cdot n_i\rfloor$.
	\item\label{prop:no-unmatched-compatible} If $(b_i,r_j)\in E$ (blue type $i$ items and red type $j$ items
	are compatible), there is no pair of bins with one containing
	only blue type $i$ items and one containing only red type $j$ items.
\end{enumerate}

Appropriate choice of
$(t_i)_{i=1}^{K+1},(\alpha_i)_{i=1}^K, (\beta_i)_{i=1}^K, (\gamma_i)_{i=1}^K$ and
$\mathcal{G}$ allows one to bound the \acr of any SH algorithm. (E.g.,
Seiden gives an SH algorithm with \acr $\binpackingub$~\cite{seiden2002online}.)

% %These invariants,  %For example, one can get the following result:
% \begin{lem}[Seiden \cite{seiden2002online}]\label{lem:good-sh}
% \end{lem}

\mysubsubsection{Simulating SH algorithms.}  In a sense, SH
algorithms ignore the exact size of large items, so we can take all
items of some type and color. This extends
Lemma~\ref{lem:c-items-per-bin} to pack at most $\beta_i$ or $\gamma_i$
of them per bin to satisfy Properties \ref{prop:num-open-bins} and
\ref{prop:frac-red}. %with a little more work.
The challenge is in maintaining Property
\ref{prop:no-unmatched-compatible}: consider a bin with
$\beta_i$ blue type $i$ items and $\gamma_j$ type-$j$ items, and suppose
the type $i$ items are all removed. Suppose there exists an
open bin with items of type $i' \neq i$ compatible with $j$. If the
movement costs of both type $j$ and type $i'$ items are significantly
higher than the cost of the type $i$ items, we cannot afford to place
these groups together, violating
Property~\ref{prop:no-unmatched-compatible}. To avoid such a problem, we
use ideas from \emph{stable matchings}. We think of groups of $\beta_i$
blue type-$i$ items and $\gamma_j$ red type-$j$ items as nodes in a
bipartite graph, with an edge between these nodes if $\mathcal{G}$
contains the edge $(b_i,r_j)$. We maintain a stable matching under
updates, with priorities being the value of the costliest item in a
group of same-type items packed in the same
bin. %~\aknote{should we just say ``priority being the constliest item in this node ''?}.
The stability of this matching implies
Property~\ref{prop:no-unmatched-compatible};  we maintain this
stable matching % for any compatibility graph
% $\mathcal{G}$, we 
using (a variant of) the Gale-Shapley algorithm.
%  and crucially use the fact that the relevant
% parameters(i.e.~$\beta_i$ and $\gamma_i$) are all constant.
%~\aknote{Which parameters? May be its worth adding a line or two here}
%\dwnote{missing some oomph}\agnote{What do you want to say?}  
Finally, relying on our solution for packing small
items as in \S\ref{sec:harmonic}, we can pack the small
items in bins which are $1-\eps$ full on average.
%, which together with
%the above approach for packing large items in a way satisfying the properties of SH algorithms. 
Combined with Lemma \ref{lem:good-sh}, we
obtain following:

\begin{restatable}{thm}{UBGeneral}\label{thm:general-ub}
  There exists a fully-dynamic \BP algorithm with \acr $\binpackingub$
  and constant additive term using constant recourse under
  general movement costs.
\end{restatable}

\section{Size Movement Costs (Migration Factor)}\label{sec:migration}

In this section, we settle the optimal recourse to a.c.r tradeoff for size movement cost (referred to as \emph{migration factor} in the literature); that is, $c_i=s_i$ for each item $i$. For worst case recourse in this model (studied in \cite{epstein09robust,jansen13robust,berndt15fully}), $poly(\epsilon^{-1})$ upper and lower bounds are known for $(1+\epsilon)$-a.c.r.~algorithms \cite{berndt15fully}, though the optimal tradeoff remains elusive. We show that for amortized recourse the optimal tradeoff is $\Theta(\epsilon^{-1})$ recourse for $(1+\epsilon)$-a.c.r.

\begin{restatable}{fact}{trivialAmortizedMigration}\label{thm:ub-amortized-migration}
	For all $\epsilon \leq 1/2$, 
	there exists an algorithm requiring
	$(1+O(\eps))\cdot OPT(\mathcal{I}_t)+O(\eps^{-2})$ bins at all times
	$t$ while using only $O(\eps^{-1})$ amortized migration factor.
\end{restatable}

This upper bound is trivial -- it suffices to repack according to an AFPTAS whenever the volume changes by a multiplicative $(1+\epsilon)$ factor (for completeness we prove this fact in  \S\ref{sec:migration-appendix}). 
The challenge here is in showing this algorithm's a.c.r to recourse tradeoff is \emph{tight}. We do so by constructing an instance where addition or removal of small items of size
$\approx \eps$ causes \emph{every} near-optimal solution to be far from every
near-optimal solution after addition/removal. 

\begin{restatable}{thm}{LBAmortizedMigration}\label{thm:lb-amortized-migration}
	For infinitely many $\epsilon>0$, any fully-dynamic bin packing
	algorithm with a.c.r $(1+\eps)$ and additive
	term $o(n)$ must have \emph{amortized} migration factor of
	$\Omega(\eps^{-1})$.
\end{restatable}

Our matching lower bound relies on the well-known Sylvester
sequence~\cite{sylvester1880point}, given by the recurrence relation
$k_1=2$ and $k_{i+1}=\big(\prod_{j\leq i} k_j\big)+1$,  
%(Observe that a similar sequence is used in Euclid's argument for the infinitude of primes.)  
%or equivalently $k_{i+1}=k_i(k_i-1)+1$ for $i\geq 0$,
the first few terms of which are $2,3,7,43,1807,\dots$ 
While this sequence has been used previously in the context of \BP, our proof relies on more fine-grained divisibility
properties. 
In particular, letting $c$ be a positive integer specified later and
$\eps:=1/\prod_{\ell=1}^c k_\ell$, we use the following properties:
\begin{enumerate}[(P1)]
	\item\label{prop:sum} 
	$\frac{1}{k_1} +
	\frac{1}{k_2} + \ldots + \frac{1}{k_c} = 1- \frac{1}{\prod_{\ell=1}^c
		k_\ell} {= 1 - \eps}. $
	\item\label{prop:coprime} If $i \neq j$, then $k_i$ and $k_j$ are relatively prime. %In particular, $1/k_i$ is not an integer product of $1/k_j$.
	\item\label{prop:prod-eps} For all $i\in[c]$, the value $1/k_i=\prod_{\ell\in [c]\setminus\{i\}} k_\ell/\prod_{\ell=1}^c k_\ell$ is an integer product of $\eps$.
%	 = 1/\prod_{\ell=1}^c k_\ell$.
	\item\label{prop:divisibility} If $i\neq j\in [c]$, then $1/k_i = \prod_{\ell\in [c]\setminus\{i\}} k_\ell / \prod_{\ell=1}^c k_\ell$ is an integer product of $k_j\cdot \epsilon$.
%	 = k_j / \prod_{\ell=1}^c k_\ell$.
\end{enumerate}

We define a vector of item sizes $\vec{s} \in [0,1]^{c+1}$ in our instances as follows: for $i\in [c]$ we let $s_i = \frac{1}{k_i}\cdot (1-\frac{\eps}{2})$, and $s_{c+1} = \eps\cdot(\frac{3}{2}-\frac{\eps}{2})$.
The adversarial input sequence will alternate between two instances, $\mathcal{I}$ and $\mathcal{I}'$. For some large $N$ a product of $\prod_{\ell=1}^c k_\ell$, Instance $\mathcal{I}$ consists of $N$ items of sizes $s_i$ for all $i\in [c+1]$. Instance $\mathcal{I}'$ consists of $N$ items of all sizes but $s_{c+1}$. 

Properties \ref{prop:sum}-\ref{prop:divisibility} imply on the one hand that $\mathcal{I}$ can be packed into completely full bins containing one item of each size, while any bin which does not contain exactly one item of each size has at least $\Omega(\eps)$ free space.
Similarly, an optimal packing of $\mathcal{I}'$ packs items of the same size in one set of bins, using up exactly $1-\frac{\epsilon}{2}$ space, while any bin which contains items of at least two sizes has at least $\epsilon$ free space.
These observations imply the following.

\begin{restatable}{lem}{migrationKeyLemma}
	\label{lem:migration-instances}
	Any algorithm $\mathcal{A}$ with $(1+\eps/7)$-\acr and $o(n)$ additive term packs instance $\mathcal{I}$ such that
        at least $2N/3$ bins contain exactly one item of each size
        $s_i$, and packs instance $\mathcal{I}'$ such that at least
        $N/2$ bins contain items of exactly one size.
\end{restatable}

%For instance $\mathcal{I}'$, however, the picture is the exact opposite. By the choice of the $s_i$ and Properties \ref{prop:sum}-\ref{prop:divisibility}, an optimal packing can pack all items of the same size in one set of bins, using up exactly $1-\frac{\epsilon}{2}$ space in each bin, and thus requiring roughly $N(1-\epsilon)$ bins to store the $N(1-\epsilon)(1-\frac{\epsilon}{2})$ volume. On the other hand, any bin which contains two different item sizes has at least $\epsilon$ free space. This implies the following observation.
% \begin{lem}
% 	\label{obs:inst2}
% 	Algorithm $\mathcal{A}$ packs instance $\mathcal{I}'$ such that at least $N/2$ bins contain item of exactly one size.
% \end{lem}

Theorem \ref{thm:lb-amortized-migration} follows from Lemma \ref{lem:migration-instances} % and Lemma \ref{obs:inst2} 
in a rather straightforward fashion. The full details of this proof and
the lemmas leading up to it can be found in
\S\ref{sec:migration-appendix}.
%%% Local Variables:
%%% mode: latex
%%% TeX-master: "main"
%%% End:

\section{Acknowledgments}
The work of Anupam Gupta and Guru Guruganesh is supported in part by NSF awards CCF-1536002, CCF-1540541, and CCF-1617790. Amit Kumar and Anupam Gupta are part of the Joint Indo-US Virtual Center for Computation and Uncertainty.
The work of David Wajc is supported in part by NSF grants CCF-1527110, CCF-1618280 and NSF CAREER award CCF-1750808. 

%%% Local Variables:
%%% mode: latex
%%% TeX-master: "main"
%%% End:

%%%%%%%%%%%%%%%%%%%%%%%%%%%%%%%%%%%%%%%%%%%%%%%%%%%%%%%%%%%%
%%%%%%%%%%%%%%%%%%%%%%  APPENDIX  %%%%%%%%%%%%%%%%%%%%%%%%%%
%%%%%%%%%%%%%%%%%%%%%%%%%%%%%%%%%%%%%%%%%%%%%%%%%%%%%%%%%%%%
\newpage
\appendix
\section*{Appendix}
\section{Tabular List of Prior and Current Results}
\label{sec:tabula}

Here we recap and contrast the previous upper and lower bounds for dynamic bin packing with bounded recourse, starting with the upper bounds, in Table \ref{table:bin-packing-ubs}.
\begin{table*}[h]
	\captionsetup{justification=centering}
	\caption{Fully dynamic bin packing with limited recourse: Positive results\newline(Big-$O$ notation dropped for notational simplicity)}
	\label{table:bin-packing-ubs}
	\vspace{-0.2cm}
	\centering
	\begin{tabular}{| c | c | c | c | c | c | c |}
		\hline
		Costs & A.C.R & Additive & Recourse & W.C. & Notes & Reference \\ 	\hline	
		
		\multirow{3}{*}{General} & $2$ & $\log n$ & $1$ & \xmark & w.c. if $n$ known 
		& \textbf{Fact \ref{fact:simple-2}}
		\\
%		& $\binpackingub + 1/n$ & $1$ & $1$ & \cmark &  if $n$ known & \\				
		& $\binpackingub$ & $1$ & $1$ & \xmark &  & \textbf{Theorem \ref{thm:general-ub}}\\		
		& $1.333$ & $1$ & $1$ & \xmark & insertions only & Gambosi et al.~\cite{ %gambosi1990new,
			gambosi2000algorithms}	\bigstrut\\	
%		& $4/3$ & $1$ & $1$ & \xmark & insertions only & Gambosi \& al.~\cite{ %gambosi1990new,		gambosi2000algorithms}	\bigstrut\\			
		\hline
		\multirow{4}{*}{\vspace{-0.2cm}Unit} & $1.5+\eps$ & $\epsilon^{-1}$ & $\epsilon^{-1}$ & \cmark & insertions only & Balogh et al.~\cite{balogh2014line}	\bigstrut\\	
		& $\alpha + \eps$ & $\epsilon^{-2}$ & $\eps^{-2}$ & \xmark & $\alpha\bound$ & \textbf{Theorem \ref{thm:unit-ub}}	\bigstrut\\	
		& $\alpha + \eps$ & $poly(\epsilon^{-1})$ & $\eps^{-4}\log(\eps^{-1})$ & \cmark & $\alpha\bound$ & \textbf{Theorem \ref{thm:unit-ub}}	\bigstrut\\	
		& $\alpha + \eps$ & $\epsilon^{-1}$ & $\eps^{-2}$ & \cmark & $\alpha\bound$ & \citet{feldkord2017tight}	\bigstrut\\			
%		& $\alpha + \eps$ & $\eps^{-O(\epsilon^{-2})}$ & $\eps^{-O(\epsilon^{-2})}$ & \cmark & & 	\bigstrut\\								
%		& $\alpha + \eps$ & $\epsilon^{-8}$ & $\eps^{-O(\epsilon^{-2})}$ & \cmark &  & 	\bigstrut\\										
		\hline	 
		
		\multirow{5}{*}{Size}		 & $1+\eps$ & $1$ & $\eps^{-O(\epsilon^{-2})}$ & \cmark & insertions only & \citet{epstein09robust}	\bigstrut \\
		& $1+\eps$ & $\epsilon^{-2}$ & $\epsilon^{-4}$ & \cmark & insertions only & \citet{jansen13robust} \\
		& $1+\eps$ & $poly(\epsilon^{-1})$ & $\epsilon^{-3}\log(\eps^{-1})$ & \cmark & insertions only & \citet{berndt15fully} \\		 
		& $1+\eps$ & $poly(\epsilon^{-1})$ & $\epsilon^{-4} \log(\eps^{-1})$ & \cmark &  & Berndt et al.~\cite{berndt15fully} \\
		& $1+\eps$ & $\epsilon^{-2}$ & $\epsilon^{-1}$ & \xmark &  & \textbf{Theorem \ref{thm:ub-amortized-migration}} \\	
		\hline
	\end{tabular}
	%	(For notational simplicity, big-$O$ notation is dropped for the algorithmic results.)
\end{table*}

We contrast these upper bounds with matching and nearly-matching lower bounds, in Table \ref{table:bin-packing-lbs}. Note that here an amortized bound is a stronger bound than a corresponding worst case bound.

\begin{table*}[h]
	\captionsetup{justification=centering}
	\caption{Fully dynamic bin packing with limited recourse: Negative results}
	\label{table:bin-packing-lbs}
	\vspace{-0.2cm}
	\centering
	\begin{tabular}{| c | c | c | c | c | c | c |}
		\hline
		Costs & A.C.R & Additive & Recourse & W.C. & Notes & Reference \\ 	\hline	
		General & $\binpackinglb$ & $o(n)$ %$o(n/\log \log (\epsilon^{-1}))$ 
		& $\infty$ & \xmark & as hard as online & \textbf{Theorem \ref{thm:lb-gen-recourse}} \\
		\hline
		
		\multirow{3}{*}{\vspace{-0.3cm}Unit} & $1.333$ & $o(n)$ & $1$ & \cmark &  &  \citet{ivkovic1996fundamental} \bigstrut	\\
		& $\alpha - \eps$ & $o(n)$ & $1$  & \cmark & $\alpha\bound$ &  \citet{balogh2008lower} \bigstrut	\\
		& $\alpha - \eps$ & $o(\epsilon^2\cdot n^\delta)$ & $\Omega(\eps^2 \cdot n^{1-\delta})$ & \xmark & for all $\delta \in (0,1/2]$ &  \textbf{Theorem \ref{thm:final-lb}}\bigstrut	\\
		%		& &  &  & & $\alpha\bound$ & \\
		\hline
		%		\multirow{10}{*}{Size} & $4/3$ & $?$ & $1$ & \xmark & insertions only & Gambosi \& al.~\cite{gambosi1990new} \\
		%		 & $5/4$ & $?$ & $\log n$ & \cmark & & Ivkovi\'c \& Lloyd~\cite{ivkovic1998fully} \\
		%		 & $1+\epsilon$ & $?$ & $\log n$ & \xmark & & Ivkovi\'c \& Lloyd~\cite{ivkovic1998fully} \\		 
		\multirow{2}{*}{Size}	
		& $1+\eps$ & $o(n)$ & $\Omega(\epsilon^{-1})$ & \cmark &  & Berndt et al.~\cite{berndt15fully} \\				
		& $1+\eps$ & $o(n)$ & $\Omega(\epsilon^{-1})$ & \xmark &  & \textbf{Theorem \ref{thm:lb-amortized-migration}} \\				
		\hline
	\end{tabular}
\end{table*}

We emphasize again the tightness and near-tightness of our upper and lower bounds for the different movement costs. For general movement costs we show that the problem is at least as hard as online bin packing (without repacking), while the problem admits a $\binpackingub$-asymptotically competitive algorithm, nearly matching the state of the art $1.578$ online algorithm of \cite{balogh2017new}. For unit movement costs, we show the lower bound of $\alpha\bound$ of Balogh et al.~\cite{balogh2008lower} is sharp, by presenting an algorithm with a.c.r of $\alpha+\epsilon$ and additive term and recourse polynomial in $1/\epsilon$. We further simplify and strengthen the previous lower bound, by showing 
that any algorithm with a.c.r better than $\alpha$ requires polynomial additive term times recourse. Finally, for size movement costs, we show that an a.c.r of $1+\epsilon$ implies a recourse cost of $\Omega(\epsilon^{-1})$, even allowing for amortization (strengthening the hardness result of Berndt et al.~\cite{berndt15fully}). A simple lazy algorithm which matches this bound proves this tradeoff to be optimal too.
\section{Omitted Proofs of Section \ref{sec:unit} (Unit Movement Costs)}\label{sec:unit-appendix}

Here we provide proofs for our unit recourse upper and lower bounds.

First, we show that the optimal value of \eqref{lpeps}, restated below for ease of reference, is roughly 
$\alpha\bound$, where we recall that $\alpha$ is such that 
$1-1/\alpha$ is a solution to the following equation
\begin{equation}
\label{eqn:alpha}
3+ \ln(1/2) = \ln(x) + 1/x.
\end{equation}

\LPopt*
\begin{alignat}{5}
\text{minimize }  & \alpha_{\eps} & \tag{LP$_{\eps}$} \\
\text{s.t. } & \ts N_0 + \sum_{x \in \mathcal{S}_{\eps}} (1-x)\cdot N_x &&\geq B-1/B^c & \tag{Vol$_{\eps}$} \\
& \ts N_0 + \sum_{ x \in \mathcal{S}_{\eps}} N_x && \leq \alpha_{\eps} \cdot B  \tag{small$_{\eps}$} \\
& \ts N_0 + \sum_{x \in \mathcal{S}_{\eps}, x \leq t-\eps} N_x
+ \Big \lfloor \frac{B}{1-t} \Big \rfloor && \ts \leq \alpha_{\eps} \cdot \Big \lceil \frac{B}{1-t} \Big \rceil  & \quad \forall t\in \mathcal{S}_{\eps} \tag{CR$_\eps$}  \\
& N_x\geq 0 & \notag
\end{alignat}

\begin{proof}
We first modify~\eqref{lpeps} slightly to make it easier to work with---this
will affect its optimal value only by $O(\eps)$. 
\begin{enumerate}
	\item[(i)] Change the $B - 1/B^c$ term in the RHS of
	inequality~\eqref{Vol-eps} to $B$, and remove the floor and ceiling
	in the inequalities~\eqref{CR-t}. As $B \geq 1/\eps$, this
	affects the optimal value by $O(\eps)$.
	\item[(ii)] Divide the inequalities through by $B$, and introduce
	new variables $n_x$ for $N_x/B$, and
	\item[(iii)] Replace $\alpha_\eps - 1$ by a new variable
	$\alpha'_\eps$, and change the objective value to $\alpha'_\eps+1$.
\end{enumerate}
This gives the LP~\eqref{lpepsnew}, whose
optimal value is $\alphas_\eps \pm O(\eps)$. We will provide a feasible solution to this linear program
and a feasible dual solution for its dual linear program~\eqref{dualeps} whose objective
values are $\alpha+O(\eps)$ and $\alpha - O(\eps)$, respectively. This proves the desired result.

	\begin{figure}[h]
	\fbox{
		\begin{minipage}{.5\textwidth}
			\vspace{-0.235cm}
			\begin{alignat}{5}
			\text{min. }  & \alpha_{\eps}' +1& \tag{LPnew$_{\eps}$}  \label{lpepsnew}\\
			& n_0 + \sum_{x \in \mathcal{S}_{\eps}} (1-x)\cdot n_x &&\geq 1 &  \notag \\
			& n_0 + \sum_{ x \in \mathcal{S}_{\eps}} n_x - \alpha'_{\eps} && \leq 1 \notag  \\		
			&  n_0 + \sum_{x \in \mathcal{S}_{\eps},  x \leq t-\eps} n_x  && \leq \frac{\alpha_\eps'}{1-t}   & \quad \forall t\in \mathcal{S}_{\eps} \notag  \\
			& n_x\geq 0 & \notag
			\end{alignat}
		\end{minipage}}
	\fbox{
		\begin{minipage}{.5\textwidth}
			% \begin{mdframed}
			\begin{alignat}{5}
			\text{max. }  & Z - q_0+1 && \tag{Dual$_{\eps}$}  \label{dualeps}\\
			& q_0 + \sum_{t \in \mathcal{S}_{\eps} } \frac{q_t}{1-t} &&\leq 1 \tag{d1}\label{d1} \\
			&   q_0  + \sum_{t \in \mathcal{S}_{\eps} } q_t  && \geq Z    \tag{d2}\label{d2} \\
			&  q_0 + \sum_{t \geq x+\eps, t \in \mathcal{S}_{\eps}} q_t &&\geq (1-x)\cdot Z     \qquad  \forall x \in \mathcal{S}_{\eps} \tag{d3}\label{d3}
			\end{alignat}
			% \end{mdframed}
		\end{minipage}}
		\caption{The modified LP and its dual program}
		\label{fig:newlp}	
	\end{figure}
		
	\textbf{Upper Bounding $\mathbf{OPT(}$\ref{lpepsnew}$\mathbf{)}$.} We first give a solution that is nearly feasible, and
	then modify it to give a feasible solution with value at most $\alpha
	+ O(\eps)$. 
	
	Let $C$ denote $\alpha -1$. Define $n_x := \int_{x-\eps}^x
	\frac{C}{(1-y)^2}\, \,dy$ for all $x\in \mathcal{S}_{\eps}$ and
	\begin{align*}
	n_0 & := 1- \int_{\frac12}^{\frac{1}{\alpha } } \frac{C\cdot
		\,dy}{(1-y)} = 1 - C\ln \Big(\frac12 \Big) + C \ln
	\Big(1-\frac{1}{\alpha} \Big).
	\end{align*}
	The first constraint of~\eqref{lpepsnew} is satisfied up to an additive
	$O(\eps)$:
	\begin{align*}
	n_0 + \sum_{x} (1-x) n_x &=  1- \int_{\frac12}^{\frac{1}{\alpha}} \frac{C}{(1-y)} \,dy + \sum_{x \in \mathcal{S}_{\eps}} (1-x) \int_{x-\eps}^{x} \frac{C}{(1-y)^2} \,dy \\
	&\geq  1- \int_{\frac12}^{\frac{1}{\alpha}} \frac{C}{(1-y)} \,dy +  \int_{\frac12}^{\frac{1}{\alpha}} \frac{C(1-y)}{(1-y)^2} \,dy - O(\eps) \\
	&= 1- O(\eps).
	\end{align*}
	Next, the second constraint of~\eqref{lpepsnew}. %constraint~\eqref{small-eps}.
	\begin{align*}
	n_0 + \sum_{x \in \mathcal{S}_{\eps}} n_x &= 1-\int_{1/2}^{\frac{1}{\alpha} } \frac{C\cdot \,dy}{1-y} + \sum_{x\in \mathcal{S}_{\eps}} \int_{x-\eps}^{x} \frac{C \cdot \,dy}{(1-y)^2}  \\
	&= 1-\int_{1/2}^{\frac{1}{\alpha} } \frac{C\cdot \,dy}{1-y} + \int_{\frac12}^{\frac{1}{\alpha} } \frac{C \cdot \,dy}{(1-y)^2}  \\
	&= 1+  C \Big( -\ln(1/2) + \ln \Big(1-\frac{1}{\alpha} \Big) - 2 + \frac{1}{1- \frac{1}{\alpha}} \Big) \\
	& =  1+C \\
	& =  \alpha ,
	\end{align*}
	where the penultimate step follows by~\eqref{eqn:alpha}, and the last
	step uses $C=\alpha-1$.  Performing a similar calculation for the last
	set of constraints in~\eqref{lpepsnew}, we get % ~\eqref{CR-t}, we get
	\begin{align*}
	n_0 + \sum_{x \in \mathcal{S}_{\eps}, x<t - \eps} n_x &= n_0 + \sum_{x \in \mathcal{S}_{\eps}} n_x - \sum_{x\in \mathcal{S}_{\eps},x\geq t-\eps} n_x \\
	&= \alpha - \sum_{x\geq t-\eps,x \in \mathcal{S}_{\eps}} \int_{x-\eps}^{x} \frac{C}{(1-y)^2} \,dy  \\
	&\leq  \alpha - \int_{t}^{\frac{1}{\alpha}} \frac{C}{(1-y)^2}\,dy +O(\eps) \\
	&=  \alpha - C \Big( \frac{\alpha}{\alpha-1} - \frac{1}{1-t} \Big) +
	O(\eps) \\ 
	&=  \frac{\alpha-1 }{1-t}  + O(\eps),
	\end{align*}
	where the second equality follows from the previous sequence of
	calculations. To satisfy the constraints of~\eqref{lpepsnew}, we
	increase $n_0$ to $n_0 + O(\eps)$ and set $\alpha'_\eps$ to $\alpha -1
	+ O(\eps)$. Since $t$ is always $\geq 1/2$, this will also satisfy the
	last set of constraints. Since the optimal value of~\eqref{lpepsnew}
	is $\leq \alpha -1 + O(\eps)$, which implies that $\alphas_\eps \leq
	\alpha-1$, since we had subtracted 1 from the objective function when
	we constructed~\eqref{lpepsnew} from~\eqref{lpeps}).
	
	\textbf{Lower Bounding $\mathbf{OPT(}$\ref{lpepsnew}$\mathbf{)}$.}
	As with our upper bound on $OPT($\ref{lpepsnew}$)$, we start with a nearly-feasible dual solution to \eqref{dualeps} and
	later modify it to obtain a feasible solution. Set $Z =
	\alpha(\alpha-1)$, $q_0 = (\alpha-1)^2$, $q_{\frac12 + \eps} = \alpha
	(\alpha-1)/2$, and $q_t = \alpha(\alpha-1)\eps$ for all $t \in
	\mathcal{S}_{\eps}$ with $t \geq \frac{1}{2}+2\eps$.  The objective
	value of~\eqref{dualeps} with respect to this solution is exactly
	$\alpha(\alpha-1) - (\alpha-1)^2 +1= \alpha$.  We will now show that
	it (almost) satisfies the constraints. For sake of brevity, we do not
	explicitly write that variable $t$ takes values in $\mathcal{S}_{\eps}$
	in the limits for the sums below. First, consider constraint~\eqref{d1}:
	\begin{align*}
	q_0 + \sum_{t=\frac12 + \eps}^{\frac{1}{\alpha}} \frac{q_t}{1-t}
	&= q_0 + \frac{q_{\frac12+\eps} }{\frac12 -\eps} + \sum_{t > \frac12 + \eps}^{\frac{1}{\alpha}} \frac{q_t}{1-t} \\
	&\leq q_0 + 2q_{\frac12+\eps} + \sum_{t > \frac12 + \eps}^{\frac{1}{\alpha}} \frac{q_t}{1-t} \\
	&\leq (\alpha-1)^2 + \alpha(\alpha-1) + \int_{\frac12}^{\frac{1}{\alpha}} \frac{\alpha (\alpha-1) \,dx}{1-x} \\
	&= (\alpha-1)^2 + \alpha(\alpha-1) + \alpha (\alpha-1) \big(\ln\left(\frac12\right) - \ln\left(1-\frac{1}{\alpha}\right) \big)  \\
	&= (\alpha-1)^2 + \alpha(\alpha-1) + \alpha (\alpha-1)
	\Big(\frac{3-2\alpha}{\alpha-1} \Big) \qquad = 1
	\end{align*}
	where we used Equation~\eqref{eqn:alpha}, which follows from the definition of $\alpha$, in the
	penultimate equation. Next, consider constraint~\eqref{d2}.
	\begin{align*}
	q_0 + \sum_{t=\frac12 + \eps}^{\frac{1}{\alpha}} q_t
	&= q_0 +  q_{\frac12 + \eps} + \sum_{t>\frac12 + \eps}^{\frac{1}{\alpha}} q_t   \\
	&\geq (\alpha-1)^2  + \frac12 \alpha (\alpha-1) + \int_{\frac12}^{\frac{1}{\alpha}} \alpha(\alpha-1) \,dx - O(\eps) \\
	&= (\alpha-1)^2  + \frac12 \alpha (\alpha-1) +  \alpha(\alpha-1)
	\Big(\frac{1}{\alpha}-\frac12 \Big)  - O(\eps) \qquad = Z - O(\eps)
	\end{align*}
	Finally, consider constraint~\eqref{d3} for any $x \in \mathcal{S}_{\eps}$:
	\begin{align*}
	q_0 + \sum_{t \geq x+\eps, t \in \mathcal{S}_{\eps}} q_t &\geq (\alpha-1)^2 + \int_{x+\eps}^{\frac{1}{\alpha}} \alpha (\alpha-1) \,dx  - O(\eps)\\
	&= (\alpha-1)^2 + \Big(\frac{1}{\alpha} - x-\eps\Big) \alpha(\alpha-1) - O(\eps)  \\
	&= (\alpha-1) \Big(\alpha -1 + \Big(\frac{1}{\alpha} - x -\eps\Big)\cdot  \alpha \Big) - O(\eps) \\
	% &= (\alpha-1) \Big\(\alpha  -  x \alpha \Big) - O(\eps)  \\
	&= Z\cdot (1-x-\eps) - O(\eps).
	\end{align*}
	Finally, increase $q_0$ to $q_0 + O(\eps)$ to ensure that
	constraints~\eqref{d2} and~\eqref{d3} are satisfied. This is now a
	feasible solution to \eqref{dualeps} with objective value
	$Z-q_o+1-O(\eps) = \alpha-O(\eps)$. Hence the optimal value
	$\alphas_\eps$ for the LP~\eqref{lpepsnew} is at least $\alpha-O(\eps)$.
\end{proof}

\subsection{Proof of the Lower Bound}
We now prove that $\alphas_\eps$, the optimal value of \eqref{lpeps}, is such that any algorithm with \acr below
this $\alphas_\eps$ must have either polynomial additive term or recourse (or both).

\unitLpValidLb*

\begin{proof}	
	For any $x\in \mathcal{S}_{\eps} \cup \{0\}$, again define $N_x$ as the
	number of bins with free space in the range $[x,x+\eps)$ when
	$\mathcal{A}$ faces input $\mathcal{I}_s$.  Inequality~\eqref{Vol-eps}
	is satisfied for the same reason as above.  Recall that $B=\Theta(n^\delta)$.  As $\mathcal{A}$ is
	$(\alphas_\eps-\Omega(\eps))$-asymptotically competitive with additive
	term $o(\eps\cdot n^\delta)$, i.e., $o(\eps\cdot B)$, and
	$OPT(\mathcal{I}_s)=B$, we have $N_0 + \sum_{ x \in
		\mathcal{S}_{\eps}} N_x \leq (\alphas_{\eps}-\Omega(\eps)+o(\eps))
	\cdot B \leq \alphas_\eps\cdot B$. That is, the $N_x$'s satisfy
	constraint \eqref{small-eps} with $\alpha_\eps = \alphas_\eps$.
	
	We now claim that there exists a $\ell \in \mathcal{S}_{\eps}$ such that
	\begin{equation}\label{ineq:CR-opposite}
	N_0 + \sum_{x \in \mathcal{S}_{\eps},  x \leq  \ell-\eps} N_x +
	\Big\lfloor \frac{B}{1-\ell} \Big\rfloor \geq \alphas_{\eps} \cdot
	\Big\lceil \frac{B}{1-\ell} \Big\rceil 
	\end{equation}
	holds (notice the opposite inequality sign compared to constraint
	\eqref{CR-t}). Suppose not. Then the quantities $N_0, N_x $ for $x \in
	\mathcal{S}_{\eps}$, and $\alphas_\eps$ strictly satisfy the
	constraints~\eqref{CR-t}.  If they also strictly satisfy the
	constraint~\eqref{small-eps}, then we can maintain feasibility and
	slightly reduce $\alphas_\eps$, which contradicts the definition of
	$\alphas_\eps$.  Therefore assume that constraint~\eqref{small-eps} is
	satisfied with equality. Now two cases arise: (i) All but one variable
	among $\{N_0\} \cup \{N_x \mid x \in \mathcal{S}_{\eps}\}$ are zero.  If
	this variable is $N_0$, then tightness of~\eqref{small-eps} implies
	that $N_0 = \alphas_\eps B$. But then we satisfy~\eqref{Vol-eps} with
	slack, and so, we can reduce $N_0$ slightly while maintaining
	feasibility. Now we satisfy all the constraints strictly, and so, we
	can reduce $\alphas_\eps$, a contradiction. Suppose this variable
	happens to be $N_x$, where $x \in \mathcal{S}_{\eps}$. So, $N_x =
	\alphas_\eps B$. We will show later in Theorem~\ref{lem:lp-opt-lb-ub-combo}
	that $\alphas_\eps \leq 1.4$. Since $(1-x) \leq 1/2$, it follows that
	$(1-x) N_x \leq 0.7 B$, and so we satisfy~\eqref{Vol-eps} with
	slack. We again get a contradiction as argued for the case when $N_0$
	was non-zero, (ii) There are at least two non-zero variables among
	$\{N_0\} \cup \{N_x \mid x \in \mathcal{S}_{\eps}\}$ -- let these be
	$N_{x_1}$ and $N_{x_2}$ with $x_1 < x_2$ (we are allowing $x_1$ to be
	0). Now consider a new solution which keeps all variables $N_x$
	unchanged except for changing $N_{x_1}$ to $N_{x_1} +
	\frac{\eta}{1-x}$, and $N_{x_2}$ to $N_{x_2} - \frac{\eta}{1-x}$,
	where $\eta$ is a small enough positive constant (so that we continue
	to satisfy the constraints~\eqref{CR-t} strictly).  The LHS
	of~\eqref{Vol-eps} does not change, and so we continue to satisfy
	this. However LHS of~\eqref{small-eps} decreases strictly. Again, this
	allows us to reduce $\alphas_\eps$ slightly, which is a
	contradiction. Thus, there must exist a $\ell$ which
	satisfies~\eqref{ineq:CR-opposite}. We fix such a $\ell$ for the
	rest of the proof.
	
	Let $\mathcal{B}$ denote the bins which have less than $\ell$ free space.
	So, $|\mathcal{B}| = N_0 + \sum_{x \in {\mathcal N}_\eps: x \leq \ell -
		\eps} N_x. $ Now, we insert $\lfloor \frac{B}{1-\ell-\eps} \rfloor$
	items of size $\ell+\eps$. (It is possible that $\ell = 1/\alpha$, and so $\ell
	+ \eps \notin \mathcal{S}_{\eps}$, but this is still a valid instance).  We
	claim that the algorithm must move at least $\eps$ volume of small
	items from at least $\eps B$ bins in $\mathcal{B}$. Suppose not. Then
	the large items of size $\ell+\eps$ can be placed in at most $\eps B$
	bins in $\mathcal{B}$. Therefore, the total number of bins needed for
	$\cI_\ell$ is at least $N_0 + \sum_{x \in {\mathcal N}_\eps: x
		\leq \ell - \eps} N_x - \eps B + \Big\lfloor \frac{B}{1-\ell}
	\Big\rfloor$, which by inequality~\eqref{ineq:CR-opposite}, is at
	least $(\alphas_\eps - O(\eps)) \cdot OPT(\mathcal{I}_{\ell+\eps})$,
	because $OPT(\mathcal{I}_{\ell+\eps}) = \Big\lceil \frac{B}{1-\ell-\eps}
	\Big\rceil = \Big\lceil \frac{B}{1-\ell} \Big\rceil + O(\eps B)$. But we
	know that $\mathcal{A}$ is
	$(\alphas_\eps-\Omega(\eps))$-asymptotically competitive with additive
	term $o(\eps\cdot n^\delta)$ (which is $o(\eps\cdot
	OPT(\mathcal{I}_{\ell+\eps}))$. So it should use at most
	$(\alphas_\eps-\Omega(\eps)+o(\eps))\cdot OPT(\mathcal{I}_{\ell+\eps})$ bins,
	which is a contradiction. Since each small item has size $1/B^c$, the
	total number of items moved by the algorithm is at least $\eps^2
	B/B^c$. This is $\Omega(\eps \cdot n^{1-\delta})$, because $\eps \geq
	1/B$, and $B^c=\Theta(n^{1-\delta})$.
	\end{proof}

\subsection{Matching Algorithmic Results}\label{sec:unit-cost-algo-appendix}

We now address the omitted proofs of our matching upper bound.

\subsubsection{\ref{lpeps} as a Factor-Revealing LP}

We now present the omitted proofs allowing us to use 
\eqref{lpeps} to upper bound our algorithm's \acr.
We start by showing that an optimal solution to \eqref{lpeps} 
induces a packing of the small items which can be trivially extended (i.e., without moving any items) to
an $(\alphas_\eps+O(\eps))$-competitive packing of any number of $\ell$-sized items, for any $\ell>1/2$.

\UnitStrongCurve*

\begin{proof}
	Fix $\ell$ and $k$. Let $N_x$ be the bins with exactly free space in the packing of $\cI_s$ and let $N = \sum_{x \mid x \geq \ell, x \in \mathcal{S}_{\eps}} N_x$ be the
	bins with at least $\ell$ free space; if $\ell \geq 1/\alpha$, then $N =
	0$. Let $N' = \sum_{x \mid x \leq \ell-\eps, x \in \mathcal{S}_{\eps}}
	N_x$. Our algorithm first packs the size-$\ell$ items in the $N$ bins of
	the packing before using new bins, and hence uses $N' + \max(N, k)$ bins. If
	$k \leq N$, we are done because of the constraint~\eqref{small-eps},
	so assume $k \geq N$.  A volume argument bounds the number of bins in
	the optimal solution for $\cI_\ell^k$:
	$$
	\text{OPT}(\cI_\ell^k) \geq \begin{cases}
	k &\quad \text{ if } k(1-\ell) \geq B\\
	k + \big( B - k(1-\ell) \big) &\quad \text{ else. }
	\end{cases}
	$$
	
	We now consider two cases:
	\begin{itemize}
		\item $k(1-\ell) \geq B$: Using constraint~\eqref{CR-t}, the number of
		bins used by our algorithm is
		$$ \textstyle  N' + k \leq \frac{\alpha_\eps B}{1-\ell} + O(\eps B) +
		\left(k - \frac{B}{1-\ell} 
		\right) \leq (\alpha_\eps + O(\eps)) k. $$
		\item $k(1-\ell) < B$: Since $k$ lies between $N$ and $\frac{B}{1-\ell}$, we
		can write it as a convex combination $\frac{\lambda_1 B}{1-\ell} +
		\lambda_2 N$, where $\lambda_1 + \lambda_2 = 1, \lambda_1, \lambda_2
		\geq 0$. We can rewrite constraints~\eqref{small-eps}
		and~\eqref{CR-t} as
		$$ \textstyle  N' + \frac{B}{1-\ell} \leq \frac{\alpha_\eps B}{1-\ell} +
		O(\eps B) \quad \text{ and } \quad N' + N \leq \alpha_\eps B. $$
		Combining them with the same multipliers $\lambda_1, \lambda_2$, we
		see that $N' + k$ is at most %(up to additive $O(\eps B)$)
		$$ \textstyle \alpha_\eps \left( \frac{\lambda_1 B}{1-\ell} + \lambda_2 B \right) + O(\eps B) =
		\alpha_\eps \left( B + \frac{\lambda_1 \ell B}{1-\ell} 
		\right)  + O(\eps B) \leq  \alpha_\eps \left( B + \ell k \right) + O(\eps B). $$ 
		The desired result follows because $B + \ell k = k + (B-k(1-\ell))$ and $B$ is a lower bound on $\text{OPT}(\cI_\ell^k)$.
		%This proves
		%the desired result.
		% \agnote{A little vague with the additive term,
		%      see if can make precise.}
		\qedhere
	\end{itemize}
\end{proof}

We now proceed to provide a linear-time algorithm which for any fixed $\epsilon$, given an input $\cI$ produces a packing into $(1+O(\eps))\cdot OPT(\cI)$ bins such that in almost all bins large items occupy either no space or more than half the bin.
\UnitNiceOpt*

\begin{proof}
	Suppose this packing uses $N\leq (1+\eps)\cdot OPT(\mathcal{I})+f(\eps^{-1})$ bins of which $B\geq 2\eps\cdot OPT(\mathcal{I})+2f(\eps^{-1})+4$ are ``bad'' bins -- bins with some $v\in (0,1/2]$ volume taken up by large items. We show that this packing can be improved to require fewer than $OPT(\cI)$, which would lead to a contradiction. 
	Indeed, repeatedly combining the contents of any two bad bins until no two such bins remain would decrease the number of bins by one and the number of bad bins by at most two for each combination (so this process can be repeated at least $\lfloor B/2 \rfloor\geq \eps\cdot OPT(\mathcal{I})+f(\eps^{-1})+1$), and so would result in a new packing of the $\cI$ using at most $N-\lfloor B/2 \rfloor< OPT(\mathcal{I})$ bins.
\end{proof}

%\color{blue}
\UnitOnlyLargeUBamortized*

\begin{proof}
    We run the linear-time algorithm of \citet{de1981bin} to compute a packing
    of the large items of $\cI$ into at most $(1+\eps)\cdot
    OPT(\cI)+O(\eps^{-2})$ many bins.  By \Cref{obs:nice-opt}, all but at most
    $2\eps\cdot OPT(\cI)+O(\eps^{-2})$ of these bins have at most $\frac12$
    volume occupied by small items. We use these bins in our packing of $\cI'$.
    For the remaining bins we ``glue'' all large items occupying the same
    bin into a single \emph{huge} item.  Note that any packing of $\cI'$
    trivially induces a similar packing of the as-of-yet unpacked large items
    of $\cI$ with the same number of bins (simply pack large items glued together in
    the place of their induced glued item). Moreover, by construction all large
    items of $\cI'$ are huge (i.e., have size greater than $1/2$), and clearly
    $OPT(\cI')\leq (1+O(\eps))\cdot OPT(\cI)+O(\eps^{-2})$, as $\cI'$ can be
    packed using this many bins. As the free space in all bins is an
    integer multiple of $\epsilon$, we can round the huge items' sizes to
    integer multiples of $\epsilon$ and obtain a packing with the same number of
    bins for $\cI'$. Such a rounding allows us to bucket sort the huge items of
    $\cI'$ (and bins) in time $O(n/\epsilon)$. All the above steps take
    $O_\eps(n)$ time. It remains to address the obtained packing's
    approximation ratio.
 
    \begin{figure}[h]
	 	\centering
	 	\begin{subfigure}{.485\textwidth}
	 		\centering
	 		\includegraphics[width=1.8in]{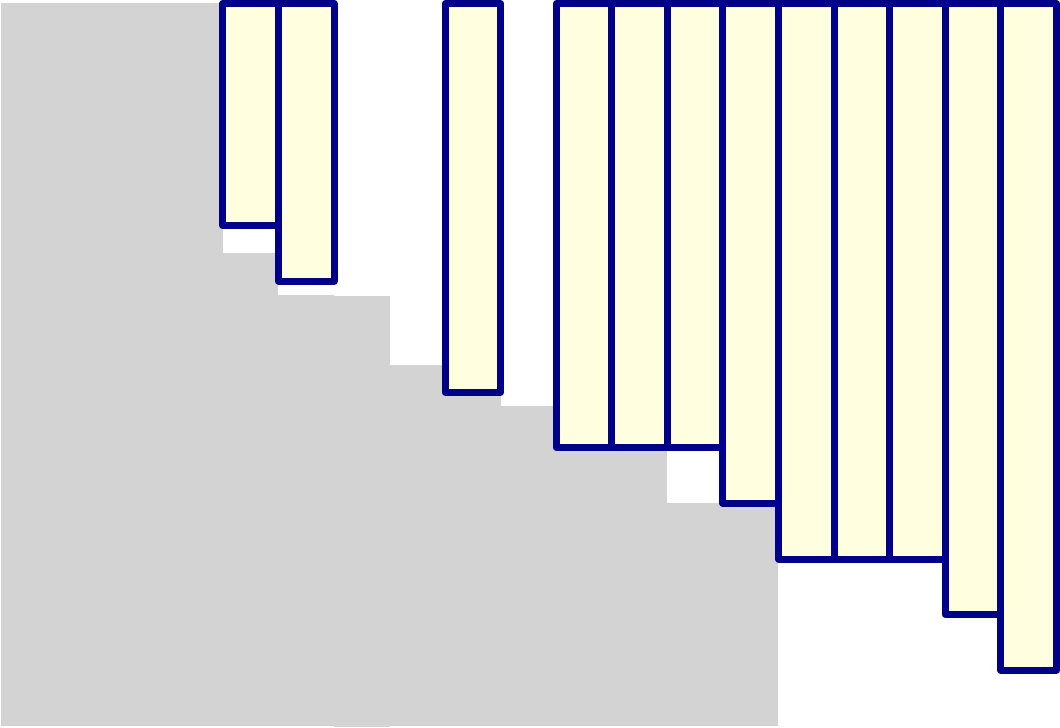}
	 		\caption{\small A greedy packing of instance $\mathcal{I}'$. Large items are packed ``on top'' of the small items (in grey).}
	 		\label{fig:upperbound-optimal-packing}
	 	\end{subfigure}%
	 	\quad 
	 	\begin{subfigure}{.485\textwidth}
	 		\centering
	 		\includegraphics[width=1.8in]{figs/upperbound}
	 		\caption{\small Packing of instance $\cI_\ell^k$ obtained from $\cI'$ by removing parts of huge items $\cI_\ell^k$ (in red).}
	 		\label{fig:upperbound-creating-Itk}
	 	\end{subfigure}
	 	\vspace{-0.15cm}
	 	\caption{A packing of instance $\cI'$ and the instance $\cI_\ell^k$ obtained from $\cI'$.}
	 	\label{fig:upperbound-app}
	 \end{figure}
	 
	 \vspace{-0.2cm}
	 In order to 
	 upper bound the number of bins used to pack $\cI'$ (and 
	 therefore $\cI$), we create a new instance of the form 
	 $\cI_\ell^k$ for some $k$ and $\ell>1/2$. Specifically, if we 
	 sort the bins containing small items in decreasing order 
	 of free space, we remove all large items packed in a bin 
	 with $f$ free space such that some bin with at least $f$
	 free space contains no large item of $\cI$.
	 Let the smallest remaining large item size be $\ell$. We 
	 decrease the size of all $k$ remaining larger items to 
	 $\ell$, yielding the instance $\cI_\ell^k$ packed on top of
	 the curve using 
	 the same number of bins as our packing of $\cI$ (see 
	 \Cref{fig:upperbound-creating-Itk}). By 
	 \Cref{lem:strong-curve}, ignoring the additive $O(\epsilon^{-2})$ term due to the packing of the large items, the packing of $\cI'$ on top of 
	 the small items -- and hence the packing we obtain
	 for $\cI$ -- uses at most $\alpha_\eps\cdot 
	 OPT(\cI_\ell^k)\leq \alpha_\eps\cdot OPT(\cI') \leq 
	 \alpha_\eps\cdot (1+O(\epsilon))\cdot OPT(\cI)$ bins. 
\end{proof}

Theorem~\ref{thm:only-large-ub-amortized} immediately gives us amortized
recourse bounds. Extending it slightly allows us to obtain worst-case
recourse bounds, as follows: replace the static $(1+\epsilon)$-\acr
algorithm by a dynamic algorithm and round the huge items to sizes which
are multiples of $\epsilon$. Doing this naively yields an
$(\alpha+O(\epsilon))$ \acr, with worst-case recourse bounds
%(see first two paragraphs of the following proof).  Using this idea, one
%can easily obtain worst-case recourse bounds which
that are at most $O(\epsilon^{-3})$ times the recourse bounds of a
fully-dynamic $(1+\epsilon)$-\acr bin packing algorithm for items of
size at least $\epsilon$. Using the worst-case recourse bounds from
\citet[Theorem 9]{berndt15fully}, we get an $\tilde{O}(\epsilon^{-6})$
worst-case recourse bound.

The following theorem improves this worst-case recourse bound to being
only an $O(\epsilon^{-1})$-times worse, instead of the naive $O(\epsilon^{-3})$
times worse. This, combined with \citet[Theorem 9]{berndt15fully}, gives an 
% on a fully-dynamic
% $(1+\epsilon)$-\acr bin packing algorithm for large items (and
% consequently yields an
improved $\tilde{O}(\epsilon^{-4})$ worst-case
recourse bound.
% \agnote{Guys, the flow
%   here seems a little disjointed. In particular the following sentence
%   (``We now show...'') seems to be subsumed by the previous
%   sentence. And it was difficult to see how to .}

% We now show how to extend the above ideas to obtain a fully-dynamic $(\alpha+\eps)$-competitive algorithm with polynomial worst-case recourse.

\UnitOnlyLargeUBwc*

\begin{proof}
	The idea here is similar to the proof of \Cref{thm:only-large-ub-amortized}, and we mainly discuss the recourse bound.
	
	By the theorem's hypothesis, we have a $(1+\epsilon)$-\acr packing of the subinstance made up of the large items of size greater than $1/4$. We ``glue'' items in bins which are at least half full into single huge items, yielding an instance $\cI'$ with $OPT(\cI')\leq (1+O(\eps))\cdot OPT(\cI)+O(f(\eps^{-1}))$, and which by \Cref{obs:nice-opt} has at most $2\eps\cdot OPT(\cI)+O(f(\eps^{-1}))$ non-huge items. We pack these items in individual bins.
	We pack $\cI'$ greedily on top of the curve, by packing the huge items in order of increasing huge item size, according to \FF, with the bins sorted in increasing order of free space
	(see \Cref{fig:upperbound-optimal-packing}).
	As in the linear-time algorithm of \Cref{thm:only-large-ub-amortized}, rounding the huge items' sizes to integer
	products of $\epsilon$ does not increase the number of bins used to pack $\cI'$. For \Cref{thm:only-large-ub-amortized} this allowed us to pack the huge items in 
	linear time. This rounding also proves useful in obtaining low worst-case recourse, as follows. 
	
    First, consider only the small items and the large items of size greater
    than $1/4$. We dynamically pack the former into some
    $\alpha_\epsilon$-feasible packing, and the large items using the
    fully-dynamic $(1+\epsilon)$-\acr algorithm of the theorem's statement.
    Whenever a bin is changed in the packing of the large items of size greater
    than $1/4$ (an item added/removed by the dynamic packing algorithm), we
    either changed a bin which was less than half full (in which case we move
    its $O(1)$ items in our packing of $\cI_t$), or we create a new item in the
    dynamic instance $\cI'_t$. Now, we can dynamically keep the huge items of
    $\cI_t'$ packed as though inserted in sorted \FF order as above as
    follows: whenever a huge item is added, we add it in the bin ``after'' all
    items of the same size, possibly removing an item of larger size from this
    bin, which we then re-insert (deletion is symmetric). As the number of
    different huge item sizes is $O(\eps^{-1})$, we move at most $O(\eps^{-1})$
    huge items, and so at most $O(\eps^{-1})$ large items of $\cI_t$ (as these
    large items have size greater than $1/4$, and so the huge items correspond
    to at most three such items). Finally, we now discuss how to pack medium
    items (items of size in the range $(\epsilon,1/4]$). 
%    \todo{\small Guru, drawing?..}
	
    We strive to keep the medium-sized items “on top of” our packing for the
    large and small items so as to guarantee that the bins in the prefix of all
    but the last bin (sorted by time of opening) which contain a medium-sized
    item are all at least 3/4 full if we ignore the fact that small item groups
    and huge items have their sizes rounded to products of $\epsilon$.  By a
    volume argument (accounting for above-mentioned rounding), if we open a new
    bin, our packing has \acr $4/3+O(\eps)<\alpha+(\eps)$. If no new bin is
    opened, we can safely ignore the medium-sized items and compare the number
    of bins we open against an easier instance which does not contain these
    medium-sized items, which as we shall see later yields an \acr of
    $\alpha+O(\epsilon)$. We now discuss details of dynamically packing
    medium-sized items in this way.

    Insertion of a medium-sized item is trivial using no recourse, by adding it
    to the first bin which is less than $3/4$ full and accommodates this new item
    (or opening a new bin if necessary).  Removal of a medium-sized item is not
    much harder. Upon the removal of some or all medium-sized items from some
    bin $B$, we take the last medium items (according to the bins' order) and
    reinsert them into $B$ until $B$ is at least $3/4$ full or until it is the
    last bin with medium items. As medium items have size in the range
    $(\eps,1/4]$, this requires $O(\epsilon^{-1})$ worst-case recourse.
    Similarly, a change in the packing of small items can only increase or
    decrease the volume used in a bin by $O(\eps)$. Such an insertion can be
    addressed by removing at most $O(1)$ items from the bin until it does not
    overflow, and then reinserting them (in the case of an insertion of small
    items into a bin), or by inserting some last $O(1)$ medium items into this
    bin in the case of deletion. 	It remains to address changes due to
    updates in the packing of large items (which are grouped into huge items). 
	
    Recall that all the huge items have their sizes rounded up to products of
    $\epsilon$. This rounding allows us to obtain a packing as in
    \Cref{fig:upperbound-optimal-packing} by only repacking $O(\epsilon^{-1})$
    such huge items following an update to the packing of the large items, if
    we ignore medium items.  Now, consider the medium items that are displaced
    due to such a move. Following a removal of a huge item, a huge item of
    larger size may replace the removed item, and this huge item may be
    replaced in turn by a larger huge item, and so on for the $\epsilon^{-1}$
    different huge item sizes. We address potential overflowing of these bins
    by removing the  minimum number of medium items to guarantee these bins do
    not overflow and repacking them. As each medium item has size greater than
    $\epsilon$, we remove at most one such item per huge item size class; that
    is, $O(\epsilon^{-1})$ such items. Similarly, the removal of a huge item
    potentially causes a larger huge item to take its place, and so
    on for $O(\epsilon^{-1})$ size classes, until no next-sized huge item
    exists or the next huge item does not fit. To avoid overfilling of these
    bins, we move at most one medium item per size class (as the medium items
    have size greater than $\epsilon$). These $O(\epsilon^{-1})$ are repacked
    as above. The last bin affected may find that it has too few medium items
    and is less than $3/4$ full; we address this using $O(\epsilon^{-1})$
    worst-case recourse as in our solution for deletion of medium items in a
    bin.

    To summarize, our worst-case recourse is $O(\eps^{-1})$ per item move in
    the $(1+\eps)$-\acr dynamic packing of the large items, $O(\epsilon^{-1})$
    per addition or deletion of a medium item and $O(1)$ per item move in the
    packing of small items. It remains to bound the obtained \acr of our
    packing obtained by extending the $\alpha_\eps$-feasible packing of the
    small items.
	
	Similarly to the proof of \Cref{thm:only-large-ub-amortized}, the \acr of this algorithm is at most $\alpha+O(\epsilon)$ if no bins are opened due to the medium items. On the other hand, as argued before, this algorithm's \acr is at most $4/3<\alpha+O(\epsilon)$ if such bins are opened. Finally, we observe that the additive term is at most $O(f(\epsilon^{-1}))$, incurred by the packing of the large items and the number of non-huge items stored in designated bins.
\end{proof}

\subsubsection{Dealing With Small Items: ``Fitting a Curve''}

We now consider the problem of packing $\eps$-small items according to
an approximately-optimal solution of \eqref{lpeps}, which we abstract thus:

\BinCurveFit*

If we have $K=0$ with $b_0 = 1$ and $f_0 = 1$, we get standard \BP. We
want to solve the problem only for (most of the) small items, in the fully-\,dynamic
setting. We consider the special case with relative frequencies $f_x$
being multiples of $1/T$, for $T \in \mathbb{Z}$; e.g., $T =
O(\eps^{-1})$. Our algorithm maintains bins in
increasing sorted order of item sizes. The number of bins is
always an integer product of $T$. Consecutive bins are aggregated into
\emph{clumps} of exactly $T$ bins each, and clumps aggregated into
$\Theta(1/\eps)$ \emph{buckets} each.  Formally, each clump has $T$
bins, with $f_x \cdot T\in \mathbb{N}$ bins of size $\approx b_x$ for
$x=0, \ldots, K$. The bins in a clump are ordered according to their
target $b_x$, so each clump looks like a curve. Each bucket except the last 
consists of some $s\in[1/\eps, 3/\eps]$ consecutive clumps (the last bucket may have fewer than $1/\epsilon$ clumps).  See
Figure~\ref{fig:target-curve-app}.  For each bucket, all bins except those
in the last clump are full to within additive $\eps$ of their target
size. 
%(And the last bucket may have fewer than $1/\eps$ clumps.) 
%\agnote{Redraw figure.}

\begin{figure}[h]
	\begin{center}
		%   \includegraphics[scale=0.3]{figs/Curve.png}
		%   \includegraphics[scale=0.3]{figs/Curve.png}
		%   \includegraphics[scale=0.3]{figs/Curve.png}
		%   \dots \dots \dots
		\includegraphics[scale=0.625]{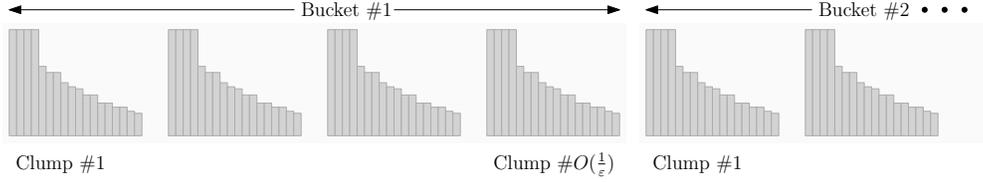}
		\caption{\small Buckets have $O(1/\eps)$ clumps, clumps have $T$ bins.}
		\label{fig:target-curve-app}
	\end{center}
\end{figure}

\vspace{-0.5cm}

Inserting an item adds it to the correct bin according to its size. If
the bin size becomes larger than the target size for the bin, the
largest item overflows into the next bin, and so on. Clearly this
maintains the invariant that we are within an additive $\eps$ of the
target size.  We perform $O(T/\epsilon)$ moves in the same bucket; if we
overflow from the last bin in the last clump of the bucket, we add a new
clump of $T$ new bins to this bucket, etc. If a bucket contains too many
clumps, it splits into two buckets, at no movement cost. An analogous
(reverse) process happens for deletes. Loosely, the process maintains
that on average the bins are full to within $O(\eps)$ of the target
fullness -- one loss of $\eps$ because each bin may have $\eps$ space,
and another because an $O(\eps)$ fraction of bins have no guarantees
whatsoever.

As observed in \S\ref{sec:unit}, the above approach approach solves bin curve-fitting using $O(T/\eps)=O(\eps^{-2})$ worst-case recourse, provided all items have size at most $\eps$ (as in our case).  provided all freq
We now relate this process to the value of
\ref{lpeps}. We first
show that setting $T = O(1/\eps)$ and restricting to frequencies to
multiples of $\eps$ does not hurt us. Indeed, for us, $b_0 = 1$, and
$b_x = (1-x)$ for $x \in \mathcal{S}_{\eps}$. Since \eqref{lpeps} depends on the
total volume $B$ of small items, and $f_x$ may change if $B$ changes, it
is convenient to work with \fullshort{the normalized LP \eqref{lpepsnew}}{a normalized
	version, referred to as~\eqref{lpepsnew} (see Figure~\ref{fig:newlp})}. Now $n_x$ can be interpreted as
just being proportional to number of bins of size $b_x$, and we can
define $f_x = {n_x}/{\sum_x n_x}$. However, we also need each $f_x$
to be integer multiples of $1/T$ for some integer $T = O(1/\eps)$. We
achieve this by slightly modifying the LP solution (which requires a
somewhat careful proof).

\UnitTildeN*

\begin{proof}
	For sake of brevity, let $\mathcal{S}_{\eps}' := \mathcal{S}_{\eps} \cup \{0\}$. Consider
	the indices $x \in \mathcal{S}'_\eps$ in increasing order, and modify $n_x$ in
	this order. Let $\Delta_x := \sum_{x' \in \mathcal{S}'_\eps : x' < x} n_{x'}$,
	and let $\tilde{\Delta}_x$ be the analogous expression for
	$\tilde{n}_x$. We maintain the invariant that $|\Delta_x -
	\tilde{\Delta}_x| \leq \eps$, which is trivially true for the base
	case $x = 0$. Inductively, suppose it is true for $x$. If $\Delta_x >
	\tilde{\Delta}_x$, define $\tilde{n}_x$ be $n_x$ rounded up to the
	nearest multiple of $\eps$, otherwise it is rounded down; this
	maintains the invariant. If we add $O(\eps)$ to the old $\alpha_\eps'$
	value, this easily satisfies the second and third set of constraints,
	since our rounding procedure ensures that the prefix sums are
	maintained up to additive $\eps$, and $t \in [0, \frac12]$. Checking
	this for the first constraint turns out to be more subtle. 
	We claim
	that 
	$$ \tilde{n}_0 + \sum_{x \in \mathcal{S}_{\eps}} (1-x) \tilde{n}_x \geq 1 -
	O(\eps). $$
	
	For an element $x \in \mathcal{S}_{\eps}'$, let $\Delta n_x$ denote $n_x - \tilde{n}_x$. It is enough to show that 
	$|\sum_{x \in \mathcal{S}_{\eps}'} x \cdot \Delta n_x| \leq O(\eps)$. Indeed, 
	\begin{align*}
	\tilde{n}_0 + \sum_{x \in \mathcal{S}_{\eps}} (1-x) \tilde{n}_x & = n_0 + \sum_{x \in \mathcal{S}_{\eps}} (1-x) n_x - \Delta n_0 - \sum_{x \in \mathcal{S}_{\eps}} (1-x) \Delta n_x \\
	& \geq 1  -  \sum_{x \in \mathcal{S}_{\eps}'} \Delta n_x  + \sum_{x \in \mathcal{S}_{\eps}'} x \cdot \Delta n_x \geq 1 - O(\eps) +  \sum_{x \in \mathcal{S}_{\eps}'} x \cdot \Delta n_x . 
	\end{align*}
	
	We proceed to bound $\sum_{x \in \mathcal{S}_{\eps}'} x \cdot \Delta
	n_x$. Define $\Delta n_{\leq x}$ as $\sum_{x' \in \mathcal{S}_{\eps}': x' \leq
		x} \Delta n_{x'}$.  Define $\Delta n_{< x}$ analogously.  Note that
	$\Delta n_{\leq x}$ stays bounded between $[-\eps, +\eps]$.  Let $I$
	denote the set of $x \in \mathcal{S}_{\eps}'$ such that $\Delta n_{\leq x}$
	changes sign, i.e., $\Delta n_{< x}$ and $\Delta n_{\leq x}$ have
	different signs. We assume w.l.o.g.\ that $\Delta n_{\leq x} = 0$ for any
	$x \in I$---we can do so by splitting $\Delta n_x$ into two parts
	(and so, having two copies of $x$ in $\mathcal{S}_{\eps}'$).  Observe that for
	any two consecutive $x_1, x_2 \in I$, the function $\Delta n_{\leq x}$ is
	unimodal as $x$ varies from $x_1$ to $x_2$, i.e., it has only one local
	maxima or minima. This is because $\Delta n_x$ is negative if $\Delta
	n_{< x}$ is positive, and \emph{vice versa}.
	
	Let the elements in $I$ (sorted in ascending order) be $x_1, x_2,
	\ldots, x_q$. Let $S_i$ denote the elements in $\mathcal{S}_{\eps}'$ which lie
	between $x_i$ and $x_{i+1}$, where we include $x_{i+1}$ but exclude
	$x_i$. Note that $\sum_{x \in S_i} \Delta n_x = \Delta n_{\leq
		x_{i+1}} - \Delta n_{\leq x_i} = 0. $ Let $x_i' = x_i + \eps$ be the
	smallest element in $S_i$. Now observe that
	$$ \sum_{x \in S_i} x \cdot \Delta n_x = x_i' \cdot \sum_{x \in S_i} \Delta n_x  + \sum_{x \in S_i} (x - x_i') \cdot \Delta n_x = 
	\sum_{x \in S_i} (x - x_i') \cdot \Delta n_x.$$ Because of the
	unimodal property mentioned above, we get $\sum_{x \in S_i} |\Delta
	n_x| \leq 2 \eps$. Therefore, the absolute value of the above sum is
	at most $2 \eps \cdot (x_{i+1} - x_i') \leq 2 \eps^2 |S_i|$, using
	that $x_{i+1} - x_i' = (|S_i|-1)\eps$.  Now summing over all $S_i$ we
	see that $\sum_x x \cdot \Delta n_x \leq O(\eps)$ because $|\mathcal{S}_{\eps}|$
	is $O(1/\eps)$. This proves the desired claim.

	%  Indeed, arrange the index set $\cN'_\eps$ in increasing
	%  order, and partition this sequence into maximal subsequences $S_1,
	%  S_2, \ldots, S_m$ such that for any $S_q$, the quantities $\{n_x\}_{x
	%    \in S_q}$ are either all rounded up or all rounded down. Let $s_q :=
	%  |S_q|$.  Consider a fixed subsequence $S_q$, and assume that $n_x \geq
	%  {\tilde n}_x$ for all $x \in S_q$ (the other case is similar).  Let
	%  indices in $S_q$ be $x_1, \ldots, x_{s_q}$ (arranged in increasing
	%  order). Now observe that
	%  \begin{eqnarray*}
	%    \sum_{x \in S_q} (1-x) n_x - \sum_{x \in S_q} (1-x) {\tilde n}_x & \leq & (1-x_1)  \sum_{x \in S_q} n_x - (1-x_{s_q}) \sum_{x \in S_q} {\tilde n}_x \\
	%    & \leq & (1-x_{s_q}) \left( \sum_{x \in S_q} n_x - \sum_{x \in S_q} \tilde{n}_x \right) + (x_{s_q} - x_1) \sum_{x \in S_q} n_x \\
	%    & \leq & (1-x_{s_q})\cdot 2\eps + \alert{\eps^2 \cdot s_q}  \ \leq \ \alert{\eps^2 s_q}.
	%  \end{eqnarray*}
	%  where used that the rounding of $n_x$ to $\tilde{n}_x$ perturbs
	%  prefix sums by an additive $\eps$, so the sum over any contiguous
	%  interval changes by at most $2\eps$. Also, $x_{s_q} - x_1 = (s_q - 1)
	%  \eps$. \agnote{Sorry, I got lost, I don't see where the $\sum_{x \in
	%      S_q} n_x$ went, and also the $(1 - x_{s_q})\eps$ term.}  The
	%  claim follows by summing the above over all subsequences, and using
	%  the fact that $\sum_q s_q = O(1/\eps)$, and $\sum_{x \in \cN'_\eps}
	%  n_x \geq 1$.  $\eps$
	
	So to satisfy the first constraint we can increase $n_0$ by
	$O(\eps)$. And then, increasing $\alpha_\eps'$ by a further $O(\eps)$
	satisfies the remaining constraints, and proves the lemma.
\end{proof}

Let ${\tilde n}_x$ be $i_x \cdot \eps$ where $i_x$ is an integer.  Note
that $\sum_x {\tilde n}_x \leq \alpha + O(\eps) \leq 2$, so dividing
through by $\eps$, $\sum_x i_x \leq 2/\eps$. Now for any index $x \in
\{0\} \cup \mathcal{S}_{\eps}$, we define $f_x := \frac{{\tilde n}_x}{\sum_{x'}
	{\tilde n_{x'}}} = \frac{i_x}{\sum_{x'} i_{x'}}$. If we set $T :=
\sum_x i_x \leq 2/\eps$, then $T$ is an integer at most $2/\eps$, and
$f_x$ are integral multiplies of $1/T$, which satisfies the requirements
of our algorithm.
Next, we show that the \,dynamic solution maintained by our algorithm
corresponds to a near-optimal
solution to~\ref{lpeps}.

\UnitSmallCurveFitting*

\begin{proof}
	The recourse bound is immediate, as each insertion or deletion causes
	a single item to move from at most $T\cdot 3/\eps$ bins and $T=O(\eps^{-1})$. 
	% If $B$ is the
	%   number of bins used, for each $x$ the number of bins of size $b_x$ is
	%   precisely $B\cdot f_x$. This is because we open bins in clumps. And each
	%   clump of $T$ consecutive bins contains precisely $T\cdot f_x$ bins of
	%   size $b_x$.
	For the rest of the argument, ignore the last bucket, contributing
	$O(\eps^{-2})$ bins to our additive term.  Let the total volume of items in the other bins
	be $B$.  Since $\eta = \sum f_x b_x$ is the average bin-size, we
	expect to use $\approx B/\eta$ bins for these items. We now show that
	we use at most $(1+O(\eps)) \cdot \frac{f_x B}{\eta}$ and at least
	$(1-O(\eps)) \cdot \frac{f_x B}{\eta}$ bins of size $b_x$ for each $x$.  
	
	Indeed, each (non-last) bucket satisfies the property that all bins in
	it, except perhaps for those in the last clump, are at least
	$\eps$-close to the target value. Since each bucket has at least
	$1/\eps$ clumps, it follows that if there are $N$ clumps and the
	target average bin-size is $\eta$, then $(1-\eps)N$ clumps are at
	least $(\eta-\eps)$ full on average. The total volume of a clump is
	$\eta \cdot T$, so $N \leq \frac{B}{(1-\eps) (\eta - \eps) \cdot T} =
	\frac{B}{\eta T}(1+O(\eps))$, where we use that $\eta \geq
	1/4$. Therefore, the total number of bins of size $b_x$ used is $f_x T
	\cdot N \leq (1+O(\eps)) \cdot \frac{B f_x}{\eta}$. The lower bound
	for the number of bins of size $b_x$ follows from a similar argument
	and the observation that if we scale the volume of small items up by a
	factor of $(1+O(\eps))$, this volume would cause each bin to be filled
	to its target value. This implies that we use at least $(1-O(\eps))
	\cdot \frac{B f_x}{\eta}$ bins with size $b_x$.
	
	We now show that the $\bar{N_x}$ satisfy \ref{lpepsnew} with $\alphas_\epsilon+O(\epsilon)$. Recall that we started with an optimal solution to~\ref{lpepsnew} of value
	$\alphas_\eps + O(\eps)$, used Lemma~\ref{lem:tildeN} to get $f_x =
	{\tilde{n_x}}/{\sum_{x'} \tilde{n_{x'}}}$ and ran the algorithm above.
	%  in \S\fullshort{\ref{sec:fitcurve-full}}{\ref{sec:fitcurve}}.
	By the computations above, ${\bar N}_x$,
	the number of bins of size $b_x$ used by our algorithm, is
	$$ (1+O(\eps)) \cdot \frac{f_x B }{\sum_{x'} f_{x'} b_{x'}} =
	(1+O(\eps)) \cdot \frac{{\tilde n}_x B}{\sum_{x'} \tilde n_{x'}b_{x'}}
	\leq (1 + O(\eps))\cdot  {\tilde n}_x B, $$ where the last inequality
	follows from the fact that $\sum_{x'} \tilde n_{x'}b_{x'} \geq 1$ (by
	the first constraint of~\ref{lpepsnew}). 
	Likewise, by the same argument, we find that these $\bar{N_x}$ satisfy \ref{CR-t} with $\alpha_\epsilon = \alphas_\eps + O(\epsilon)$.
	%    , as the number of bins of size $b_x$ used is at most $(1+O(\epsilon))\cdot N_x$, for $N_x$ the solution to 
	Finally, since $\tilde{n}_x$
	satisfies the constraints of~\ref{lpepsnew} (up to additive $O(\eps)$
	changes in $\alpha_\eps$), we can verify that the quantities ${\bar
		N}_x$ satisfy the last two constraints of~\ref{lpeps} (again up to
	additive $O(\eps)$ changes in $\alpha_\eps$). To see that they also
	satisfy~\ref{Vol-eps}, we use the following calculation: 
	$$ \sum_x {\bar N}_x \geq (1-3 \eps) \cdot \sum_x \frac{f_x B}{\sum_x b_x
		f_x } = (1-O(\eps))\cdot  \frac{B}{\sum_x b_x f_x} \geq (1-O(\eps))\cdot  B, 
	$$
	because $\sum_x f_x = 1$ and $b_x \leq 1$ for all $x$. Therefore,
	scaling all variables with $(1-O(\eps))$ will satisfy
	constraint~\ref{Vol-eps} as well.  It follows that ${\bar N}_x$
	satisfy~\ref{lpeps} with $\alpha_\eps = \alphas_\eps + O(\eps)$.
%	Finally, since $\alphas_\eps = \alpha + O(\eps)$, we get the claim.
\end{proof}

\subsubsection{Our Algorithm}
\label{sec:mainalgo-app}
Here we provide missing proofs of some of the claims made in our algorithm's analysis in \S\ref{sec:mainalgo}.

\mysubsubsection{Amortized Algorithm} 
In \S\ref{sec:mainalgo} we claimed that if we start an epoch with a packing using $N\leq (\alpha + 
O(\eps))\cdot OPT(\cI_t) + O(\eps^{-2})$ bins, as we do (here 
$\cI_t$ is the instance at the beginning of the epoch), then if for 
$\eps\cdot N$ updates we address updates na\"ively, our solution at 
any given time $t'$ during the epoch uses at most $N_{t'}\leq (\alpha 
+ O(\eps))\cdot OPT(\cI_{t'}) + O(\eps^{-2})$ bins. 
To see this, note that as each update can only change 
$OPT(\cI_{t'})$ and the number of bins we use, $N_{t'}$, by one. Therefore,  $N\cdot(1-\eps) \leq N_{t'}\leq N\cdot (1+\eps)$ and $OPT(\cI_t')\geq OPT(\cI)+\eps\cdot N$. By our upper bound on $N$, we obtain
\begin{align*}
N_{t'} &\leq (1+\eps)\cdot N \\
& \leq (\alpha+O(\eps))\cdot OPT(\cI_t) + \eps\cdot N +
O(\eps^{-2})\\
& \leq (\alpha+O(\eps))\cdot (OPT(\cI_{t'})+\eps\cdot N) + \eps\cdot N + O(\eps^{-2}) \\
& = (\alpha+O(\eps))\cdot OPT(\cI_{t'}) + O(\eps\cdot N) + O(\eps^{-2}) \\
& \leq (\alpha+O(\eps))\cdot OPT(\cI_{t'}) + O(\eps\cdot N_{t'}) + O(\eps^{-2}),
 \end{align*}
where we used in the last step the fact that $N\leq N_{t'}/(1-\eps)$. Subtracting the $O(\eps\cdot N_{t'})$ term from both sides of the above and dividing through by $(1-O(\eps))$, we obtain the claimed bound.

\section{Omitted Proofs of Section \ref{sec:general} (General Movement Costs)}\label{sec:general-appendix}

Here we provide proofs for our general amortized recourse upper and lower bounds, and discuss cases for which the algorithmic bounds can be made worst case (in \S\ref{sec:general-wc}).

\subsection{Matching the Lower Bounds for Online Algorithms}
\label{sec:general-lower-bounds-appendix}

\LBGeneralCost*

% We first show a black-box reduction implying that the fully-dynamic
% bounded-recourse problem cannot have an \acr smaller than the
% arrival-only online setting without recourse.
% \agnote{Do we need deterministic? Double-check.} For
% simplicity we consider deterministic lower bounds here, the proofs can
% be extended to randomized algorithms as well.

\begin{proof}[Proof of \Cref{thm:lb-gen-recourse}]
  Take the adversary process $\calB$, and use it to generate an instance
  for the fully-dynamic algorithm $\calA$ as follows. (When there are
  $k$ items in the system, let them be labeled $e_1, e_2, \ldots, e_k$.)
  \squishlist
  \item[I:] Given system in a state with $k$ elements, use
    $\calB$ to generate the next element $e_{k+1}$, having movement cost
    $(2\beta)^{-(k+1)}$. 

  \item[II:] If $\calA$ places $e_{k+1}$ into some bin and does not
    move any other item, go back to Step~I to generate the next element.

  \item[III:] However, if $\calA$ moves some items, and $e_j$ is
    the item with the smallest index that is repacked, delete items
    $e_{j+1}, \ldots, e_{k+1}$.  This may in turn cause elements to be
    repacked, so delete all items with indices strictly higher than than
    the smallest index item that is repacked. Eventually we stop at a
    state with $k' \geq 1$ elements such that only element $e_{k'}$ has
    been repacked. Now go back to Step I. (Also, $e_{k'+1}, \ldots,
    e_{k+1}$ are deemed undefined.)
  \squishend

  Since the location of each item $e_i$ is based only on the knowledge
  of prior elements in the sequence $e_1, e_2, \ldots, e_{i-1}$ and
  their bins, the resulting algorithm is another online algorithm. So if
  the length of the sequence eventually goes to infinity, we are
  guaranteed to reach an instance for which the \acr of this algorithm
  will be at least $c$.
  Hence we want to show that for any $n$, the length of the sequence
  eventually reaches $n$ (or the adversary process stops, having
  showed a lower bound of $c$). Consider a potential function $\Phi$
  which is zero when the system has no elements. When a new element is
  added by $\calB$, we increase $\Phi$ by $\beta$ times the movement
  cost for this element. Moreover, when $\calA$ moves elements, we
  subtract the movement costs of these elements from $\Phi$. Since
  $\calA$ ensures an amortized recourse bound of $\beta$, the potential
  must remain non-negative. 

  For a contradiction, suppose the length of the sequence remains
  bounded by $n$. Hence, there is some length $k < n$ such that Step~III
  causes the sequence to become of length $k$ arbitrarily often. Note
  that the total increase in $\Phi$ between two such events is at most
  $\beta \sum_{i = k+1}^n (2\beta)^{-i} \leq
  \frac{(2\beta)^{-(k+1)}}{2\beta-1}$. Since $\beta \geq 2$, this 
  increase is strictly less than
  $(2\beta)^{-k}$, the total decrease in $\Phi$ due to the movement of
  element $k$ alone. Since the potential decreases between two such
  events, there can only be finitely many such events, so the length of
  the sequence, i.e., the number of items in the system increases over
  time, eventually giving us the claimed lower bound. 
% \agnote{I am not
%     completely happy with this proof, needs more polish.} 
\end{proof}
\subsection{(Nearly) Matching the Upper Bounds for Online Algorithms} 

We start by proving our lemma for packing similarly-sized items.

\SimilarlySized*

\begin{proof}
	We round down movement costs of each item to the next-lower power of
	two. We maintain all items sorted by movement cost, with the costliest
	items in the first bin. All bins (except perhaps the last bin) contain
	$k-1$ items; if all items have size $1/k$ then bins contain $k$ items.
	Insertion and deletion of an item of cost $c_i=2^{\ell}$ can be assumed
	to be performed at the last bin containing an item of this cost
	(possibly incurring an extra movement cost of $c_i$, by replacing item
	$i$ with another item of cost $c_i$). If addition or deletion leaves a
	bin with one item too many or too few, we move a single item to/from
	the last bin containing an item of the next lower cost, and so
	on. Since items are sorted and the costs are powers of $2$, the total
	movement cost is at most $c_i + c_i \cdot \left(1+1/2+1/4+\dots\right)
	= 3\cdot c_i$; the loss due to rounding means the (worst case) recourse is 
	$\leq 6$. The lemma follows.
\end{proof}

\subsubsection{A Simple $2$-approximate Solution}
\label{sec:general-simple-two-apx}

Suppose we round all item sizes to powers of 2 
and use Lemma
\ref{lem:c-items-per-bin} on items of size at least $1/n$, where $n$ is the
maximum number of items in the instance $\mathcal{I}_t$ over all times $t$,
while packing all items of size less than $1/n$ into a single bin.{\footnote{We
assume we know $n$. If $n$ is unknown, we can obtain the same amortized
recourse by a simple ``guess and double/halve'' approach.}  Then, we
ensure that all but $1+\log_2n$ bins are full. This approach can be applied to general instances by rounding up item sizes to powers of two, yields the
following fact.

\begin{fact}\label{fact:simple-2}
	There exists a fully-dynamic \BP algorithm with \acr~$2$ and additive
	term $1+\log_2 n$, using constant worst case recourse.
\end{fact}

However, this simple result is highly unsatisfactory for two reasons.
Firstly, as we shall show, its \acr is suboptimal.  
The second reason concerns its additive term, which can be
blown up to be arbitrarily large without effecting the optimal solution
in any significant way, by adding $N$ items of size $1/N$ for
arbitrarily large $N$. In what follows, we will aim to design algorithms
with better \acr and both additive term and
recourse independent of $n$. 

\subsubsection{The Super Harmonic family of online algorithms}
\label{sec:SH-family}

A \emph{Super-harmonic} (abbreviated as SH) algorithm
consists of a partition of the unit interval $[0,1]$ into $K+1$
intervals, $[0,\eps],(t_0=\eps,t_1](t_1,t_2],\dots,(t_{K-1},t_K=1]$. 
Small items (i.e., items of size at most $\eps$) are packed
using \FF into dedicated bins; because the items are small, all but one
of these are at least $1-\eps$ full. For larger items, each arriving
item is of type $i$ if its size is in the range $(t_{i-1},t_i]$. Items
are colored either blue or red by the algorithm, with each bin
containing items of at most two distinct item types $i$ and $j$. If a
bin contains only one item type, its items are all colored the same, and
if a bin contains two item types $i$ and $j$, then all items of type $i$
are colored blue and items of type $j$ are colored red (or vice
versa).
The SH algorithm has associated with it three number sequences $(\alpha_i)_{i=1}^K,
(\beta_i)_{i=1}^K, (\gamma_i)_{i=1}^K$, and an underlying bipartite \emph{compatibility graph} $\mathcal{G}
= (V,E)$ whose role will be made clear
shortly.  A bin with blue (resp., red) type $i$ items
contains at most $\beta_i$ (resp., $\gamma_i$) items of type $i$, and is
\emph{open} if it contains less than $\beta_i$ type $i$ (resp., less
than $\gamma_j$ type $j$ items). The compatibility graph determines 
which pair of (colored) item types can share a bin.  The compatibility graph
$\mathcal{G} = (V,E)$ is defined on the vertex set 
$V=\{b_i \mid i\in [K]\}\cup \{r_i \mid j\in [K]\}$, with an edge $(b_i,r_j)\in E$ indicating blue items of type $i$
and red items of type $j$ are \emph{compatible}; they are
allowed to be placed in a common bin. Apart from the
above properties, an SH algorithm must satisfy the following invariants.

\begin{enumerate}[(P1)]
	\item\label{prop:num-open-bins} The number of open bins is $O(1)$.
	\item\label{prop:frac-red} If $n_i$ is the number of type-$i$ items, the number of red
	type-$i$ items is $\lfloor \alpha_i\cdot n_i\rfloor$.
	\item\label{prop:no-unmatched-compatible} If $(b_i,r_j)\in E$ (blue type $i$ items and red type $j$ items
	are compatible), there is no pair of bins with one containing
	nothing but blue type $i$ items and one containing nothing but red type $j$ items.
\end{enumerate}

The above invariants allow one to bound the asymptotic competitive ratio
of an SH algorithm, depending on the choices of $(\alpha_i)_{i=1}^K,
(\beta_i)_{i=1}^K, (\gamma_i)_{i=1}^K$ and the compatibility graph
$\mathcal{G}$. In particular, Seiden~\cite{seiden2002online} showed the
following.
\begin{lem}[Seiden \cite{seiden2002online}]\label{lem:good-sh}
	There exists an SH algorithm with \acr $\binpackingub$.
\end{lem}
A particular property the SH algorithm implied by Lemma \ref{lem:good-sh} which we will make use of 
later is that its parameters or inverses are all at most a constant; in particular, 
$\eps^{-1}=O(1)$ and $K=O(1)$, and similarly $\beta_i,\gamma_i=O(1)$ for all $i\in [K]$.

In the following sections we proceed to describe how to maintain the
above invariants that suffice to bound the competitive ratio of an SH
algorithm. We start by addressing the problem of
maintaining the SH invariants for the large items, in Section \ref{sec:sh-large-items}. 
We then show how
to pack small items (i.e., items of size at most $\eps$) into bins
such that all but a constant of these bins are at least $1-\eps$
full, in Section \ref{sec:gen-small-items}.\footnote{Strictly speaking, we will only pack small items into bins which are $1-\epsilon$ full \emph{on average}. However, redistributing these small items or portions of these items within these bins will make these bins $1-\epsilon$ full without changing the number of bins used by the solution. The bounds on SH algorithms therefore carries through.}
Finally we
conclude with our upper bound in Section \ref{sec:sh-summary}.

\subsubsection{SH Algorithms: Dealing with Large Items}
\label{sec:sh-large-items}

First, we round all movement costs to powers of $2$, increasing our recourse cost by at most a factor of $2$. Now, our algorithm will have recourse cost which will be some function of $K,(\beta_i)_{i=1}^K, (\gamma_i)_{i=1}^K$; as for our usage, we have
$K=O(1)$, and $\beta_i,\gamma_i=O(1)$ for all $i\in [K]$, we will simplify notation and assume that these values are indeed all bounded from above by a constant. Consequently, when moving around groups of up to $\beta_i$ blue (resp. $\gamma_i$ red) type $i$ items whose highest movement cost is $c$, the overall movement cost will be $O(c)$. The stipulation that $K=O(1)$ will prove useful shortly. We now explain how we maintain the invariants of SH algorithms.

\paragraph{Satisfying Property \ref{prop:num-open-bins}.}
We keep all blue items (resp. red items) of type $i$ in bins containing up to $\beta_i$ items (resp.  $\gamma_i$ items).
We sort all bins containing type-$i$ items by the cost of the costliest type-$i$ item in the bin,
where only the last bin containing type-$i$ items contains less than $\beta_i$ blue (alternatively,  $\gamma_i$ red) type-$i$ items. 
Therefore, if we succeed in maintaining the above, the number of open bins is $O(1)$; i.e., Property \ref{prop:num-open-bins} is satisfied. We explain below how to maintain this property together with Properties \ref{prop:frac-red} and \ref{prop:no-unmatched-compatible}.

\paragraph{Satisfying Property \ref{prop:frac-red}.} We will only satisfy Property \ref{prop:frac-red} \emph{approximately}, 
such that the number of red type-$i$ items is in the range $[\lfloor\alpha_i \cdot n_i\rfloor, \lfloor\alpha_i \cdot n_i\rfloor + \delta_i]$, 
where $\delta_i=\max\{\alpha_i(\beta_i-\gamma_i)+1+\gamma_i,0\}$. %(the choice of $\delta_i$ will become apparent shortly.)
Removing these at most $\delta_i$ red type-$i$ items results in a 
smaller bin packing instance, and a 
solution requiring the same number of bins, satisfying the invariants of SH algorithms.
Notice that this change to the solution can change the number of open bins by at most $\sum_i \delta_i \leq \sum_i \alpha_i\beta_i + K=O(1)$.
Therefore, satisfying Property \ref{prop:frac-red} approximately suffices to obtain our sought-after \acr. 
In fact, we satisfy a stronger property:
each prefix of bins containing type-$i$ items has a number of red type-$i$ items in the range 
$[\lfloor\alpha_i \cdot n'_i\rfloor, \lfloor\alpha_i \cdot n'_i\rfloor + \delta_i]$, where $n'_i$ is the number of type-$i$ items in the prefix.

Maintaining the above prefix invariant on deletion is simple enough: when
removing some type-$i$ item of cost $2^k$, we move the last type-$i$ item of
the next movement cost to this item's place in the packing, continuing until we
reach a type-$i$ item with no cheaper type-$i$ items. The movement cost here is
at most $2^k+2^{k-1}+2^{k-2}+\dots=O(2^k)$, so the (worst case) recourse is
constant. As the prefix invariant was satisfied before deletion, it is also
satisfied after deletion, as we effectively only remove an item from the last
bin. When inserting a type-$i$ item, we insert the item into the last
appropriate bin according to the item's movement cost. This might cause this
bin to overflow, in which case we take the cheapest item in this bin and move
it into the last appropriate bin according to this item's cost, continuing in
this fashion until we reach an open bin, or are forced to open a new bin. (The
recourse here is a constant, too). The choice of color for type-$i$
items in a newly-opened bin depends on the number of type-$i$ items before this
insertion, $n_i$, and the number of red type-$i$ items before this insertion,
$m_i$. If $m_i+\gamma_i \leq \lfloor \alpha_i(n_i + \gamma_i)\rfloor +
\delta_i$, the new bin's red items are colored red. By this condition, the
number of red items for the following $\gamma_i$ insertions into this bin will
satisfy our prefix property. If the condition is not satisfied, the bin is
colored blue. Now, as 
\begin{align*}
m_i+\gamma_i & > \lfloor \alpha_i(n_i + \gamma_i)\rfloor + \delta_i \\
& \geq \lfloor \alpha_i(n_i + \beta_i)\rfloor + \alpha_i(\gamma_i-\beta_i) - 1 + \delta_i \\
& \geq \lfloor \alpha_i(n_i + \beta_i)\rfloor - 1 + \gamma_i
\end{align*} 
we find that after this new bin contains $\beta_i$ type $i$ items, the number
of red items in the prefix is $m_i$ while the number of type-$i$ items is
$n'_i=n_i + \beta_i$, and so this prefix too satisfies the prefix condition.

We conclude that the above methods approximately maintain Property
\ref{prop:frac-red} (as well as Property \ref{prop:num-open-bins}) while only
incurring constant worst case recourse.

\paragraph{Satisfying Property \ref{prop:no-unmatched-compatible}.} Finally, in order to satisfy Property \ref{prop:no-unmatched-compatible}, we consider the groups of up to $\beta_i$ and $\gamma_j$ blue
and red items of type $i$ and $j$ packed in the same bin as nodes in a bipartite graph.
A blue type $i$ (resp. red type $j$) group is a copy of node $b_i$ (resp. $r_j$) in the compatibility graph $\mathcal{G}$, 
and copies of nodes $b_i$ and $r_j$ are connected if they are connected in $\mathcal{G}$. If a particular node
$b_i$ and $r_j$ are placed in the same bin, then we treat that edge as matched. 
Each node has a cost which is simply the maximum movement cost of an item in the group which the node represents.
All nodes have a preference order over their neighbors, preferring a costlier neighbor, while breaking ties consistently.
We will maintain a stable matching in this subgraph, where two nodes are matched if the items they represent. 
This stability clearly implies Property \ref{prop:no-unmatched-compatible}. We now proceed
to describe how to maintain this bipartite graph along with a stable matching in it.

The underlying operations we will have is addition and removal of a node of red or blue type $i$ items of cost $2^k$; 
i.e., with costliest item having movement cost $2^k$. 
As argued before, as each group contains $O(1)$ items, the
movement cost of moving items of such a group is $O(2^k)$.
Using this operation we can implement insertion and deletion of single items, 
by changing the movement cost of groups with items moved, implemented by removal of a group and 
re-insertion with a higher/lower cost if the cost changes (or simple insertion/removal for a new bin opened/bin closed). 
As argued above, the cost of items moved is during updates in order to satisfy Property \ref{prop:frac-red} is $O(2^k)$, 
where $2^k$ is the cost of the item added/removed; consequently, the costs of removals and insertions of groups is $O(2^k)$. We therefore need to show that the cost of insertion/removal of a group of cost $2^k$ is $O(2^k)$.

Insertions and deletions of blue and red nodes is symmetric, so we consider insertions and deletions of a blue node $b$ only. 
Upon insertion of some node $b$ of cost $2^k$, we insert $b$ into its place in the ordering, and scan its neighbors for the first neighbor $r$ which strictly prefers $b$ to its current match, $b'$. 
If no such $r$ exists, we are done. 
If such an $r$ exists, $b'$ is unmatched from $r$ (its items are removed from its bin) and $b$ is matched to $r$ (the items of $b$ are placed in the same bin as $r$'s bin). 
As $b'$ has strictly less than $b$, its cost is at most $2^{k-1}$. 
We now proceed similarly for $b'$ as though we inserted $b'$ into the graph. 
The overall movement cost is at most $2^k+2^{k-1}+2^{k-2}+\dots=O(2^k)$. 

Upon deletion of a blue node $b$ of cost $2^k$, if it had no previous match, we are done. If $b$ did have a match $r$, this match scans its neighbors,
starting at $b$, for its first neighbor $b'$ of cost at most $2^k$ which prefers $r$ to its current match. 
If the cost of $b'$ is $2^k$, we match $r$ to the last blue node of the same type as $b'$, denoted by $b''$. 
If $b''$ was previously matched, we proceed to match its match as if $b''$ were removed (i.e., as above). As every movement decreases the number of types of blue nodes of a given cost to consider, the overall movement cost is at most $K\cdot(2^k+2^{k-1}+2^{k-2}+\dots)=O(K\cdot 2^k)=O(2^k)$, where here we rely the number of types being $K=O(1)$. 

We conclude that Property \ref{prop:no-unmatched-compatible} can be maintained using constant worst case recourse.

\begin{lem}\label{lem:sh-properties}
	Properties \ref{prop:num-open-bins}, \ref{prop:frac-red} and \ref{prop:no-unmatched-compatible} can be maintained using constant worst-case recourse.
\end{lem}

\subsubsection{SH Algorithms: Dealing with Small Items}
\label{sec:gen-small-items}

Here we address the problem of packing small items into bins so that all
but a constant number of bins are kept at least $1-\eps$ full. Lemma \ref{lem:c-items-per-bin} allows us to do just this for items of size in the range $[\eps',\eps]$, for
$\eps'=\Omega(\eps)$. Specifically, considering all integer values $c$
in the range $[\lceil \frac{1}{\eps} \rceil, \lceil \frac{1}{\eps'}
\rceil]$, then, as $\eps^{-1}=O(1)$ and consequently $\lceil
\frac{1}{\eps'} \rceil - \lceil \frac{1}{\eps} \rceil = O(1)$, we obtain
the following.

\begin{cor}\label{cor:smallish-items}
	All items in the range $[\eps',\eps]$ for any $\eps'=\Omega(\eps)$ can
	be packed into bins which are all (barring perhaps $O(\epsilon^{-1})=O(1)$ bins) at least
	$1-\eps$ full.
\end{cor}

It now remains to address the problem of efficiently maintaining a
packing of items of size at most $\eps'$ into bins which are at least
$1-\eps$ full (again, up to some $O(1)$ possible additive term). As
$\eps'=\Omega(\eps)$, we will attempt to pack these items into bins
which are at least $1-O(\eps')$ full on average.
For notational
simplicity from here on, we will abuse notation and denote $\eps'$ by
$\eps$, contending ourselves with a packing which is
$1+O(\eps)$-competitive, and is therefore keeps bins $1-O(\eps)$ full on
average (as $OPT$ is close to the volume bound for instance made of only
small items).
%We now address the problem of maintaining approximately optimal packings for small items assuming no upper bound on $n$.

Let us first give the high-level idea of the algorithm before presenting
the formal details. 
Define the {\em density} of an item as $c_j/s_j$,
the ratio of its movement cost to its size. We arrange the bins in some
fixed order, and the items in each bin will also be arranged in the
order of decreasing density. This means the total order on the items
(consider items in the order dictated by the ordering of bins, and then
by the ordering within each bin) is in decreasing order of density as
well. Besides this, we want all bins, except perhaps the last bin, to
be approximately full (say, at least $1-O(\eps)$ full). The latter
property will trivially guarantee $(1+O(\eps))$-competitive ratio. When
we insert an item $j$, we place it in the correct bin according to its
density. If this bin overflows (i.e., items in it have total size more
than 1), then we remove some items from this bin (the ones with least
density) and transfer them to the next bin -- these items will have only
smaller density than $j$, and so, their movement cost will be comparable
to that of $j$.  If the next bin can accommodate these items, then we
can stop the process, otherwise this could lead to a long cascade. To
prevent such long cascades, we arrange the bins in {\em buckets} -- each
bucket consists of about $O(\eps^{-1})$ consecutive bins, and all these
buckets are approximately full except for the last bin in the
bucket. Again, it is easy to see that this property will ensure
$(1+O(\eps))$-competitive ratio. Note that this extends the idea of Berndt et al. Once we have these buckets, the
above-mentioned cascade stops when we reach the last bin of a
bucket. Consequently we have cascades of length at most $O(\eps^{-1})$. If
the last bin also overflows, we will add another bin at the end of this
bucket, and if the bucket now gets too many bins, we will split it into
two smaller buckets. One proceeds similarly for the case of deletes -- if
an item is deleted, we {\em borrow} some items from the next bin in the
bucket, and again this could cascade only until the last bin in the
bucket. (If the bucket ever has too few bins, we merge it with the next
bucket).

However, this cascade is not the only issue. Because items are atomic
and have varying sizes, it is possible that insertion of a tiny item
(say of size $O(\eps^2)$) could lead us to move items of size
$O(\eps)$. In this case, even though the density of the latter item is
smaller than the inserted item, its total movement cost could be much
higher. To prevent this, we ensure that whenever a bin overflows, we
move out enough items from it so that it has $\Omega(\eps)$ empty
space. Now, the above situation will not happen unless we see tiny items
amounting to a total size of $\Omega(\eps)$. In such a case, we can
charge the movement cost of the larger item to the movement cost of all
such tiny items.

The situation with item deletes is similar. When a tiny item is deleted,
it is possible that the corresponding bin underflows, and the item
borrowed from the next bin is large (i.e., has size about
$\eps$). Again, we cannot bound the movement cost of this large item in
terms of that of the item being deleted. To take care of such issues, we
do not immediately remove such tiny items from the bin. We call such
items {\em ghost} items -- they have been deleted, but we have not
removed them from the bins containing them. When a bin accumulates ghost
items of total size about $\Omega(\eps)$, we can afford to remove all
these from the bin, and the total movement cost of such items can pay
for borrowing items (whose total size would be $O(\eps)$) from the next
bin. 

The analysis of the movement
cost is done via by a potential function argument, to show the following
result (whose proof appears in Section~\ref{sec:small-SH-analysis}):
\SmallGeneralRecourse*

We first describe the algorithm formally.  Let $B_i$ denote the items
stored in a bin $i$. As mentioned above, our algorithm maintains a
solution in which items are stored in decreasing order of density. I.e.,
for all $i<i'$, for every pair of jobs $j\in B_i$ and $j'\in B_{i'}$ we
will have $c_j/v_j\geq c_{j'}/v_{j'}$. Recall that $B_i$ could contain
ghost jobs. Let $A_i$ denote the jobs in $B_i$ which have not been
deleted yet (i.e., are not ghost jobs), and $G_i$ denote the ghost jobs.
We shall use $s(B_i)$ to denote the total size of items in $B_i$ (define
$s(A_i)$ similarly).  We maintain the following invariants, satisfied by
all bins $B_i$ that are not the last bin in their bucket:
\begin{align}
\p_0:\ & 1-3\eps\leq s(B_i)\leq 1. \label{eq:p0p1}\\
\p_1:\ & s(A_i)\geq 1-4\eps. \notag
\end{align}
Finally, a bin $B_i$ which is the last bin in its bucket has no ghost
jobs. That is, $s(G_i) = 0$.  Each bucket has at most $3/\eps$
bins. Furthermore, each bucket, except perhaps for the last bucket, has
at least $\eps^{-1}$ bins.

Our algorithm is given below. We use two functions {\sc GrowBucket($U$)}
and {\sc {SplitBucket($U$)}} in these procedures, where $U$ is a
bucket. The first function is called when the bucket $U$ {\em
	underflows}, i.e., when $U$ has less than $\eps^{-1}$ bins. If $U$ is the
last bucket, then we need not do anything. Otherwise, let $U'$ be the
bucket following $U$. We merge $U$ and $U'$ into one bucket (note that
the last bin of $U$ need not satisfy the invariant conditions above, and
so, we will need to do additional processing to ensure that the
conditions are satisfied for this bin). The function {\sc
	{SplitBucket($U$)}} is called when $U$ contains more than $3/\eps$
bins. In this case, we split it into two buckets, each of size more than
$\eps^{-1}$.

\begin{center}
	\begin{minipage}[t]{.5\textwidth}

		\null
		\begin{algorithm}[H]
			\caption{\textsc{insert($j$)}}
			\begin{algorithmic}[1]
				\medskip
				\STATE{Add job $j$ into appropriate bin $i$}
				\IF{$s(B_i)>1$}
				\STATE {\sc EraseGhost($i, s_j$)} % $\min\{v_j,s(G_i)\}$ ghosts from bin $B_i$ 
				\label{line:1}
				\ENDIF
				\IF{$s(B_i)>1$ }
				\STATE \textsc{overflow}($i$, $s(B_i)-1+2\eps$)
				\ENDIF		
			\end{algorithmic}
		\end{algorithm}
	\end{minipage}%
	\begin{minipage}[t]{.5\textwidth}
		\null
		\begin{algorithm}[H]
			\caption{\textsc{delete($j$)}}
			\begin{algorithmic}[1]
				\medskip
				\STATE{Let $i$ be the bin containing $j$}
				\IF{$i$ is the last bin in its bucket}
				\STATE{Erase $j$ from $i$}
				\ELSE
				\STATE{Mark $j$ as a ghost job.}
				\IF{$s(A_i)<1-4\eps$ }
				\STATE Erase all ghost jobs  from bin $i$
				\STATE{\textsc{borrow}($i$, $1-3\eps-s(A_i)$)}
				\ENDIF
				\ENDIF		
			\end{algorithmic}
		\end{algorithm}
	\end{minipage}
\end{center}
	
\begin{center}	
	\begin{minipage}[t]{.5\textwidth}
		\null
		\begin{algorithm}[H]
			\caption{\textsc{overflow($i$, $v$)}}
			\begin{algorithmic}[1]
				\medskip
				\STATE{Let $X$ be the minimum density jobs in}
				\STATE{ \ \ \  $B_i$ s.t. $s(X)\geq v$}
				\IF{$i$ is the last bin its bucket or \\ \ \ \ \ \ \ $v \geq 1-3 \eps$}
				\STATE{Add a new bin $i'$ after $i$ in this bucket}
				\STATE{Move $X$ from $i$ to $i'$.}
				\STATE{ Let $U$ be the bucket containing $i$.}
				\IF{ $U$ has more than $3/\eps$ bins}
				\STATE{{\sc {SplitBucket($U$)}}}
				\ENDIF
				\ELSE
				\STATE{Move $X$ from bin $i$ to bin $i+1$ }
				%\STATE{Move $A\triangleq X\cap A_i$ to $B_{i+1}$}
				\IF{$s(B_{i+1})>1-\eps$}
				\STATE {\sc EraseGhost($i+1, s(B_{i+1})-1 + 2\eps$)}
				\ENDIF
				\IF{$s(B_{i+1})>1-\eps$}
				\STATE \textsc{overflow}$(i+1, s(B_{i+1})-1+2\eps$) \label{line:over}
				\ENDIF	
				\ENDIF	
			\end{algorithmic}
		\end{algorithm}
	\end{minipage}%
	\begin{minipage}[t]{.5\textwidth}
		\null
		\begin{algorithm}[H]
			\caption{\textsc{borrow($i$, $v$)}}
			\begin{algorithmic}[1]
				\medskip
				\STATE{Let $X$ be the minimum density jobs in $B_{i+1}$ s.t. $s(X\cap A_{i+1})\geq \min(v, s(A_{i+1})$}
				\STATE{Remove $X$ from $B_{i+1}$ (erase ghosts in $X$)}
				\STATE{Move $A\triangleq X\cap A_{i+1}$ to $B_{i}$}
				\IF{bin  $i+1$ is empty}
				\STATE{Remove this bin from its bucket $U$}
				\IF{$U$  has less than $\eps^{-1}$ bins and \\ \ \ \ \ \ \ \ \ \ \ \ \ it is not the last bucket}
				\STATE {\sc GrowBucket($U$)}
				\IF{$U$ has more than $3/\eps$ bins}
				\STATE {\sc SplitBucket($U$)}
				\ENDIF
				\ENDIF
				\IF { $s(A_i) < 1 - 3 \eps$}
				\STATE {\textsc{borrow($i$, $1-s(A_i)-3\eps$)}}
				\ENDIF
				\ELSE
				\IF{$s(A_{i+1})<1-3\eps$ }
				\STATE Erase all ghost jobs from bin ${i+1}$
				\STATE \textsc{borrow}(${i+1}$, $1-3\eps-s(A_{i+1})$)
				\ENDIF	
				\ENDIF	
			\end{algorithmic}
		\end{algorithm}
	\end{minipage}
	
\end{center}

We first describe the function {\textsc{insert($j$)}} for an item
$j$. We insert job~$j$ into the appropriate bin~$i$ (recall that the
items are ordered by their densities). If this bin overflows, then we
call the procedure {\sc EraseGhost($i, s_j$)}. The procedure {\sc
	EraseGHost($i, s$)}, where $i$ is a bin and $s$ is a positive
quantity, starts erasing ghost jobs from $B_i$ till one of the following
events happen: (i) $B_i$ has no ghost jobs, or (ii) total size of ghost
jobs removed exceeds $s$. Since all jobs are of size at most $\eps$,
this implies that the total size of ghost jobs removed is at most
$\min(s(G_i), s+\eps)$. Now its possible that even after removing these
jobs, the bin overflows (this will happen only if $s(G_i)$ was at most
$s_j$). In this case, we offload some of the items (of lowest density)
to the next bin in the bucket. Recall that when we do this, we would
like to create $O(\eps)$ empty space in bin $i$. So we transfer jobs of
least density in $B_i$ of total size at least $s(B_i) -1 + 2 \eps$ to
the next bin (since all jobs are small, the empty space in bin $i$ will
be in the range $[2\eps, 3 \eps]$) . This is done by calling the
procedure \textsc{overflow($i$, $s(B_i)-1+2 \eps$)}. The procedure
\textsc{overflow($i$, $v$)}, where $i$ is a bin and $v$ is a positive
quantity, first builds a set $X$ of items as follows -- consider items
in $B_i$ in increasing order of density, and keep adding them to $X$
till the total size of $X$, denoted by $s(X)$, exceeds $v$ (so, the
total size of $X$ is at most $v + \eps$). We now transfer $X$ to the
next bin in the bucket (note that by construction, $X$ does not have any
ghost jobs).  The same process repeats at this bin (although we will say
that overflow occurs at this bin if $s(B_i)$ exceeds $1-\eps$). This
cascade can end in several ways -- it is not difficult to show that
between two consecutive calls to \textsc{overflow}, the parameter $v$
grows by at most $\eps$. If $v$ becomes larger than $1-3\eps$, we just
create a new bin and assign all of $X$ to this new bin.  If we reach the
last bin, we again create a new bin and add $X$ to it. In both these
cases, the size of the bucket increases by 1, and so we may need to
split this bucket. Finally, it is possible that even after transfer of
$X$, $s(A_i)$ does not exceed $1-\eps$ -- we can stop the cascade at
this point.

Next we consider the case of deletion of an item $j$. Let $i$ be the bin
containing $j$. If $i$ is the last bin in its bucket, then we simply
remove $j$ (recall that the last bin in a bucket cannot contain ghost
jobs). Otherwise, we mark $j$ as a ghost job. This does not change
$s(B_i)$, but could decrease $s(A_i)$. If it violates property~$\p_1$,
we borrow enough items from the next bin such that $i$ has free space of
about $2\eps$ only. The function {\textsc{borrow($i$, $v$)}}, where $i$
is a bin and $v$ is a positive quantity, borrows densest (non-ghost)
items from the next bin $i+1$ of total size at least $v$. So it orders
the (non-ghost) items in $i+1$ in decreasing density, and picks them
till the total size accumulated is at least $v$.  This process may
cascade (and will stop before we reach the last bin). However there is
one subtlety -- between two consecutive calls to {\textsc{borrow}}, the
value of the parameter $v$ may grow (by up to $\eps$), and so, it is
possible that bin $i+1$ becomes empty, and we are not able to transfer
enough items from $i+1$ to $i$. In this case, we first remove $i+1$, and
continue the process (if the size of the bucket becomes too small, we
handle this case by merging it with the next bin and splitting the
resulting bucket if needed). Since we did not move enough items to $i$,
we may need to call {\textsc{borrow}($i, v'$)} again with suitable value
of $v'$.  There is a worry that the function {\textsc{borrow($i$, $v$)}}
may call {\textsc{borrow($i$, $v'$)}} with the same bin $i$, and so,
whether this will terminate. But note that, whenever this case happens,
we delete one bin, and so, this process will eventually terminate.

\subsubsection{Analysis}
\label{sec:small-SH-analysis}

We begin by showing properties of the
\textsc{overflow} and the \textsc{borrow} functions.
\begin{lem}
	\label{lem:overflow}
	Whenever \textsc{overflow}($i$, $v$) is called, bin $i$ has no ghost
	jobs. Furthermore, $s(B_i) = v + 1 - 2\eps$. Finally, when
	\textsc{overflow}($i$, $v$) ends, $1 - 3 \eps \leq s(B_i) \leq 1 - 2
	\eps$, and $s(G_i) = 0$. 
	
	Similarly, whenever {\textsc{borrow}}($i$, $v$) is called, bin $i$ has
	no ghost jobs. Furthermore, $s(B_i) = v + 1 - 3\eps$. Finally, when
	\textsc{borrow}($i$, $v$) ends, either $i$ is the last bin in its
	bucket or $1 - 3 \eps \leq s(B_i) \leq 1 - 2 \eps$, and $s(G_i) = 0$.
\end{lem}
\begin{proof}
	When we insert an item $j$ in a bin $i$, we erase ghost items from $i$
	till either (i) we erase all ghost items, or (ii) we erase ghost items
	of total size at least $s_j$. If the second case happens, then the
	fact that $s(B_i) \leq 1$ before insertion of item $j$ implies that we
	will not call \textsc{overflow} in Line~\ref{line:1} of
	{\textsc{insert($j$)}}. Therefore, if we do call \textsc{overflow}, it
	must be the case that we ended up deleting all ghost jobs in
	$i$. Further we call \textsc{overflow} with $v = s(B_i) - 1+
	2\eps$. Similarly, before we call \textsc{overflow}($i+1, v$) in
	Line~\ref{line:over}, we try to ensure ghost jobs of the same volume
	$v$ from bin $(i+1)$. If we indeed manage to remove ghost jobs of
	total size at least $v$, we will not make a recursive call to
	{\textsc{overflow}}. This proves the first part of the lemma.
	
	When {\textsc{overflow}}($i,v$), terminates, we make sure that we have
	transferred the densest $v$ volume out of $i$ to the next
	bin. Since job sizes are at most $\eps$, we will transfer at most
	items of total size in the range $[v, v+\eps]$ out of $i$. Since
	$s(B_i) = v + 1 - 2 \eps$ before calling this function, it follows
	that $s(B_i)$ after end of this function lies in the range $[1-3\eps,
	1-2\eps]$. The claim for the {\textsc{borrow}} borrow function follows
	similarly.
\end{proof}
The following corollary follows immediately from the above lemma. 
\begin{cor}\label{cor:full-bins}
	Properties \p$_0$ and \p$_1$ from (\ref{eq:p0p1}) are satisfied
	throughout by all bins which are not the last bins in their
	bucket. All bins $B_i$ which are last in their bucket satisfy
	$s(G_i)=0$.
\end{cor}
%\begin{cor}\label{cor:full-bins}
%	Properties \p$_0$ and \p$_1$ are satisfied throughout by all bins which are not the last bins in their bucket. All bins $B_i$ which are last in their bin satisfy $s(G_i)=0$.
%\end{cor}
%\begin{cor}\label{cor:potential-null}
%	All bins $B_i$ for which addition or deletion of some item $j$ entail a call to \textsc{overflow}($B_i$, $v$) or \textsc{borrow}($B_i$, $v$) satisfy $\Phi(B_i)=0$ after this addition or deletion.
%\end{cor}
% Note that after borrow/overflow, there is no further change to this bin.
We are now ready to prove the claimed bounds obtained by our algorithm for small items.

\begin{proof}[Proof of Lemma \ref{lem:small-ub-gen}]
	First we bound the competitive ratio. Barring the last bucket, which
	has $O(\eps^{-1})$ bins, all other buckets have at least $\eps^{-1}$
	bins. All bins (except perhaps for the last bin) in each of these
	buckets have at most $4\eps$ space that is empty or is filled by
	ghost jobs. Therefore, the competitive ratio is $(1+O(\eps))$ with an
	additive $O(\eps^{-1})$ bins.  
	
	It remains to bound the recourse cost of the algorithm.
	% Note that in our exposition we do not explicitly address opening new
	% bins in case the last bin overflows or in case we empty a last bin in
	% a bucket due to an overflow from its preceding bin; nor do we address
	% the issue of changing the bucketing. These are easily addressed in the
	% the obvious way while maintaining our invariants, by erasing the ghost
	% jobs in the last bin (this only possibly decreases the potential).
	For a solution $\mathcal{S}$, define the potential:
	% ($C$ is a constant to be determined later):
	\begin{gather*}
	\Phi(\mathcal{S})=\frac{4}{\eps^2}\sum_{\mbox{\tiny{bins \ }}i} \Phi(i),
	~~~~\text{where} \\
	\Phi(i) = c(G_i) + c(\mbox{fractional sparsest items above } 1-\eps
	\mbox{ in } B_i). 
	\end{gather*}
	The second term in the definition of $\Phi(i)$ is evaluated as
	follows: arrange the items in $B_i$ in decreasing order of density,
	and consider the items occupying the last $\eps$ space in bin $i$ (the
	total size of these items could be less than $\eps$ if the bin is not
	completely full). Observe that at most one item may be counted
	fractionally here. The second term is simply the sum of the movement
	costs of all such item, where a fractional item contributes the
	appropriate fraction of its movement cost.
	
	% The competitive ratio follows from the argument in the previous
	% paragaph. % follows by Corollary \ref{cor:full-bins}, by Observation
	% \ref{obs:bucket-approx}. 
	% It remains to analyze the recourse cost of our solution. 
	Now we bound the amortized movement cost with respect to 
	potential $\Phi$. 
	First consider the case when we insert item $j$ in bin $i$.  This
	could raise $\Phi_i$ by $c_j$. If bin
	$i$ does not overflow, we do not pay any movement cost. Further, deletion of ghost jobs
	from bin $i$ can only decrease the potential. Therefore, the amortized
	movement cost is bounded by $c_j$. On the other hand, suppose
	bin $i$ overflows and this results in calling the function
	\textsc{overflow}($i,v$) with a suitable $v$.  Before this function
	call, let $I$ denote the set of items which are among the sparsest
	items above $(1-\eps)$ volume in bin $i$ (i.e., these contribute
	towards $\Phi_i$).  Let $d$ be the density of the least density item in
	$I$. Since $s(B_i) \geq 1$, it follows that $\Phi_i \geq d \eps$.
	When this procedure ends, Lemma~\ref{lem:overflow} shows that $s(B_i)
	\leq 1 - \eps$, and so, $\Phi_i$ would be 0.  Thus, $\Phi_i$
	decreases by at least $d \eps - c_j$. Lemma~\ref{lem:overflow} also
	shows that if we had recursively called \textsc{overflow}($i',v'$) for
	any other bin $i'$, then $s(B_{i'})$ would be at most $1-\eps$, and
	so, $\Phi_{i'}$ would be 0 when this process ends. It follows that the
	overall potential function $\Phi$ decreases by at least $4(d\eps -
	c_j)/\eps^2$. Let us now estimate the total movement cost. We transfer
	items of total size at most 1 from one bin to another. This
	process will clearly end when we reach the end of the bucket, and so
	the total size of items moved during this process is at most $3/\eps$,
	i.e., the number of bins in this bucket. The density of these items
	is at most $d$, and so the total movement cost is at most $3
	d/\eps$. Thus, the amortized movement cost is at most $c_j/\eps^2$.
	
	Now we consider deletion of an item $j$ which is stored in bin
	$i$. Again the interesting case is when this leads to calling the
	function \textsc{borrow}. Before $j$ was deleted, $s(B_i)$ was at
	least $1 - 3\eps$ (property~$\p_0$). After we mark $j$ as a ghost job,
	$s(A_i)$ drops below $1 - 4\eps$. So the total size of ghost jobs is
	at least $\eps$. Since we are removing all these jobs from bin $i$,
	$\Phi_i$ decreases by at least $d \eps$, where $d$ is the density of
	the least density ghost job in $i$. As in the case of insert,
	Lemma~\ref{lem:overflow} shows that whenever we make a function call
	\textsc{borrow}($i',v')$, $\Phi_{i'}$ becomes 0 when this process
	ends. So, the potential $\Phi$ decreases by at least $4d/\eps$. Now,
	we count the total movement cost.  We transfer items of size at most 1
	between two bins, and so, we just need to count how many bins are
	affected during this process (note that if we make several calls to
	\textsc{borrow} with the same bin $i'$, the the {\em total} size of
	items transferred to $i'$ is at most 1). Let $U$ be the bucket
	containing $i$.  If we do not call {\sc SplitBucket($U$)}, then we
	affect at most $3/\eps$ bins. If we call {\sc SplitBucket($U$)}, then
	it must be the case that $U$ had only $\eps^{-1}$ bins. When we merge $U$
	with the next bucket, and perhaps split this merged bucket, the new
	bucket $U$ has at least $2/\eps$ bins, and so, we will not call {\sc
		SplitBucket($U$)} again. Thus, we will touch at most $3/\eps$
	buckets in any case. It follows that the total movement cost is at
	most $3d/\eps$ (all bins following $i$ store items of density at most
	$d$). Therefore, the amortized movement cost is negative. This proves
	the desired result.
\end{proof}

\paragraph{Dealing with Small Items: Summary.} Combining
Corollary~\ref{cor:smallish-items} and Lemma~\ref{lem:small-ub-gen}, we
find that we can pack small items into bins which, ignoring some $O(1)$
many bins are $1-\eps$ full on average. Formally, we have the following.

\begin{lem}\label{lem:sh-small}
	For all $\eps\leq \frac{1}{6}$ there exists a fully-dynamic bin
	packing algorithm with $O(\frac{1}{\eps^2})$ amortized recourse for
	instances where all items have size at most $\eps$ which packs items
	into bins which, ignoring some $O(1)$ bins, are at least $1-\eps$ full
	on average.
\end{lem}

\paragraph{Worst case bounds.}\label{sec:general-wc}
Note that the above algorithm yields worst case bounds for several natural scenarios, given in the following corollaries.

\begin{cor}\label{cor:wc-delta}
	For all $\eps\leq \frac{1}{6}$ there exists a fully-dynamic bin
	packing algorithm with $O(\frac{1}{\delta\cdot \eps^2})$ worst case recourse for
	instances where all items have size in the range $[\delta,\eps]$ which packs items
	into bins which, ignoring some $O(1)$ bins, are at least $1-\eps$ full
	on average.
\end{cor}
\begin{proof}[Proof (Sketch)]
	The algorithm is precisely the algorithm of Lemma \ref{lem:sh-small}, only now in our analysis, as each item has size $\delta$, any single removal can incur movement of at most $\epsilon^2$ size, by our algorithm's definition. The worst case migration factor follows.
\end{proof}

Finally, if we group the items of size less than $(1/n,\epsilon]$ into ranges of size $(2^i,2{i+1}]$ (guessing and doubling $n$ as necessary) requires only $O(\log n)$ additive bins (one per size range), while allowing worst case recourse bounds by the previous corollary.

\begin{cor}\label{cor:wc-logn}
	For all $\eps\leq \frac{1}{6}$ there exists a fully-dynamic bin
	packing algorithm with $O(\frac{1}{\eps^2})$ worst case recourse for
	instances where all items have size at most $\eps$ which packs items
	into bins which, ignoring some $O(\log n)$ bins, are at least $1-\eps$ full
	on average.
\end{cor}

\subsubsection{SH Algorithms: Conclusion}
\label{sec:sh-summary}

\S\ref{sec:gen-small-items} and \S\ref{sec:sh-large-items} show how
to maintain all invariants of SH algorithms, using $O(1)$ amortized
recourse, provided $\eps^{-1}=O(1)$ and $K=O(1)$, and
$\beta_i,\gamma_i=O(1)$ for all $i\in [K]$. As the SH algorithm implied
by Lemma~\ref{lem:good-sh} satisfies these conditions, we obtain this
section's main positive result.

\UBGeneral*

%%% Local Variables: 
%%% mode: latex
%%% TeX-master: "main"
%%% End: 

\section{Omitted Proofs of Section \ref{sec:migration} (Size Movement Costs)}\label{sec:migration-appendix}

Here we provide proofs for our matching amortized size cost upper and lower bounds.

\subsection{Amortized Migration Factor Upper Bound}\label{sec:migration-ub}

We start with a description and analysis of the trivial optimal algorithm for size costs.

\trivialAmortizedMigration*
\begin{proof}
  We divide the input into {\em epochs}. The first epoch starts at time
  $0$. For an epoch starting at time $t$, let $V_t$ be the total volume
  of items present at time $t$.  The epoch starting at time $t$ ends
  when the total volume of items inserted or deleted during this epoch
  exceeds $\eps V_t$.  We now explain the bin packing
  algorithm. Whenever an epoch ends (say at time $t$), we use an offline
  A(F)PTAS (e.g., %such as the algorithms of %De La Vega and Leuker~
  \cite{de1981bin} or
% and Karmarkar and Karp 
  \cite{karmarkar1982efficient}) to efficiently compute a solution using
  at most $(1+\eps)\cdot OPT(\mathcal{I}_t) + O(\eps^{-2})$
  bins. (Recall that $\mathcal{I}_t$ denotes the input at time $t$.) We
  pack items arriving during an epoch in new bins, using the first-fit
  algorithm to pack them. If an item gets deleted during an epoch, we
  {\em pretend} that it is still in the system and continue to pack
  it. When an epoch ends, we remove all the items which were deleted
  during this epoch, and recompute a solution using an off-line A(F)PTAS
  algorithm as indicated above.

  To bound the recourse cost, observe that if the starting volume of
  items in an epoch starting at time $t$ is $V_t$, the volume at the end
  of this epoch is at most $V_t + A_t$, where $A_t$ is the volume of
  items that arrived and departed during this epoch. As $A_t > \eps
  V_t$, the cost of reassigning these items of volume at most $V_t+A_t$
  can be charged to $A_t$. Specifically, the amortized migration cost of
  the epoch starting at time $t$ is at most $(V_t+A_t)/A_t <
  V_t/\epsilon V_t + 1 = O(\eps^{-1})$. Consequently, the overall
  amortized migration cost is $O(\eps^{-1})$.

  To bound the competitive ratio, we use the following easy 
  fact: the optimal number of bins to pack a bin packing instance of
  total volume $V$ lies between $V$ and $2V$. (E.g., the first-fit
  algorithm achieves this upper bound.)
  % proof: conisder the first $m-1$ bins used by FF (out of the $m$ bins
  % overall). These bins are all at least half full. Consequently, V >
  % (m-1)/2. But then OPT >= V > (m-1)/2 implies 2OPT > m-1, or in other
  % words, 2OPT >= m.
  Consider an epoch starting at time $t$. Let $V_t$ denote the volume of
  the input $\mathcal{I}_t$ at time $t$.
  % \underline{If $V_t \leq 1/\eps^2$, then we know that algorithm uses
  % $O(\eps^{-2})$ bins at time $t$. During this epoch, we will see
  % total volume inserted and/or deleted of at most $1/\eps$, and so,
  % the algorithm will use at most $O(\eps^{-1})$ additional
  % bins. Therefore, the algorithm will use $O(1/\eps^{-2})$ bins at any
  % time $t'$ during this epoch, which is at most $(1+\eps)\cdot
  % OPT(\mathcal{I}_{t'})+O(\eps^{-2})$. Therefore, let us assume that
  % $V_t \geq 1/\eps^2$.}
  The algorithm uses at most $(1+\eps)\cdot OPT(\mathcal{I}_t) +
  O(\eps^{-2})$ bins at time $t$. Consider an arbitrary time $t'$ during
  this epoch. As a packing of the instance $\mathcal{I}_{t'}$ can be
  extended to a packing of $\mathcal{I}_t$ by packing
  $\mathcal{I}_{t}\setminus \mathcal{I}_{t'}$ with the first fit
  algorithm, using a further $2\eps V_t$ bins, we have
  $OPT(\mathcal{I}_{t}) \leq OPT(\mathcal{I}_{t'}) + 2\eps V_t$.
  % We first claim that $OPT(\mathcal{I}_{t'}) \geq OPT(\mathcal{I}_t) -
  % 2\eps V_t$. Suppose not. Then consider $OPT(\mathcal{I}_{t'}) $ bins
  % for packing all items at time $t'$. Between $t$ and $t'$, we would
  % have lost items of total volume at most $\eps V_t$, and so, they can
  % be packed using $2 \eps V_t$ bins. Hence, we can pack all items in
  % $\mathcal{I}_t$ using at most $OPT(\mathcal{I}_{t'}) + 2 \eps V_t <
  % OPT(\mathcal{I}_t)$ bins, which is a contradiction.
  On the other hand, our algorithm uses at most an additional $2 \eps
  V_t$ bins at time $t'$ compared to the beginning of the
  epoch. Therefore, the number of bins used at time $t'$ is at most
  \begin{eqnarray*}
    (1+\eps)\cdot OPT(\mathcal{I}_t) + 2 \eps V_t + O(\eps^{-2})  & \leq
    & (1+\eps) \cdot (OPT(\mathcal{I}_{t'}) + 2 \eps V_t) + 2 \eps V_t +
    O(\eps^{-2}) \\ 
    & \leq & (1+\eps)\cdot  OPT(\mathcal{I}_{t'}) + 5 \eps V_t + O(\eps^{-2}).
  \end{eqnarray*}
  Now observe that $OPT(\mathcal{I}_{t'}) \geq V_t - \eps V_t =
  (1-\eps)\cdot V_t$ because the total volume of jobs at time $t'$ is at
  least $V_t - \eps V_t$, and so $OPT(\mathcal{I}_{t'}) \geq V_t/2$.
  Therefore, $5 \eps V_t \leq 10 \eps\, OPT(\mathcal{I}_{t'})$. It follows
  that the number of bins used by the algorithm at time $t'$ is at most
  $(1 + O(\eps))\cdot OPT(\mathcal{I}_{t'}) + O(\eps^{-2})$.  
\end{proof}

\subsection{The Matching Lower Bound}

Here we prove our matching lower bound for amortized recourse in the size costs setting.

We recall that our proof relies on Sylvester's sequence. For ease of reference we restate the salient properties of this sequence which we rely on in our proofs.
The Sylvester sequence is given by the recurrence relation $k_1=2$ and $k_{i+1}=\big(\prod_{j\leq
	i} k_j\big)+1$, or equivalently $k_{i+1}=k_i\cdot(k_i-1)+1$ for $i\geq 0$. 
The first few terms of this sequence are $2,3,7,43,1807,\dots$ 
In what follows we 
let $c$ be a large positive integer to be specified later, and
$\eps:=1/\prod_{\ell=1}^c k_\ell$.
For notational simplicity, since our products and sums will always be taken over the range $[c]$ or
$[c]\setminus\{i\}$ for some $i\in [c]$, we will write $\sum_{\ell}
k_\ell$ and $\prod_\ell 1/k_\ell$, $\sum_{\ell\neq i}1/k_\ell$,
$\prod_{\ell \neq i}1/k_\ell$, etc., taking the range of $\ell$ to be
self-evident from context.
In our proof of Theorem \ref{thm:lb-amortized-migration} and the lemmas building up to it, we shall make use of the
following properties of the Sylvester sequence, its reciprocals and this $\epsilon=1/\prod_{\ell} k_\ell$. 
\begin{enumerate}[(P1)]
	\item\label{prop:sum-app} 
	$\frac{1}{k_1} +
	\frac{1}{k_2} + \ldots + \frac{1}{k_c} = 1- \frac{1}{\prod_{\ell}
		k_\ell} {= 1 - \eps}. $
	\item\label{prop:coprime-app} If $i \neq j$, then $k_i$ and $k_j$ are relatively prime. %In particular, $1/k_i$ is not an integer product of $1/k_j$.
	\item\label{prop:prod-eps-app} For all $i\in[c]$, the value $1/k_i=\prod_{\ell \neq i} k_\ell/\prod_\ell k_\ell$ is an integer product of $\eps = 1/\prod_\ell k_\ell$.
	\item\label{prop:divisibility-app} If $i\neq j\in [c]$, then $1/k_i = \prod_{\ell\neq i} k_\ell / \prod_\ell k_\ell$ is an integer product of $k_j\cdot \epsilon = k_j / \prod_\ell k_\ell$.
\end{enumerate}

Properties \ref{prop:divisibility-app} and \ref{prop:prod-eps-app} are immediate, while property \ref{prop:coprime-app} follows directly from the recursive definition of the sequence's terms, $k_{i+1}=\big(\prod_{j\leq i} k_j\big)+1$. Finally, the same definition readily implies $\frac{1}{k_i} = \frac{1}{k_i - 1} - \frac{1}{k_{i+1}-1}$, from which property \ref{prop:sum-app} follows by induction.

%$$\sum_{j=1}^i \frac{1}{k_j} = \sum_{j=1}^i \left(\frac{1}{k_i - 1} - \frac{1}{k_{i+1}-1}\right) = \frac{1}{k_1 - 1} - \frac{1}{k_{i+1}-1} = 1-\frac{1}{\prod_{j=1}^i k_j}.$$

The instances we consider for our lower bound will be comprised of items with sizes given by entries of the following vector $\vec{s} \in [0,1]^{c+1}$, defined as follows:
$$
s_i := \begin{cases}
\frac{1}{k_i}\cdot (1-\frac{\eps}{2}) & i\in [c] \\
\eps\cdot(\frac{3}{2}-\frac{\eps}{2}) & i = c+1. \\
\end{cases}
$$

Let $\vec{1}$ be the all-ones vector, and $\vec{e}_i$ the $i$-th standard basis vector. In what follows we will denote by $\vec{\chi}\in \mathbb{N}^{c+1}$ a characteristic vector of a feasibly-packed bin; i.e., $\chi_i$ is the number of items of size $s_i$ in the bin, and as the bin is not over-flowing, we have $\vec{\chi}\cdot \vec{s} = \sum_i \chi_i \cdot s_i \leq 1$. The following fact is easy to check from the definition of $\vec{s}$. 

\begin{fact}\label{lem:s-packable}
	$\vec{1}\cdot \vec{s} = 1$.
\end{fact}
\begin{proof}
	By property~\ref{prop:sum-app}, we have
	$\sum_{i=1}^c s_i = (1-\frac{\eps}{2})\cdot  (1 - \eps) = 1 - s_{c+1}.$
\end{proof}

We now show that the item sizes given by the entries of $\vec{s}$ are such that in many situations (in particular, situations where no $s_{c+1}$-sized items are present), any feasible packing in a bin will leave some gap.
\begin{lem}
	\label{lem:characteristic}
	Let $\chi$ denote a feasible packing of items into a bin; i.e., $\vec{\chi}\in\mathbb{N}^{c+1}$, $\vec{\chi}\cdot \vec{s} \leq 1$. Then, 
	\begin{enumerate}[(a)]
		\item\label{part:small} If $\chi_{c+1} =  1$ and $\vec{\chi} \neq \vec{1}$, then $\vec{\chi}\cdot \vec{s} \leq 1 - \eps/2$.
		\item\label{part:no-small} If $\chi_{c+1} =  0$, then $\vec{\chi}\cdot \vec{s} \leq 1 - \eps/2$. 
		\item\label{part:no-small-not-unit-vector} If $\chi_{c+1} =  0$ and $\vec{\chi} \neq k_i \vec{e}_i$, then $\vec{\chi}\cdot \vec{s} \leq 1 - \eps$.
	\end{enumerate}
\end{lem}

\begin{proof}
	For ease of notation, we define a vector $\vec{t}$, where $t_i = 1/k_i$ for $i=1, \ldots, c$ and $t_{c+1} = 0$. So, $s_i = (1 - \eps/2)\cdot t_i$ for $i=1, \ldots, c$. 
	\begin{obs}
		\label{obs:vectort}
		The dot product $\vec{\chi}\cdot \vec{t}$ is at most 1. 
	\end{obs}
	\begin{proof}
		By property~\ref{prop:prod-eps-app}, each $t_i$ is an integral multiple of $\eps$. It follows that $\vec{\chi}\cdot \vec{t}$ is also an 
		integral multiple of
		$\eps$. Clearly, 1 is an integral multiple of $\eps$. Therefore, if $\vec{\chi}\cdot \vec{t}>1$, then $\vec{\chi}\cdot \vec{t}$ is at least $1 + \eps$. 
		But then, $\vec{\chi}\cdot \vec{s} \geq (1-\eps/2) \cdot \vec{\chi}\cdot \vec{t} \geq (1-\eps/2)\cdot (1+\eps) > 1$, a contradiction. 
	\end{proof}
	
	Part~\ref{part:small}: First, observe that $\vec{\chi}\cdot \vec{t}
	\neq 1$. Indeed, it is at most 1, by \Cref{obs:vectort}).
	Moreover, if $\vec{\chi}\cdot \vec{t} = 1$, the fact that $\chi_{c+1} =
	1$ would imply $\vec{\chi}\cdot \vec{s} = \vec{\chi}\cdot
	(1-\eps/2)\cdot \vec{t} + s_{c+1} = (1-\eps/2) + s_{c+1} > 1$, a
	contradiction. Also, there must be some index $i\in [c]$ such that
	$\chi_i = 0$. Indeed, otherwise the fact that $\vec{\chi} \neq \vec{1}$
	and Fact~\ref{lem:s-packable} imply that $\vec{\chi}\cdot \vec{s} >
	\vec{1}\cdot \vec{s} = 1. $ So let $i$ be an index such that $\chi_i =
	0$. By property~\ref{prop:divisibility-app}, for all $j \neq i$, the value
	$t_j=1/k_j$ is an integral multiple of $k_i \eps$. Therefore,
	$\vec{\chi}\cdot \vec{t} $ is a multiple of $k_i \eps$. Since
	$\vec{\chi}\cdot \vec{t} < 1$ and 1 is an integral multiple of $k_i
	\eps$, it follows that $\vec{\chi}\cdot \vec{t} \leq 1 - k_i \eps \leq
	1 - 2 \eps$.  Therefore, $\vec{\chi}\cdot \vec{s} \leq (1-\eps/2)\cdot (1 -
	2\eps ) + s_{c+1} = 1 - \eps + \eps^2/2 \leq 1-\eps/2$.
	
	Part~\ref{part:no-small} follows directly from Observation
	\ref{obs:vectort} above. Indeed, since $\chi_{c+1}=0$, we have that
	$\vec{\chi}\cdot \vec{s} = (1-\eps/2) \vec{\chi}\cdot \vec{t} \leq
	(1-\eps/2).$ 
	
	Part~\ref{part:no-small-not-unit-vector}: First, observe that $\chi_i <
	k_i$ for $i=1, \ldots, c$. Indeed, if $\chi = k_i \vec{e}_i + \chi'$
	for some index $i$ and some non-zero non-negative vector $\chi'$, then
	$\vec{\chi}\cdot \vec{t} = 1 + \vec{\chi'}\cdot \vec{t} > 1$, which
	contradicts \Cref{obs:vectort}. Now let $i$ be an index such
	that $\chi_i \geq 1$. Since $\chi_i < k_i$, property~\ref{prop:coprime-app}
	implies that $\chi_i \cdot t_i = \chi_i \cdot \eps \cdot \prod_{\ell \neq i}
	k_\ell$ is not an integral multiple of $k_i \eps$.  But, by
	property~\ref{prop:divisibility-app} for all $j \neq i$ we have that
	$t_j=1/k_j$ is an integral multiple of $k_i \eps$. Thus,
	$\vec{\chi}\cdot \vec{t} (\leq 1)$ is not an integral multiple of $k_i
	\eps= k_i/\prod_{\ell} k_\ell = 1/\prod_{\ell\neq i}
	k_\ell$. Consequently, $\vec{\chi}\cdot \vec{t} < 1$. But, by
	property~\ref{prop:prod-eps-app}, we have that $\vec{\chi}\cdot \vec{t}$ is
	an integral multiple of $\eps$, from which it follows that
	$\vec{\chi}\cdot \vec{t} \leq 1 - \eps$. Therefore, as $\chi_{c+1} =
	0$, we have $\vec{\chi}\cdot \vec{s} = \vec{\chi}\cdot
	(1-\epsilon/2)\cdot \vec{t} \leq (1-\epsilon/2)\cdot (1 - \eps) \leq 1-\eps$.
\end{proof}

Equipped with the above lemma, we can show two instances with similar overall weight, $\cI$ to $\cI'$ for which near-optimal solutions different significantly. Specifically, we define the input instances $\mathcal{I}$ and $\mathcal{I}'$, as follows. For some large $N$ a product of $\prod_\ell k_\ell$, Instance $\mathcal{I}$ consists of $N$ items of sizes $s_i$ for all $i\in [c+1]$. Instance $\mathcal{I}'$ consists of $N$ items of all sizes but $s_{c+1}$. It is easy to check that
$c = \Theta(\log \log 1/\eps)$, and so the total number of items is
$n = \Theta(N\cdot \log \log (1/\eps)).$ Therefore the additive
$o(n)$ term, denoted by $f(n)$, satisfies
%	\begin{equation}\label{eqn:N-large}
$$	f(n) < \epsilon \cdot n/(42\cdot \Theta(\log \log (1/\eps)))=\epsilon N/42.$$
%	\end{equation}
We now proceed to proving that approximately-optimal packings of the above two instances differ significantly.

\migrationKeyLemma*
\begin{proof}
	We begin by proving our claimed bound on $\mathcal{A}$'s packing of $\cI$.
	\Cref{lem:s-packable} shows that the optimal number of bins is
	$N$. Therefore, $\mathcal{A}$ is allowed to use at most $(1 +
	\eps/7) N + \eps N/42 = (1+\eps/6)N$ bins. Now suppose there more
	than $N/3$ bins for which the characteristic vector is not
	$\vec{1}$. \Cref{lem:characteristic}\ref{part:small} shows
	that each such bin must leave out at least $\eps/2$
	space. Therefore, the total unused space in these bins is greater
	than $\eps N/6$, which implies that the algorithm must use at least
	$N + \eps N/6$ bins, a contradiction.
	
	We now proceed to prove our claimed bound on $\mathcal{A}$'s packing of $\cI'$. Let us first find the optimal value $OPT(\mathcal{I}')$. By
	property~\ref{prop:sum-app}, the total volume of all the items is equal
	to $N(1-\eps)(1-\eps/2)$. By
	\Cref{lem:characteristic}\ref{part:no-small}, any bin can be
	packed to an extent of at most $(1-\eps/2)$. Therefore the optimal
	number of bins is at least $N(1-\eps)$. Furthermore, we can achieve
	this bound by packing $N$ items of size $s_i$ in $N/k_i$ bins for
	each index $i$.  Therefore, the algorithm is allowed to use at most
	$(1+\eps/7)(1-\eps)N + \eps N/42 \leq (1 - \frac{35 \eps}{42}) N$
	bins when packing $\mathcal{I'}$.
	
	Suppose there are at least $N/2$ bins each of which is assigned
	items of at least two different sizes by the algorithm. By
	\Cref{lem:characteristic}\ref{part:no-small-not-unit-vector},
	the algorithm will leave at least $\eps$ unused space in such bins;
	moreover, by \Cref{lem:characteristic}\ref{part:no-small}, the
	algorithm will leave at least $\eps/2$ unused space in every
	bin. Thus, the total unused space in the bins is at least $\eps N/2
	+ \eps N/4 = 3\eps N/4$. Since the total volume of the items is
	equal to $N(1-\eps)(1-\eps/2)$, we see that the total number of
	bins used by the algorithm is at least $N(1-\eps)(1-\eps/2) + 3
	\eps N/4 \geq N(1-3\eps/4) > (1 - \frac{35 \eps}{42}) N$, a
	contradiction.
\end{proof}

We are now ready to prove this section's
main result.
\LBAmortizedMigration*

\begin{proof}
	We will show that any algorithm $\mathcal{A}$ using at most
	$(1+\eps/7)\cdot OPT(\mathcal{I}_t)+o(n)$ bins at time $t$ uses at
	least $1/160\eps$ amortized recourse, for arbitrarily large optima
	and instance sizes, implying our theorem. We consider the two instance $\cI$ and $\cI'$ defined above.
	
	Suppose we first provide instance $\mathcal{I}$ to
	$\mathcal{A}$, and then remove all items of size $s_{c+1}$ to get the
	instance $\mathcal{I}'$.  Let $B_1$ and $B_2$ be the sets of bins
	guaranteed by \Cref{lem:migration-instances} when we
	had the instances $\mathcal{I}$ and $\mathcal{I}'$,
	respectively. Notice that as algorithm $\mathcal{A}$ uses at most $(1
	- \frac{35 \eps}{42}) N\leq N$ bins while packing $\mathcal{I}'$, and
	$|B_1|\geq 2N/3$, $|B_2|\geq N/2$, we have that either $(i)$~$|B_1
	\cap B_2| \geq N/12$, or $(ii)$ at least $N/12$ of the bins of $B_1$
	are closed in $\mathcal{A}$'s packing of $\mathcal{I}'$.
	
	Consider any bin which lies in both $B_1 $ and $B_2$. In instance
	$\mathcal{I}$, algorithm $\mathcal{A}$ had assigned one item of each
	size to this bin, whereas in instance $\mathcal{I}'$ the algorithm
	assigns this bin items of only one size.  Therefore, the total size
	of items which need to go out of (or into) this bin when we
	transition from $\mathcal{I}$ to $\mathcal{I}'$ is at least
	$1/2$. Therefore, the total volume of items moved during this
	transition is at least $N/48$. Similarly, if $N/12$ of the bins of
	$B_1$ are closed in $\mathcal{A}$'s packing of $\mathcal{I}'$, at
	least $1-s_{c+1}\geq 1/2$ volume must leave each of these bins, and
	so the total volume of items moved during this transition is at least
	$N/24>N/48$.
	
	Repeatedly switching between $\mathcal{I}$ and $\mathcal{I}'$ by
	adding and removing the $N$ items of size $s_{c+1}$ a total of $T$
	times (for sufficiently large $T$), we find that the amortized
	recourse is at least
%	\begin{align*}
$$
	\frac{T\cdot N/48}{N + T\cdot 3\eps\cdot N} %\geq \frac{T\cdot \frac{11}{216}\cdot N}{\frac{220}{216}\cdot T\cdot 3\eps\cdot N} =
	\geq \frac{1}{160\eps}.
%	\qedhere
%	\end{align*}
$$
\end{proof}

%%%%%%%%%%%%%%%%%%%%%%%%%%%%%%%%%%%%%%%%%%%%%%%%%%%%%%%%%%%%
%%%%%%%%%%%%%%%%%%%%  BIBLIOGRAPHY  %%%%%%%%%%%%%%%%%%%%%%%%
%%%%%%%%%%%%%%%%%%%%%%%%%%%%%%%%%%%%%%%%%%%%%%%%%%%%%%%%%%%%
\bibliographystyle{acmsmall}
\bibliography{bin-packing}

\end{document}